\documentclass[11pt]{article}
\usepackage[pdftex]{graphicx}
\usepackage[T1]{fontenc}
\usepackage{lmodern}

\usepackage{fullpage}

\usepackage{amsmath,amsthm,amssymb}
\usepackage{verbatim}
\usepackage{mathtools}
\usepackage{algorithm}
\usepackage{algorithmicx}
\usepackage{algpseudocode}
\usepackage{subfig}
\usepackage[toc,page]{appendix}
\usepackage[usenames,dvipsnames]{color}

\usepackage{caption}

\usepackage{thmtools}
\usepackage{thm-restate}

\usepackage{hyperref}

\usepackage{cleveref}

\crefname{algocf}{alg.}{algs.}
\Crefname{algocf}{Algorithm}{Algorithms}

\declaretheorem[name=Theorem,numberwithin=section]{thm}
\declaretheorem[name=Theorem,numberlike=thm]{theorem}
\declaretheorem[name=Lemma,numberlike=thm]{lemma}

\declaretheorem[name=Corollary,numberlike=thm]{cor}
\declaretheorem[name=Definition,numberlike=thm,style=definition]{defn}

\DeclarePairedDelimiter{\ceil}{\lceil}{\rceil}
\DeclarePairedDelimiter{\floor}{\lfloor}{\rfloor}
 
\DeclareMathOperator*{\argmax}{arg\,max}
\DeclareMathOperator*{\argmin}{arg\,min}

\newcommand{\bethe}{\mathrm{bethe}}
\newcommand{\sink}{\mathrm{sinkhorn}}
\newcommand{\ssink}{\mathrm{scaledsinkhorn}}
\newcommand{\betheq}{\mathrm{F}}
\newcommand{\frst}{\mathrm{U}}
\newcommand{\scnd}{\mathrm{V}}

\newcommand{\prob}[1]{\mathrm{Pr}\left(#1 \right)}
\newcommand{\permD}{S_{\bX}}
\newcommand{\my}{\textbf{f}_{y}}
\newcommand{\mo}{\textbf{m}_{1}}

\newcommand{\mk}{\textbf{m}_{k}}
\newcommand{\mj}{\textbf{m}_{j}}

\newcommand{\co}{\textbf{c}_{1}}
\newcommand{\ct}{\textbf{c}_{2}}
\newcommand{\ck}{\textbf{c}_{k}}
\newcommand{\cj}{\textbf{c}_{j}}

\newcommand{\perm}{\mathrm{perm}}
\newcommand{\kzo}{\{0,1\}^{\bX \times k}}
\newcommand{\ka}{\textbf{K}_{\ma}}
\newcommand{\kr}{\textbf{K}_{r}}
\newcommand{\mx}{\textbf{X}}
\newcommand{\kds}{\R_{\geq 0}^{\bX \times \bX}}
\newcommand{\fna}{\textbf{h}_{\ma}}

\newcommand{\expt}[2]{\mathbb{E}_{#1}\left[#2\right]}
\newcommand{\KL}[1]{\mathrm{KL}\left(#1\right)}
\newcommand{\Pxj}{\bP_{x,j}}
\newcommand{\Axy}{\ma_{x,y}}
\newcommand{\Xxj}{\mx_{x,j}}
\newcommand{\mad}{{\hat{\ma}}}
\newcommand{\madxj}{\mad_{x,j}}
\newcommand{\may}{\ma_{:y}}
\newcommand{\Ns}{N}

\newcommand{\boo}{b_1}
\newcommand{\btt}{b_2}

\newcommand{\ztbtt}{[0,\btt]}
\newcommand{\otboo}{[1,\boo]}

\newcommand{\bttpo}{(\btt+1)}

\newcommand{\F}{\phi}

\newcommand{\1}{\overrightarrow{1}}
\newcommand{\R}{\mathbb{R}}

\newcommand{\bbP}{\mathbb{P}}

\newcommand{\Z}{\mathbb{Z}_{+}}

\newcommand{\simplex}{\Delta^{\bX}}
\newcommand{\psimplex}{\Delta_{pseudo}^{\bX}}
\newcommand{\dsimplex}{\Delta_{\bR}^{\bX}}

\newcommand{\pp}{\pvec}
\newcommand{\nn}{\mvec}
\newcommand{\probpml}{\bbP}

\newcommand{\bg}{\textbf{g}}

\newcommand{\bL}{\textbf{L}}

\newcommand{\bp}{\textbf{p}}
\newcommand{\bq}{\textbf{q}}
\newcommand{\bff}{\textbf{f}}

\newcommand{\bX}{\mathcal{D}}

\newcommand{\bH}{\textbf{H}}
\newcommand{\bP}{\textbf{P}}
\newcommand{\bQ}{\textbf{Q}}

\newcommand{\bZ}{\textbf{Z}}
\newcommand{\bR}{\textbf{R}}

\newcommand{\bS}{\textbf{S}}

\newcommand{\expo}[1]{\exp \left(#1 \right)}
\newcommand{\exps}[1]{\exp (#1 )}

\newcommand{\fng}{\textbf{h}}

\newcommand{\otilde}{\widetilde{O}}

\newcommand{\cphi}{C_{\phi}}

\newcommand{\level}{\ell}

\newcommand{\levelq}{\ell^{\bq}}

\newcommand{\ma}{\textbf{A}}

\newcommand{\mzero}{\textbf{0}}

\newcommand{\pvec}{\zeta}
\newcommand{\vvec}{\textbf{v}}

\newcommand{\mvec}{\overrightarrow{\mathrm{m}}}
\newcommand{\epso}{\epsilon_1}

\newcommand{\eqdef}{\stackrel{\mathrm{def}}{=}}
\newcommand{\defeq}{\eqdef}

\newcommand{\onevec}{\overrightarrow{\mathrm{1}}}

\newcommand{\fnff}{\textbf{f}}

\newcommand{\setd}{\textbf{M}}
\newcommand{\eled}{\textbf{m}}
\newcommand{\dstoc}{\mathbf{K}_{rc}}
\newcommand{\mapphi}{\ma^{\bp,\phi}}

\newcommand{\mE}{\textbf{E}}
\newcommand{\mone}{\textbf{1}}

\newcommand{\cjp}{\textbf{c}_{j'}}
\newcommand{\maxt}{\ma_{x.}}
\newcommand{\mayt}{\ma_{y.}}
\newcommand{\fnp}{\textbf{f}}
\newcommand{\Pyj}{\bP_{y,j}}	
\newcommand{\Xzjp}{\mx_{z,j'}}
\newcommand{\madzjp}{\hat{\ma}_{z,j'}}
\newcommand{\Xyj}{\mx_{y,j}}
\newcommand{\mY}{\textbf{Y}}
\newcommand{\Yzjp}{\mY_{z,j'}}

\newcommand{\disc}{\mathrm{disc}}
\newcommand{\mz}{\textbf{m}_{0}}

\newcommand{\bZS}{\textbf{Z}^{\phi,\bq_{\bS}}_{\bR}}
\newcommand{\bZfrac}{\textbf{Z}^{\phi,frac}_{\bR}}
\newcommand{\bqS}{\bq_{\bS}}

\newcommand{\ri}{\textbf{r}_{i}}
\newcommand{\zrm}{\R_{\geq 0}^{\ell \times (k+1)}}
\newcommand{\zrc}{\textbf{Z}_{rc}}
\newcommand{\za}{\textbf{Z}_{\bR}^{{\bq},\phi}}

\newcommand{\lpi}{\level^{\bp}_{i}}
\newcommand{\Sij}{\bS_{i,j}}
\newcommand{\pml}{\bp_{\mathrm{pml}}}
\newcommand{\maqphi}{\ma^{\bq,\phi}}
\newcommand{\bigO}[1]{O\left(#1 \right)}
\newcommand{\bSp}{\bS'}

\newcommand{\dpml}{\bq_{\mathrm{pml}}}

\newcommand{\bZqext}{\textbf{Z}^{\qext,\phi}_{\bRext}}
\newcommand{\logparam}{\Delta}
\newcommand{\bRext}{\bR^{\mathrm{ext}}}
\newcommand{\bZext}{\textbf{Z}^{\phi}_{\bRext}}
\newcommand{\bSext}{\textbf{S}^{\mathrm{ext}}}
\newcommand{\bqSext}{\bq_{\bSext}}
\newcommand{\fnggg}{\textbf{g}}
\newcommand{\mmjj}{m_{j}}

\newcommand{\newk}{\bttpo}
\newcommand{\algcreate}{\mathrm{CreateNewProbabilityValues}}
\newcommand{\bRone}{\bR^{(1)}}
\newcommand{\bRtwo}{\bR^{(2)}}
\newcommand{\bZone}{\textbf{Z}^{\phi,frac}_{\bRone}}
\newcommand{\bZtwo}{\textbf{Z}^{\phi,frac}_{\bRtwo}}
\newcommand{\bSone}{\textbf{S}^{(1)}}
\newcommand{\bStwo}{\textbf{S}^{(2)}}
\newcommand{\bHone}{\bH^{(1)}}
\newcommand{\bLone}{\bL^{(1)}}
\newcommand{\pvecone}{\pvec^{(1)}}
\newcommand{\pvectwo}{\pvec^{(2)}}
\newcommand{\bAone}{\textbf{A}^{(1)}}

\newcommand{\tsqn}{\gamma}
\newcommand{\tsqni}{\frac{1}{\gamma}}
\newcommand{\pvecext}{\pvec^{\mathrm{ext}}}
\newcommand{\bpapprox}{\bp_{\mathrm{approx}}}
\newcommand{\qext}{\bq_{\bSext}}

\newcommand{\si}{s_{i}}

\newcommand{\RomanNumeralCaps}[1]
{\MakeUppercase{\romannumeral #1}}

\usepackage[utf8]{inputenc} 
\usepackage[T1]{fontenc}    
\usepackage{hyperref}       
\usepackage{url}            
\usepackage{booktabs}       
\usepackage{amsfonts}       
\usepackage{nicefrac}       
\usepackage{microtype}      

\title{The Bethe and Sinkhorn Permanents of Low Rank Matrices\\
	and Implications for Profile Maximum Likelihood}

\author{
	 Nima Anari\\
	Stanford University\\
	\texttt{anari@cs.stanford.edu} \\
	\and
  Moses Charikar\\
  Stanford University\\
  \texttt{moses@cs.stanford.edu}\thanks{Moses Charikar was supported by a Simons Investigator Award, a Google Faculty Research Award and an Amazon Research Award.} \\
   \and
Kirankumar Shiragur\\
  Stanford University\\
  \texttt{shiragur@stanford.edu}\thanks{Kirankumar Shiragur was supported by Stanford Data Science Scholarship.} \\
     \and
Aaron Sidford\\
  Stanford University\\
  \texttt{sidford@stanford.edu}\thanks{Aaron Sidford was supported by NSF CAREER Award CCF-1844855.} \\
}

\begin{document}

\maketitle

\begin{abstract}
In this paper we consider the problem of computing the likelihood of the profile of a discrete distribution, i.e., the probability of observing the multiset of element frequencies, and computing a profile maximum likelihood (PML) distribution, i.e., a distribution with the maximum profile likelihood. For each problem we provide polynomial time algorithms that given $n$ i.i.d.\ samples from a discrete distribution, achieve an approximation factor of $\exp\left(-O(\sqrt{n} \log n) \right)$, improving upon the previous best-known bound achievable in polynomial time of $\exp(-O(n^{2/3} \log n))$ (Charikar, Shiragur and Sidford, 2019). Through the work of Acharya, Das, Orlitsky and Suresh (2016), this implies a polynomial time universal estimator for symmetric properties of discrete distributions in a broader range of error parameter.

We achieve these results by providing new bounds on the quality of approximation of the Bethe and Sinkhorn permanents (Vontobel, 2012 and 2014).
We show that each of these are $\exp(O(k \log(N/k)))$ approximations to the permanent of $N \times N$ matrices with non-negative rank at most $k$, improving upon the previous known bounds of $\exp(O(N))$. 
To obtain our results on PML,
we exploit the fact that the PML objective is proportional to the permanent of a certain Vandermonde matrix with $\sqrt{n}$ distinct columns, i.e. with non-negative rank at most $\sqrt{n}$.
As a by-product of our work we establish a surprising connection between the convex relaxation in prior work (CSS19) and the well-studied Bethe and Sinkhorn approximations.
\end{abstract}

\pagenumbering{gobble}

\newpage

\pagenumbering{arabic}

\section{Introduction}
Symmetric property estimation of distributions\footnote{Throughout this paper, we use the word distribution to refer to discrete distributions.} is an important and well studied problem in statistics and theoretical computer science. Given access to $n$ i.i.d samples from a hidden discrete distribution $\bp$ the goal is to estimate $\fnff (\bp)$, for a symmetric property $\fnff(\cdot)$.  Formally, a property is symmetric if it is invariant to permutating the labels, i.e. it is a function of the multiset of probabilities and does not depend on the symbol labels. There are many well-known well-studied such properties, including support size and coverage, entropy, distance to uniformity, Renyi entropy, and sorted $\ell_{1}$ distance. 
Understanding the computational and sample complexity for estimating these symmetric properties has led to an extensive line of interesting research over the past decade. 

Symmetric property estimation spans applications in many different fields. For instance, entropy estimation has found applications in neuroscience \cite{RWDB99}, physics \cite{VBBVP12} and others \cite{PW96, ASNRMAS01}. Support size and coverage estimation were initially used in estimating ecological diversity \cite{Chao84, Chao92, BF93, CCGLMCL12} and subsequently applied to many different applications \cite{ET76,TE87, Fur05, KLR99, PBGELLSD01, DS13, RCSWTKRWC09, GTPB07, HHRB01}. For applications of other symmetric properties we refer the reader to \cite{HJW17, HJM17, AOST14, RVZ17, ZVVKCSLSDM16, WY16a, RRSS07, WY15, OSW16, VV11a, WY16, JVHW15, JHW16, VV11b}. 

Early work on symmetric property estimation developed estimators tailored to the particular property of interest. Consequently, a fundamental and important open questions was to come up with an estimator that is \emph{universal}, i.e. the same esstimator could be used for all symmetric properties. A natural approach for constructing universal estimators is plug-in approach, where given samples we first compute a distribution independent of the property and later we output the (value of this) property for the computed distribution as our estimate. 

Our approach is based on the observation (see~\cite{ADOS16}) that a sufficient statistic for estimating a symmetric property from a sequence of samples is the profile, i.e. the multiset of frequencies of symbols in the sequence; e.g. the profile of sequence $abbc$ is $\{2,1,1\}$. We provide an efficient universal estimator that is based on 
the plug-in approach applied to
the \emph{profile maximum likelihood (PML)} distribution introduced by Orlitsky et al.~\cite{OSSVZ04}: given a sequence of $n$ samples, PML is the distribution that maximizes the likelihood of the observed profile. The problem of computing the PML distribution has been studied in several papers since, applying heuristic approaches such as Bethe/Sinkhorn approximation \cite{Von12, Von14}, the EM algorithm \cite{OSSVZ04}, a dynamic programming \cite{PJW17} and algebraic methods  \cite{ADMOP10}.

A recent paper of Acharya et al.~\cite{ADOS16} showed that a plug-in estimator using the optimal PML distribution is universal in estimating various symmetric properties of distributions. In fact it suffices to compute a $\beta$-approximate PML distribution (i.e. a distribution that approximates the PML objective to within a factor of $\beta$) for $\beta >\exps{-n^{1-\delta}}$ for constant $\delta > 0$. Previous work of the authors in \cite{CSS19}, gave the first efficient algorithm to compute a $\beta$-approximate PML for some non-trivial $\beta$. In particular, \cite{CSS19} gave a nearly linear running time algorithm to compute an $\exp(-O(n^{2/3} \log n))$-approximate PML distribution. In this work, we give an efficient algorithm to compute an $\exp(-O(\sqrt{n} \log n))$-approximate PML distribution. 

The parameter $\beta$ in $\beta$-approximate PML effects the error parameter regime under which the estimator is sample complexity optimal. Smaller values of $\beta$ yield a universal estimator that is sample optimal over broader parameter regime. For instance, \cite{CSS19} show that $\exps{-O(n^{2/3} \log n)}$-approximate PML\footnote{Throughout this paper, $\otilde(\cdot)$ hides poly $\log n$ terms.} is sample complexity optimal for estimating certain symmetric properties within accuracy for $\epsilon>n^{-0.16666}$. On the other hand \cite{ADOS16} showed that computing an $\exps{-O(\sqrt{n} \log n)}$-approximate PML is sample complexity optimal for $\epsilon>n^{-0.249}$. However note that, using the current analysis techniques \cite{ADOS16} we are unsure on how to exploit the computation of exact PML any better than computing an $\exps{-O(\sqrt{n} \log n)}$-approximate PML and they both are sample complexity optimal over the same error parameter regime.

In our work, we use the Bethe approximation of the permanent or the Bethe permanent (for short), a previously proposed heuristic to compute an approximate PML distribution.
This is based on the Bethe free energy approximation originating in statistical physics and is very closely connected to the belief propagation algorithm~\cite{YFW05,Von13}. The idea of using the Bethe permanent for computing an approximate PML distribution comes from the fact that the likelihood of a profile with respect to a distribution can be written as the permanent of a non-negative Vandermonde matrix (which we call the profile probability matrix). 
For a $N\times N$ non-negative matrix, \cite{GS14} show that the ratio between the permanent and the Bethe permanent of a matrix is upper bounded by $1.9022^{N} \leq 2^{N}$, that was later improved to $\sqrt{2}^{N}$ \cite{AR19}\footnote{Note that previous results on the Bethe permanent do not immediately imply non-trivial results for PML. For consistency with the literature, we use approximation factors $<1$ for PML.}. A natural question is \emph{whether the approximation ratio of the Bethe permanent depends on some other structural parameter better than the input dimension of the matrix?} 
In this work, 
we show that the approximation ratio between the permanent and the Bethe permanent is upper bounded by an exponential in the non-negative rank of the matrix (up to a logarithmic factor). We also give an explicit construction of a matrix to show that our result for this structural parameter is asymptotically tight. As the non-negative rank of any $N\times N$ non-negative matrix is at most $N$,  our analysis implies an upper bound of $c^{N}$ for some constant $c>0$ on the approximation ratio. Therefore our work (asymptotically) generalizes previous results for general non-negative matrices.

To obtain our efficient algorithm, we prove a slightly stronger statement than the bound of the Bethe permanent of a matrix with non-negative rank at most $k$. 
We show that a scaling of a simpler approximation of the permanent known as the Sinkhorn\footnote{Sinkhorn is also called as capacity in the literature.} permanent also approximates the permanent up to exponential in the non-negative rank of the matrix (up to log factors). 
This implies our bound for the Bethe permanent and shows that scaled Sinkhorn is a compelling alternative to Sinkhorn, with a tighter worst-case multiplicative approximation to the permanent.

An immediate application of our work on the Bethe and the scaled Sinkhorn permanent is to approximate PML. Given $n$ samples, the number of distinct columns in the profile probability matrix is always upper bounded by $\sqrt{n}$, i.e. its non-negative rank is at most $\sqrt{n}$. Therefore our analysis of the scaled Sinkhorn permanent immediately implies an $\expo{-O(\sqrt{n} \log n)}$ approximation to the PML objective with respect to a fixed distribution. This result, combined with probability discretization, results in a convex program whose optimal solution is a fractional representation of an approximate PML distribution. We round this fractional solution to output a valid distribution that is an $\expo{-O(\sqrt{n} \log n)}$-approximate PML distribution. Surprisingly the resulting convex program is exactly the same as the one in \cite{CSS19}, where a completely different (combinatorial) technique was used to arrive at the convex program. Our work here provides a better analysis of the convex program in \cite{CSS19} using a more delicate and sophisticated rounding algorithm.

\paragraph{Organization of the paper:} In \Cref{sec:prelims} we present preliminaries. In \Cref{sec:results}, we provide the main results of the paper. In \Cref{sec:bethe}, we analyze the scaled Sinkhorn permanent of structured matrices. In \Cref{sec:sinkcol}, we prove an upper bound for the approximation ratio of the scaled Sinkhorn permanent to the permanent as a function of the number of distinct columns. In \Cref{sec:lowrank}, we prove the generalized result of the scaled Sinkhorn permanent for the low non-negative rank matrices. In \Cref{sec:lbbethesink}, we prove the lower bound for the Bethe and scaled Sinkhorn approximations of the permanent.  In \Cref{sec:pml}, we combine the result for the scaled Sinkhorn permanent with the idea of probability discretization to provide the convex program that returns a fractional representation of an approximate PML distribution. In the same section, we provide the rounding algorithm to return a valid approximate PML distribution.

\subsection{Overview of Techniques}
In \cite{CSS19}, the authors
presented a convex relaxation for the PML objective.
This was obtained by a combinatorial view of the PML
problem.
In a sequence of steps, they discretized the set of probabilities and the frequencies, grouped the terms in the objective into groups and developed a relaxation for the sum of terms in the largest group, giving an $\exp({-{O}(n^{2/3} \log n)})$ approximation. 
In this paper, we exploit the fact that the likelihood of a profile with respect to a distribution is the permanent of a certain non-negative Vandermonde matrix (referred to here as the profile probability matrix with respect to a distribution) and that the PML objective is an optimization problem over such permanents.
We work with the same convex relaxation we derived earlier, but relate it to the well known Bethe and scaled Sinkhorn approximations for the permanent. In fact, Vontobel~\cite{Von12,Von14} proposed the Bethe and Sinkhorn permanents as a heuristic approximation of the PML objective, but bounding the quality of the solution was an open problem~\cite{Von11}.
We show that both the Bethe and scaled Sinkhorn permanents are within a factor 
$\expo{-O(\sqrt{n} \log n)}$ of the PML objective.
Enroute, we show that the approximation ratio of the Bethe and scaled Sinkhorn permanents for any non-negative matrix $A$ are 
upper bounded by the exponential in the non-negative rank of matrix $A$. This is a strengthening of the well known $\expo{O(N)}$ upper bound on the approximation ratio of both the Bethe and scaled Sinkhorn permanents of an $N \times N$ matrix.

In \cite{CSS19}, the fact that the convex problem we obtained was a relaxation of the PML objective followed directly from the combinatorial derivation of our relaxation.
By contrast, our analysis here exploits the non-trivial fact that the Bethe and scaled Sinkhorn approximations are lower bounds for the permanent of a non-negative matrix.
The Bethe and scaled Sinkhorn permanents of the profile probability matrix $A$ with respect to a fixed distribution are optimum solutions to maximization problems over doubly stochastic matrices $Q$ where the objective functions have entropy-like terms involving the entries of $A$ and $Q$.
In order to obtain an upper bound on the Bethe and scaled Sinkhorn approximation as a function of the non-negative rank, we show the existence of a doubly stochastic matrix $Q$ as a witness such that the objective of the Bethe and scaled Sinkhorn w.r.t.~$Q$ upper bounds the permanent of $A$ within the desired factor.

We first work with a simpler setting of matrices $A$ with at most $k$ distinct columns.\footnote{In the final preparation of this paper for posting an anonymous reviewer showed that a simpler proof for the distinct column case can be derived using Corollary 3.4.5 of Barvinok's book \cite{Bar17}. The proof of the Corollary 3.4.5 further uses the famous Bregman–Minc inequality. We thank the anonymous reviewer for this and include the derivation in \Cref{sec:altproof}. In constrast, our proof is self-contained and we believe it provides further insight into the structure of the Sinkhorn/Bethe approximations. See \Cref{sec:related} for further details.}
Here we consider a modified matrix $\hat{A}$ that contains the $k$ distinct columns of $A$.
We define a distribution $\mu$ on permutations of the domain where the probability of a permutation $\sigma$ is proportional to its contribution to the permanent of $A$.
There is a many-to-one mapping from such permutations to 0-1 $N \times k$ matrices with row sums 1 and column sums $\phi_j$, the number of times the  $j$th column of $\hat{A}$ appears in $A$.
We next define an $N \times k$ real-valued, non-negative matrix $P$ with row sums 1 and column sums $\phi_j$, in terms of the marginals of the distribution $\mu$.
We also define a different distribution $\nu$ on 0-1 $N \times k$ row-stochastic matrices by independent sampling from $P$.
Finally, we use the fact that the KL-divergence between $\mu$ and $\nu$ is non-negative to get the required upper bound on the scaled Sinkhorn approximation with a doubly stochastic witness $Q$ (obtained from $P$). This proof technique is inspired by the recent work of Anari and Rezaei \cite{AR19} that gives a tight $\sqrt{2}^N$ bound on the approximation ratio of the Bethe approximation for the permanent of an $N \times N$ non-negative matrix.

Though this bound on the quality of the Bethe and scaled Sinkhorn approximations for non-negative matrices with $k$ distinct columns suffices for our PML applications, interestingly we show that it can be extended to non-negative matrices with bounded rank. 
In order to obtain an upper bound on the Bethe and scaled Sinkhorn approximation as a function of the non-negative rank of $A$, recall that we need to show the existence of a suitable doubly stochastic witness $Q$ which certifies the required guarantee.
We express the permanent of $A$ as the sum of $O(\exp(k \log(N/k)))$ terms of the form $\perm(U)\perm(V)$ where matrices $U$ and $V$ have at most $k$ distinct columns.
We focus on the largest of these terms, and construct a doubly stochastic witness $Q$ for matrix $A$ from the witnesses for matrices $U$ and $V$ in this largest term. This doubly stochastic witness $Q$ certifies the required guarantee and we get an upper bound on the scaled Sinkhorn approximation as a function of the non-negative rank. This result for the scaled Sinkhorn approximation further implies an upper bound for the Bethe approximation.

Even with this improved bound on the quality of the Bethe and scaled Sinkhorn approximations as applied to the PML objective, challenges remain in obtaining an improved approximate PML distribution.
In particular, we do not know of an efficient algorithm to maximize the Bethe or the scaled Sinkhorn permanent of the profile probability matrix over a family of distributions as it would be needed to compute the Bethe or the scaled Sinkhorn approximation to the optimum of the PML objective. 
Prior work by Vontobel suggests an alternating maximization approach, but this is only guaranteed to produce a local optimum.
To address this, we revisit the efficiently computable convex relaxation from \cite{CSS19} and show that this is suitably close to the scaled Sinkhorn approximation.
This is quite surprising as the prior derivation of this relaxation in \cite{CSS19} was purely combinatorial and had nothing to do with the scaled Sinkhorn approximation.

The final challenge towards obtaining our PML results is to round the fractional solution produced so that the approximation guarantee is preserved.
The rounding procedure from \cite{CSS19}  does not immediately suffice, but we present a more sophisticated and delicate rounding procedure that does indeed give us the required approximation guarantee.
The rounding algorithm proceeds in three steps, where in the first step we first apply a procedure analogous to \cite{CSS19} to handle large probability values and in the later steps we provide a new procedure to the smaller probability values;  in each step, we ensure that the objective function does not drop significantly. The input to the rounding procedure is a matrix where the rows correspond to discretized probability values 
and the columns correspond to distinct frequencies.
We create rows corresponding to new probability values in the course of the rounding algorithm, maintain column sums and eventually ensure that all row sums are integral, and ensure that the objective function has not dropped significantly.
\section{Preliminaries}\label{sec:prelims}
\newcommand{\vj}{\textbf{v}_{j}}
\newcommand{\uj}{\textbf{u}_{j}}
Let $[a,b]$ and $[a,b]_{\R}$ denote the interval of integers and reals $\geq a$ and $\leq b$ respectively. 
Let $\bX$ be the domain of elements and $\Ns\defeq |\bX|$ be its size. Let $\ma \in \R^{\bX \times \bX}$ be a non-negative matrix, where its $(x,y)$'th entry is denoted by $\Axy$. We further use $\ma_{x:}$ and $\ma_{:y}$ to denote the row and column corresponding to $x$ and $y$ respectively. The non-negative rank of a non-negative matrix $\ma \in \R^{\bX \times \bX}$ is equal to the smallest number $k$ such there exist non-negative vectors $\vj,\uj \in \R^{\bX}$ for $j \in [1,k]$ such that $\ma=  \sum_{j \in [1,k]} \vj \uj^{\top}$. Let $\permD$ be the set of all permutations of domain $\bX$ and we denote a permutation $\sigma$ in the following way $\sigma=\{(x,\sigma(x)) \text{ for all } x\in \bX \}$. The permanent of a matrix $\ma$ denoted by $\perm(\ma)$ is defined as follows,
$$\perm(\ma)\defeq \sum_{\sigma \in \permD} \prod_{e \in \sigma}\ma_{e}~.$$
Let $\dstoc \subseteq \kds$ be the set of all non-negative matrices that are doubly stochastic. For any matrix $\ma \in \kds$ and $\bQ \in \dstoc$, we define the following set of functions:
\begin{equation}\label{eq:ftsd}
\frst(\ma,\bQ)\defeq \sum_{(x,y) \in \bX \times \bX} \bQ_{x,y}\log \left(\frac{\ma_{x,y}}{\bQ_{x,y}}\right)\quad \text{and} \quad \scnd(\bQ)=\sum_{(x,y) \in \bX \times \bX}(1-\bQ_{x,y})\log \left( 1-\bQ_{x,y} \right)~.
\end{equation}
Further,
$$\betheq(\ma,\bQ)\defeq \frst(\ma,\bQ) + \scnd(\bQ)~.$$
Using these definitions, we define the Bethe permanent of a matrix. 
\begin{defn}\label{def:bethe}
For a matrix $\ma \in \kds$, the Bethe permanent of $\ma$ is defined as follows,
$$\bethe(\ma)\defeq \max_{\bQ \in \dstoc} \expo{\betheq(\ma,\bQ)}~.$$
\end{defn}
A well known and important result about the Bethe permanent is that it lower bounds the value of permanent of a non-negative matrix and we state this result next.
\begin{lemma}[\cite{Gur11} based on \cite{Sch98}]\label{bethelb}
	For any non-negative $\ma \in \kds$, the following holds,
	$$\bethe(\ma)\leq \perm(\ma)$$
\end{lemma}
We next define the Sinkhorn permanent of a matrix and later we state the relationship between the Bethe and the Sinkhorn permanent.
\begin{defn}
For a matrix $\ma \in \kds$, the Sinkhorn permanent of $\ma$ is defined as follows,
	$$\sink(\ma)\defeq \max_{\bQ \in \dstoc} \expo{\frst(\ma,\bQ)}~.$$
\end{defn}
To establish the relationship between the Bethe and the Sinkhorn permanent we need the following lemma from \cite{GS14}.
\begin{lemma}[Proposition 3.1 in \cite{GS14}]\label{lem:termlb}
	For any distribution $\bp \in \R_{\geq 0}^{\bX}$, the following holds,
	$$\sum_{x\in \bX}(1-\bp_{x})\log(1-\bp_{x})\geq -1~.$$
\end{lemma}
For any matrix $\bQ \in \dstoc$, each row of $\bQ$ is a distribution; therefore the following holds,
$$\scnd(\bQ) \geq -\Ns~.$$
As a corollary of the above inequality we have,
\begin{cor}\label{cor:ssinkbethe}
	For any non-negative matrix $\ma \in \kds$, the following inequality holds,
	$$\exp(-\Ns) \sink(\ma) \leq \bethe(\ma)~.$$
	\end{cor}
Later we will see that it is convenient to work with $\exp(-\Ns) \sink(\ma)$ than $\sink(\ma)$ itself; we define this expression to be scaled Sinkhorn and we formally state it next.
\begin{defn}\label{def:ssink}
	For a matrix $\ma \in \kds$, the scaled Sinkhorn permanent of $\ma$ is defined as follows,
	$$\ssink(\ma)\defeq \max_{\bQ \in \dstoc} \expo{\frst(\ma,\bQ)-\Ns}~.$$
\end{defn}
The above expression can be equivalently stated as,
$$\ssink(\ma)= \exp(-\Ns) \sink(\ma)~.$$
Combining \Cref{bethelb} and \Cref{cor:ssinkbethe} we get the following result.
\begin{cor}\label{cor:ssinklb}
For any matrix $\ma \in \kds$, the following inequality holds,
	$$\ssink(\ma) \leq \bethe(\ma),$$
	which further implies, 
	$$\ssink(\ma) \leq \perm(\ma)~.$$
\end{cor}
Other than approximations to the permanent of a matrix, we next state two important results that will be helpful throughout the paper. The first result is the Stirling's approximation for factorial function and the second is the non-negativity result of KL divergence between two distributions.
\begin{lemma}[Stirling's approximation]\label{lem:ster}
	For all $n \in \Z$, the following holds:
	$$ \exp(n\log n-n)\leq n! \leq O(\sqrt{n}) \exp(n\log n-n)~.$$
\end{lemma}
Let $\mu$ and $\nu$ be distributions defined on some domain $\Omega$. The \emph{KL divergence} denoted $\KL{\mu\|\nu}$ between distributions $\mu$ and $\nu$ is defined as follows,
$$\KL{\mu\|\nu}\defeq \sum_{\mx \in \Omega}\mu(\mx) \log \frac{\mu(\mx)}{\nu(\mx)}=\expt{\mx \sim \mu}{\log \mu(\mx)}- \expt{\mx \sim \mu}{\log \nu(\mx)}$$
\begin{lemma}[Non-negativity of KL divergence]\label{lem:kldiv}
For any distributions $\mu$ and $\nu$ defined on domain $\Omega$, the KL divergence between distributions $\mu$ and $\nu$ satisfies,
$$\KL{\mu\|\nu} \geq 0~.$$
\end{lemma}

In the remainder of this section we provide formal definitions related to PML. 
\subsection{Profile maximum likelihood}
Let $\simplex \subset [0,1]_{\R}^{\bX}$ be the set of all discrete distributions supported on domain $\bX$. 
Here on we use the word distribution to refer to discrete distributions.
Throughout this paper we assume that we receive a sequence of $n$ independent samples from an underlying distribution $\bp \in \simplex$. Let $\bX ^n$ be the set of all length $n$ sequences and $y^n \in \bX^n$ be one such sequence with $y^n_{i}$ denoting its $i$th element.
The probability of observing sequence $y^n$ is:
$$\bbP(\bp,y^n) \defeq \prod_{x \in \bX}\bp_x^{\bff(y^n,x)}$$
where $\bff(y^n,x)= |\{i\in [n] ~ | ~ y^n_i = x\}|$ is the frequency/multiplicity of symbol $x$ in sequence $y^n$ and $\bp_x$ is the probability of domain element $x\in \bX$. 

For any given sequence one could define its profile (histogram of a histogram or fingerprint) that is sufficient statistic for symmetric property estimation. 
\begin{defn}[Profile]
For any sequence $y^n \in \bX^n$, let $\setd=\{ \bff(y^n,x) \}_{x \in \bX} \backslash \{0\}$ be the set of all its non-zero distinct frequencies and $\eled_1,\eled_2,\dots, \eled_{|\setd|}$ be elements of the set $\setd$. The \emph{profile} of a sequence $y^n \in \bX^{n}$ denoted $\phi=\Phi(y^n) \in \Z^{|\setd|}$ is 
\[
\phi\defeq(\F_j)_{j\in [1, |\setd|]}
\text{ , where }
\F_j=\F_j(y^n)\defeq|\{x\in \bX ~|~\bff(y^n,x)=\eled_{j} \}|
\]
is the number of domain elements with frequency $\eled_{j}$ in $y^{n}$\footnote{The profile does not contain information about the number of unseen domain elements.}. We call $n$ the length of profile $\F$ and as a function of profile $\phi$, $n = \sum_{j \in [1,|\setd|]} \eled_{j} \cdot \F_j$. Let $\Phi^n$ denote the set of all profiles of length $n$. We use $k$ to denote the number of distinct frequencies in the profile $\phi$ and $k= |\setd|$.\footnote{Note the number of distinct frequencies denoted $k$ in a length $n$ sequence is always upper bounded by $\sqrt{n}$.} For convenience we use $\mvec \in \Z^{k}$ to denote the vector of observed frequencies, therefore $\mvec_{j}=\eled_{j}$ for all $j \in [1,k]$.
\end{defn}
For any distribution $\bp \in \simplex$, the probability of a profile $\phi \in \Phi^n$ is defined as:
\begin{equation}\label{eqpml1}
\probpml(\bp,\phi)\defeq\sum_{\{y^n \in \bX^n~|~ \Phi (y^n)=\phi \}} \bbP(\bp,y^n)\\
\end{equation}

Let $x^n$ be a sequence such that $\Phi (x^n)=\phi$. We define a \emph{profile probability matrix} $\mapphi$ with respect to sequence $x^n$ (therefore profile $\phi$) and distribution $\bp$ as follows,
\begin{equation}\label{eq:matrixqphi}
\mapphi_{z,y}\defeq \bp_{z}^{\my} \text{ for all } z,y \in \bX,
\end{equation}
where $\my\defeq \bff(x^n,y)$ is the frequency of domain element $y\in \bX$ in sequence $x^n$ and recall $\Phi(x^n)=\phi$. We are interested in the permanent of the matrix $\mapphi$, and note that the $\perm(\mapphi)$ is invariant under the choice of sequences $x^n$ that satisfy $\Phi(x^n)=\phi$. Therefore we index the matrix $\mapphi$ with profile $\phi$ rather than sequence $x^n$ itself. The number of distinct columns in $\mapphi$ is equal to number of distinct observed frequencies plus one (for the unseen), i.e.  $k+1$.

 The probability of a profile $\phi \in \Phi^n$ with respect to distribution $\bp$ (from Equation 20 in \cite{OSZ03}, Equation 15 in \cite{PJW17}) in terms of permanent of matrix $\mapphi$ is given below:
\begin{equation}\label{eqpml2}
\probpml(\bp,\phi)=\cphi \cdot \left(\prod_{j\in [0,k]}\frac{1}{\phi_{j}!}\right) \cdot \perm(\mapphi)
\end{equation}
where $\cphi\defeq \frac{n!}{\prod_{j\in [1,k]}(\eled_{j}!)^{\phi_{j}}}$ and $\phi_{0}$ is the number of unseen domain elements\footnote{Given a distribution $\bp$, we know its domain $\bX$ and therefore the value of $\phi_{0}$.}.

The distribution which maximizes the probability of a profile $\phi$ is the profile maximum likelihood distribution which we formally define next.
\begin{defn}[Profile maximum likelihood] For any profile $\phi \in \Phi^{n}$, a \emph{profile maximum likelihood} (PML) distribution $\bp_{\mathrm{pml},\phi} \in \simplex$ is:
 $$\bp_{\mathrm{pml},\phi} \in \argmax_{\bp \in \simplex} \probpml(\bp,\phi)$$
 and $\probpml(\bp_{\mathrm{pml},\phi},\phi)$ is the maximum PML objective value.
\end{defn}

The central goal of this paper is to define efficient algorithms for computing approximate PML distributions defined as follows.
 
\begin{defn}[Approximate PML]
 For any profile $\phi \in \Phi^{n}$, a distribution $\bp^{\beta}_{\mathrm{pml},\phi} \in \simplex$ is a $\beta$-\emph{approximate PML} distribution if $$\probpml(\bp^{\beta}_{\mathrm{pml},\phi},\phi)\geq \beta \cdot \probpml(\bp_{\mathrm{pml},\phi},\phi)$$
\end{defn}
\section{Results}\label{sec:results}
Here we state the main results of this paper. In our first class of main results, we improve the analysis of the scaled Sinkhorn permanent for structured non-negative matrices. We first show that the scaled Sinkhorn permanent approximates the permanent of a non-negative matrix $\ma$, where the approximation factor (up to log factors) depends exponentially on the non-negative rank of the smatrix $\ma$. We formally state this result next.

\begin{restatable}[Scaled Sinkhorn permanent approximation for low non-negative rank matrices]{theorem}{thmmainlowrank}
	\label{thm:mainlowrank}
	For any matrix $\ma \in \kds$ with non-negative rank at most $k$, the following inequality holds,
	\begin{equation}\label{eq:conditiontwo}
	\ssink(\ma) \leq \perm(\ma) \leq \expo{O\left(k\log \frac{\Ns}{k}\right)} \ssink(\ma)~.
	\end{equation}
\end{restatable} 
Further using $\ssink(\ma) \leq \bethe(\ma)$ (See \Cref{cor:ssinklb}) and $\bethe(\ma) \leq \perm(\ma)$ (See \Cref{bethelb}) we immediately get the same result for the Bethe permanent.
\begin{cor}[Bethe permanent approximation for low non-negative rank matrices]\label{thm:bethemain}
For any matrix $\ma \in \kds$ with non-negative rank at most $k$, the following inequality holds,
\begin{equation}
\bethe(\ma) \leq  \perm(\ma) \leq \expo{O\left( k\log \frac{\Ns}{k} \right)} \bethe(\ma)~.
\end{equation}
\end{cor}
Interestingly, in the worst case, Sinkhorn is an $e^{N}$ approximation to the permanent of $\ma \in \kds$, even when $\ma$ has at most 1 distinct column (e.g. consider the all 1's matrix). Consequently, for matrices with non-negative rank at most $k$, whenever $k = o(N / \log N)$, scaled Sinkhorn is a compelling alternative to Sinkhorn, with a tighter worst-case multiplicative approximation to the permanent.

Our results improve the analysis of the Bethe permanent for such structured matrices. Previously, the best known analysis of the Bethe permanent showed an $\sqrt{2}^{\Ns}$-approximation factor to the permanent \cite{AR19}. 
The analysis in \cite{AR19} is tight for general non-negative matrices and the authors showed that this bound cannot be improved without leveraging further structure.
Our next result is of similar flavor, and we provide an asymptotically tight example for \Cref{thm:mainlowrank} and \Cref{thm:bethemain}.

\begin{restatable}[Lower bound for the Bethe and the scaled Sinkhorn permanents approximation]{theorem}{thmlb}
	\label{thm:lb}
	There exists a matrix $\ma \in \kds$ with non-negative rank $k$, that satisfies 
	\begin{equation}\label{eq:lowerbound}
	 \perm(\ma) \geq \expo{\Omega\left(k\log \frac{\Ns}{k}\right)} \bethe(\ma)~,
	\end{equation}
	which further implies,
	\begin{equation}\label{eq:lowerboundsink}
	\perm(\ma) \geq \expo{\Omega\left(k\log \frac{\Ns}{k}\right)} \ssink(\ma)~.
	\end{equation}
\end{restatable}

An immediate application of these above stated results is for PML. Recall, that for any fixed distribution $\bp$ and profile $\phi$, $\probpml(\bp,\phi)$ is proportional to the permanent of the non-negative matrix $\ma^{\bp,\phi}$ (See \Cref{sec:prelims} for the definition of $\ma^{\bp,\phi}$). Note the number of distinct columns in the profile probability matrix $\ma^{\bp,\phi}$ is upper bounded by the number of distinct frequencies plus one, which further is always less than $\sqrt{n}+1$ (where $n$ is the length of the profile). Therefore the non-negative rank of the profile probability matrix $\ma^{\bp,\phi}$ is always upper bounded by $\sqrt{n}+1$. Since $\ssink(\ma)$ can be computed in polynomial time \cite{CSS19}\footnote{$\ssink(\ma)$ corresponds to a convex optimization problem and a minor modification of the approach in \cite{CSS19} to solve a related, but slightly different optimization problem, yields a polynomial time algorithm to compute $\ssink(\ma)$ up to high accuracy.}, \Cref{thm:mainlowrank} implies an efficient algorithm to approximate the value $\probpml(\bp,\phi)$ for a fixed distribution $\bp$ up to multiplicative $\exps{O(\sqrt{n} \log n)}$ factor, and is also the best known approximation factor achieved by a deterministic algorithm. 

Analyzing the relationship between the Bethe permanent and the permanent of the profile probability matrix was posed as an interesting research direction in \cite{Von11} (See Section \RomanNumeralCaps{7}). Moreover, one of the primary interests in the area of algorithmic statistics/machine learning is to efficiently compute the PML distribution. Exploiting the structure of doubly stochastic matrix $\bQ$ that maximizes $\ssink(\ma^{\bp,\phi},\bQ)$ combined with the probability discretization idea from \cite{CSS19} we provide an efficient algorithm to compute an approximate PML distribution. We use \Cref{thm:main} to argue the approximation of our approximate PML distribution and we summarize this result below.

\begin{restatable}[$\expo{\sqrt{n} \log n}$-approximate PML]{theorem}{secondthm}
	\label{thm:second}
	For any given profile $\phi \in \Phi^n$, \Cref{euclid} computes an $\expo{-O(\sqrt{n}\log n)}$-approximate PML distribution in $\otilde(n^{1.5})$ time.
\end{restatable}
Previous to our work, the best known result \cite{CSS19} gave an efficient algorithm to compute $\exps{-O(n^{2/3}\log n)}$-approximate PML distribution.

One important application of approximate PML is in symmetric property estimation. In \cite{ADOS16}, the authors showed that approximate PML distribution based plug-in estimator is sample complexity optimal for estimating entropy, support, support coverage and distance to uniformity. Further combining their result with our \Cref{thm:second} we get the efficient version of Theorem 2 in \cite{ADOS16} and we summarize this result next.

	\begin{theorem}[Efficient universal estimator using approximate PML]
		Let $n$ be the optimal sample complexity of estimating entropy, support, support coverage and distance to uniformity. If $\epsilon \geq \frac{c}{n^{0.2499}}$ for some constant $c>0$, then there exists a PML based universal plug-in estimator that runs in time $\otilde(n^{1.5})$ and is sample complexity optimal for estimating entropy, support, support coverage and distance to uniformity to accuracy $\epsilon$.
	\end{theorem}
Note the dependency on $\epsilon$ in the above theorem and the approximation factor in \Cref{thm:second} are strictly better than \cite{CSS19}, which is another efficient PML based approach for universal symmetric property estimation; \cite{CSS19} works when the error parameter $\epsilon \geq \frac{1}{n^{0.166}}$.

Recent work \cite{HO19}, further gives the broad optimality of approximate PML. \cite{HO19} shows optimality of approximate PML distribution based estimator for other symmetric properties, such as, sorted distribution estimation (under $\ell_{1}$ distance), $\alpha$-Renyi entropy for non-integer $\alpha>3/4$, and other broad class of additive properties that are Lipschitz. \cite{HO19} also provides a PML-based tester to test whether an unknown distribution is $\geq \epsilon$ far from a given distribution in $\ell_{1}$ distance and achieves the optimal sample complexity up to logarithmic factors. Our result further implies an efficient version of all these results.


\subsection{Related work}
\label{sec:related}

We divide the related work into two broad categories: permanent approximations and profile maximum likelihood.

\paragraph{Permanent approximations:} The first set of related work is with respect to computing the permanent of matrices. \cite{Val79} showed that computing the permanent of matrices even when it has entries in {0, 1} is \#P-Hard. This led to the study of computing approximations to the permanent. Additive approximation to the permanent for arbitrary $\ma$ was given by \cite{Gur05}. On the other hand, multiplicative approximation to the permanent (or even determining the sign) is hard for general $\ma$ \cite{AA11,GS16}. This hardness results led to the study of the multiplicative approximation to the permanent for special class of matrices and one such class is the set of non-negative matrices. In this direction, \cite{JSV04} gave the first efficient randomized algorithm to approximate the permanent within $(1+\epsilon)$ multiplicative accuracy. There has also been a rich literature on coming up with deterministic approximation to the permanent of non-negative matrices. \cite{LSW98} gave the first deterministic algorithm to the permanent of $\Ns \times \Ns$ non-negative matrices with approximation ratio $\leq e^{\Ns}$. 
\cite{Gur11} using an inequality from \cite{Sch98} showed that the Bethe permanent lower bounds the value of the permanent of non-negative matrices. 
We refer the reader to \cite{Von13, GS14} for the polynomial computability of the Bethe permanent and \cite{AR19} for a more rigorous literature survey on the Bethe permanent and other related work.

As discussed in the footnote of the introduction, an anonymous reviewer showed us an alternative and simpler proof for the upper bound on the scaled Sinkhorn approximation to the permanent of matrices with at most $k$ distinct columns (\Cref{thm:main}). This alternative proof is deferred to \Cref{sec:altproof} and is derived using Corollary 3.4.5. in Barvinok's book~\cite{Bar17}. The result in turn, is proved using the Bregman-Minc inequality conjectured by Minc, cf. \cite{Minc78} and later proved by Bregman \cite{Bre73}. The Bregman-Minc inequality is well-known and there are many different proofs~\cite{Sch78,Rad97,AS04} known. 
In comparison to this alternative proof for matrices with $k$ distinct columns (\Cref{thm:main}), our proof is self contained and intuitive. We believe it could help provide further insights into the Sinkhorn/Bethe approximations.

\paragraph{Profile maximum likelihood:} The second set of related work is with respect to profile maximum likelihood and its applications. As discussed in the introduction, PML was introduced by \cite{OSSVZ04}. Many heuristic approaches such as the EM algorithm \cite{OSSVZ04}, algebraic approaches \cite{ADMOP10} and a dynamic programming approach~\cite{PJW17} have been proposed to compute approximations to  PML. Further \cite{Von12, Von14} used the Bethe permanent as a heuristic to compute the PML distribution. All these approaches don't provide any theoretical guarantees for the quality of the approximate PML distribution and it was an open question to efficiently compute a non-trivial approximate PML distribution. \cite{CSS19} gave the first efficient algorithm to compute the $\exps{-n^{2/3} \log n}$ approximate PML distribution, where $n$ is the number of samples.

The connection between PML and universal estimators was first studied in~\cite{ADOS16}. \cite{ADOS16} showed that an approximate PML distribution can be used as an universal estimator for estimating symmetric properties, namely entropy, distance to uniformity, support size and coverage. See \cite{HO19} for broad applicability of approximate PML in property testing and estimating other symmetric properties such as sorted $\ell_{1}$ distance, Renyi entropy, and other broad class of additive properties.  \cite{CSS19} combined with \cite{ADOS16}, gave the first efficient PML based universal estimator for symmetric property estimation. There have been several other approaches for designing universal estimators for symmetric properties. Valiant and Valiant~\cite{VV11a} adopted and rigorously analyzed a linear programming based approach for universal estimators proposed by \cite{ET76} and showed that it is sample complexity optimal in the constant error regime for estimating certain symmetric properties (namely, entropy, support size, support coverage, and distance to uniformity). Recent work of Han, Jiao and Weissman~\cite{HJW18} applied a local moment matching based approach in designing efficient universal symmetric property estimators for a single distribution. \cite{HJW18} achieves the optimal sample complexity in a broader error regimes for estimating the power sum function, support and entropy.

Estimating symmetric properties of a distribution is a rich field and extensive work has been dedicated to studying their optimal sample complexity for estimating each of these properties. Optimal sample complexities for estimating many symmetric properties were resolved in the past few years; support~\cite{VV11a, WY15}, support coverage~\cite{OSW16,ZVVKCSLSDM16}, entropy~\cite{VV11a, WY16},  distance to uniformity~\cite{VV11b, JHW16}, sorted $\ell_{1}$ distance \cite{VV11b, HJW18}, Renyi entropy~\cite{AOST14, AOST17}, KL divergence~\cite{BZLV16, HJW16} and many others.


{\bf Comparison to \cite{CSS19}}: \cite{CSS19} provides the first efficient algorithm for computing $\beta$-approximate PML distribution for $\beta>\exps{-n^{1-\delta}}$ for some constant $\delta>0$, where $n$ is the number of samples. Formally, \cite{CSS19} computes an $\exps{-n^{2/3} \log n}$-approximate PML distribution. Suppose $\ell$ and $k$ are the number of distinct probability values and frequencies respectively, then \cite{CSS19} provides a convex program that using \emph{combinatorial techniques} they analyze and show that it approximates the PML objective up to $\exps{-\otilde(\ell \times k)}$ multiplicative factor. Further this convex program outputs a fractional solution and \cite{CSS19} provides a rounding algorithm that outputs a valid integral solution (that corresponds to a valid distribution). \cite{CSS19} further incurs a $\exps{-\otilde(\ell \times k)}$ multiplicative loss in the rounding procedure. Using the discretization results, up to $\exps{-n^{2/3} \log n}$-multiplicative loss one can assume $\ell,k \leq n^{1/3}$ and therefore \cite{CSS19} provides a $\exps{-n^{2/3} \log n}$-approximate PML distribution.

However in our current work, using results for the scaled Sinkhorn permanent, we show that the same convex program in \cite{CSS19} approximates the PML objective up to $\exps{-\otilde(\ell + k)}$ multiplicative factor. Further we also provide a better rounding algorithm that outputs a valid distribution and incur a $\exps{-\otilde(\ell + k)}$ multiplicative loss. Further using the discretization results, up to $\exps{-\sqrt{n} \log n}$-multiplicative loss one can assume $\ell,k \leq \sqrt{n}$ and therefore our work provides a $\exps{-\sqrt{n} \log n}$-approximate PML distribution.

\section{The Sinkhorn permanent for structured matrices.}\label{sec:bethe}
In this section, we provide the proof for our first main theorem (\Cref{thm:mainlowrank}). 
 We show that the scaled Sinkhorn permanent of a non-negative matrix approximates its permanent, where the approximation factor is exponential in the non-negative of the matrix (up to log factors). Our proof is divided into two parts. First in \Cref{sec:sinkcol}, we work with a simpler setting of matrices $A$ with at most $k$ distinct columns and prove the following lemma.
 \begin{restatable}[Scaled Sinkhorn permanent approximation]{lemma}{mainthm}
 	\label{thm:main}
 	For any matrix $\ma \in \kds$ with at most $k$ distinct columns, the following holds,
 	\begin{equation}\label{eq:conditiontwo}
 	\ssink(\ma) \leq \perm(\ma) \leq \expo{O\left(k\log \frac{\Ns}{k}\right)} \ssink(\ma)~.
 	\end{equation}
 \end{restatable}
Further using the above result, in \Cref{sec:lowrank} we prove our main theorem (\Cref{thm:mainlowrank}) for low non-negative rank matrices. 
\subsection{The Sinkhorn permanent for distinct column case.}\label{sec:sinkcol}
We start this section by defining some notation that captures the structure of repetition of columns in a matrix. For the remainder of this section we fix a matrix $\ma \in \kds$. We let $k$ denote the number of distinct columns of $\ma$ and use $\co,\ct,\dots \ck$ to denote these distinct columns. Further we let $\mad=[\co~|~\ct~|\dots|~\ck]$ denote the $\bX \times k$ matrix formed by these distinct columns. We use $\may$ to denote the $y$'th column of matrix $\ma$ and let $\phi_{j} \defeq |\{y\in \bX~|~\may=\cj \}|$ denote the number of columns equal to $\cj$. It is immediate that,
 \begin{equation}
 \sum_{j\in [1,k]}\phi_{j} =\Ns ~,
 \end{equation}
where $\Ns=|\bX|$ is the size of the domain.
For any matrix $\bP \in \R_{\geq0}^{\bX \times k}$ define,
\begin{equation}\label{eq:fnp}
\fnp(\ma, \bP)\defeq \sum_{x \in \bX} \sum_{j \in [1,k]}\Pxj\log \frac{\madxj}{\Pxj}+\sum_{j \in [1,k]}\phi_{j}\log \phi_{j}-\sum_{j \in [1,k]}\phi_{j}~.
\end{equation}
In the first half of this section, we show existence of a matrix $\bP \in \R_{\geq0}^{\bX \times k}$ (See \Cref{lem:condip}) such that $\sum_{j \in [1,k]}\Pxj=1 \text{ for all } x\in \bX$, $\sum_{x \in \bX}\Pxj=\phi_{j} \text{ for all } j \in [1,k]$, and further (See \Cref{lem:upperm}),
\begin{equation}\label{eq:sone}
\log \perm(\ma)	\leq  O\left( k\log \frac{\Ns}{k} \right)+ \fnp(\ma, \bP)~.
\end{equation}
Later in the second half (See \Cref{thm:ds}), we show that for any matrix $\bP \in \R_{\geq 0}^{\bX \times k}$ that satisfies $\sum_{j \in [1,k]}\Pxj=1 \text{ for all } x\in \bX$ and $\sum_{x \in \bX}\Pxj=\phi_{j} \text{ for all } j \in [1,k]$, there exists a matrix $\bQ \in \dstoc$ (recall $\dstoc$ is the set of all $\bX \times \bX$ doubly stochastic matrices) that satisfies,
\begin{equation}\label{eq:stwo}
\fnp(\ma, \bP) = \frst(\ma,\bQ)-\Ns~.
\end{equation}
However, using \Cref{cor:ssinklb} we already know that, $\ssink(\ma) \leq \log \perm(\ma)$. Further using the definition of $\ssink(\ma)$ and combining with \Cref{eq:sone,eq:stwo} we get,
$$\ssink(\ma) \leq  \perm(\ma) \leq \expo{O\left( k\log \frac{\Ns}{k} \right)} \ssink(\ma)~.$$
In the remainder, we provide proofs for all the above mentioned inequalities and we need the following set of definitions. Let $\kr\subseteq \kzo$, be the subset of all $\bX \times k$ matrices that are row stochastic, meaning there is exactly single $1$ in each row.
Let $\ka \subseteq \kr$ be the set of matrices such that any $\mx \in \ka$ satisfies  $\sum_{x \in \bX}\Xxj=\phi_{j} \text{ for all }j \in [1,k]$.

\begin{defn}
Let $\fna:\permD \rightarrow \ka$ be the function that takes a permutation $\sigma \in \permD$ as input and returns a matrix $\mx \in \ka$ in the following way,
 	\begin{align}
 	\Xxj =
 	\begin{cases}
 	1 & \text{ if } \ma_{: \sigma(x)}=\cj \\
 	0 & \text{ otherwise}
 	\end{cases}
 	\quad \quad 	\text{ for all } x\in \bX.
 	\end{align}
 \end{defn}
{\bf Remark}: Note that as desired $\fna(\sigma)\in \ka$ for all $\sigma \in \permD$ because of the following. For any $\sigma \in \permD$, let $\mx\defeq \fna(\sigma)$. Since $\cj$ for all $j \in [1,k]$ are distinct, we have $\sum_{j \in [1,k]}\mx_{x,j}=1$. Further for any $j\in [1,k]$, $\sum_{x\in \bX}\mx_{x,j}=\sum_{\{x\in \bX~|~\ma_{: \sigma(x)}=\cj\}}1=\sum_{\{x\in \bX~|~\ma_{\cdot x}=\cj\}}1=\phi_j$.

We next define the probability of a permutation $\sigma \in \permD$ with respect to matrix $\ma$ as follows, 
\begin{equation}\label{eq:probsigma}
\prob{\sigma}\defeq \frac{\prod_{e\in \sigma} \ma_{e}}{\perm(\ma)}
\end{equation}
Further we define a marginal distribution $\mu$ on $\kr$ and later we will establish that this is indeed a probability distribution, that is, probabilities add up to 1. 
\begin{align}\label{eq:one}
\mu(\mx) \defeq 
\begin{cases}
0 & \text{ if } \mx \in \kr \backslash \ka\\
\sum_{\{\sigma \in \permD ~|~\fna(\sigma) = \mx\}} \prob{\sigma} &  \text{ if } \mx \in \ka~.
\end{cases}
\end{align}
For $\mx \in \ka$, we next provide another equivalent expression for $\mu(\mx)$.
\begin{equation}
\begin{split}
\mu(\mx)&=\sum_{\{\sigma \in \permD ~|~\fna(\sigma) = \mx\}} \prob{\sigma}=\sum_{\{\sigma \in \permD ~|~\fna(\sigma) = \mx\}}\frac{\prod_{(x,\sigma(x))} \ma_{x,\sigma(x)}}{\perm(\ma)},\\
&=\frac{1}{\perm(\ma)}\sum_{\{\sigma \in \permD ~|~\fna(\sigma) = \mx\}}\prod_{x \in \bX} \prod_{j \in [1,k]}\mad_{x,j}^{\Xxj}\\
&= \left(\prod_{j \in [1,k]} \phi_{j}! \right) \left( \prod_{x \in \bX} \prod_{j \in [1,k]}\madxj^{\Xxj} \right) \left(\frac{1}{\perm(\ma)}\right)
\end{split}
\end{equation}
In the first and second equality, we used definitions of $\mu(\mx)$ and $\prob{\sigma}$ (See \Cref{eq:probsigma}). For any $\sigma \in \permD$, let $\mx=\fna(\sigma)$. Further for any $x\in \bX$, let $j'$ be such that $\ma_{: \sigma(x)}=\cjp$, then $\ma_{x,\sigma(x)}=\mad_{x,j'}$ that is further equal to $\prod_{j \in [1,k]}\mad_{x,j}^{\Xxj}$ because $\Xxj$ is equal to $1$ if $j=j'$ and $0$ otherwise. Therefore the third equality holds. For the final equality, observe that for any $\sigma \in \permD$ if we let $\mx=\fna(\sigma)$, then for each $j\in [1,k]$, any permutation within the subset of elements $\{x\in \bX~|~\ma_{: \sigma(x)}=\cj\}$ results in a permutation $\sigma'$ that satisfies $\fna(\sigma')=\mx$. These permutations can be carried out independently for each $j \in [1,k]$ that corresponds to $\prod_{j \in [1,k]} \phi_{j}!$ number of permutations and all of them have the same $\prod_{x \in \bX} \prod_{j \in [1,k]}\mad_{x,j}^{\Xxj}$ value.

Using the derivation from above, the definition for $\mu$ can also be written as follows:
\begin{align}\label{eq:two}
\mu(\mx) =
\begin{cases}
0 & \text{ if } \mx \in \kr \backslash \ka\\
\left(\prod_{j \in [1,k]} \phi_{j}! \right) \left( \prod_{x \in \bX} \prod_{j \in [1,k]}\madxj^{\Xxj} \right) \left(\frac{1}{\perm(\ma)}\right) & \text{ if } \mx \in \ka~.
\end{cases}
\end{align}
Note for $\mx \in \ka$, the expression for $\mu(\mx)$ can be equivalently written as follows:
\begin{equation}\label{eq:mux}
\mu(\mx) =\left(\prod_{j \in [1,k]} \phi_{j}! \right) \left( \prod_{\{(x,j) \in \bX \times [1,k]~|~\Xxj=1\}} \madxj \right)\left(\frac{1}{\perm(\ma)}\right)~.
\end{equation}
We next show that the $\mu$ defined above is a valid distribution.
$$\sum_{\mx\in \kr}\mu(\mx)=\sum_{\mx\in \ka}\mu(\mx)=\sum_{\mx\in \ka}\sum_{\{\sigma \in \permD ~|~\fna(\sigma) = \mx\}}\Pr(\sigma)=\sum_{\sigma \in \permD}\Pr(\sigma)=1$$
{\bf Remark}: The domain of distribution $\mu$ is $\kr$, but its support is subset of $\ka$. 
\begin{defn}
For the distribution $\mu$, we define a non-negative matrix $\bP \in \R_{\geq 0}^{\bX \times k}$ with respect to $\mu$ as follows:
\begin{equation}\label{eq:mp}
\Pxj\defeq \Pr_{\mx \sim \mu}(\Xxj=1)=\sum_{\{\mx \in \ka~|~\Xxj=1\}}\mu(\mx)~.
\end{equation}
\end{defn}
\begin{lemma}\label{lem:condip}
	The matrix $\bP$ defined in \Cref{eq:mp} satisfies the following conditions:
	\begin{equation}
	\sum_{j \in [1,k]}\Pxj=1 \text{ for all } x\in \bX \quad \text{ and } \quad  \sum_{x \in \bX}\Pxj=\phi_{j} \text{ for all } j \in [1,k]~.
	\end{equation}
	\end{lemma}
\begin{proof}
We first evaluate the row sum. For each $x \in \bX$,
\begin{align*}
\sum_{j \in [1,k]}\Pxj=\sum_{j \in [1,k]}\sum_{\{\mx \in \ka~|~\Xxj=1\}}\mu(\mx)=\sum_{\mx \in \ka}\mu(\mx)=1~.
\end{align*}
In the second inequality we used that $\mx \in \ka$, meaning for each $x \in \bX$, $\sum_{j \in [1,k]}\Xxj=1$. Next we evaluate the column sum, for each $j \in [1,k]$,
\begin{align*}
\sum_{x \in \bX}\Pxj&=\sum_{x \in \bX}\sum_{\{\mx \in \ka~|~\Xxj=1\}}\mu(\mx)=\sum_{\mx \in \ka}\sum_{\{x \in \bX~|~\Xxj=1\}}\mu(\mx)\\
&=\sum_{\mx \in \ka}\mu(\mx)\sum_{\{x \in \bX~|~\Xxj=1\}}1=\sum_{\mx \in \ka}\mu(\mx) \phi_{j}=\phi_{j}
\end{align*}
In the first equality we used the definition of $\Pxj$. In the second inequality we interchanged the summations. In the final equality we used $\sum_{\mx \in \ka}\mu(\mx)=1$.
\end{proof}
The matrix $\bP$ defined in \Cref{eq:mp} is important because we can upper bound the permanent of matrix $\ma$ in terms of entries of this matrix. We formalize this result in our next lemma.

\begin{lemma}\label{lem:upperm}
 For matrix $\ma \in \kds$, if $\bP$ is the matrix defined in \Cref{eq:mp}, then
$$\log \perm(\ma)	\leq  O\left( k\log \frac{\Ns}{k} \right)+ \fnp(\ma, \bP)$$
\end{lemma}
\begin{proof}
We first calculate the expectation of $\log(\mu(\mx))$ and express it in terms of matrix $\bP$.
\begin{equation}\label{eq:upone}
\begin{split}
\expt{\mx\sim \mu}{\log \mu(\mx)}&=\sum_{\mx \in \kr} \mu(\mx)\log \mu(\mx)=\sum_{\mx \in \ka} \mu(\mx)\log \mu(\mx)~,\\
&=\sum_{\mx \in \ka} \mu(\mx)\log\left( \left(\prod_{j \in [1,k]} \phi_{j}! \right) \left( \prod_{\{(x,j) \in \bX \times [1,k]~|~\Xxj=1\}} \madxj \right)\left(\frac{1}{\perm(\ma)}\right) \right)~,\\
&=\log \left(\prod_{j \in [1,k]} \phi_{j}! \right)-\log \perm(\ma) + \sum_{\mx \in \ka} \mu(\mx) \log \left( \prod_{\{(x,j) \in \bX \times [1,k]~|~\Xxj=1\}} \madxj \right)~.
\end{split}
\end{equation}
The second equality holds because the support of distribution $\mu$ is subset of $\ka$. In the third equality we used \Cref{eq:mux}.
We now simplify the last term in the final expression from the above derivation.
\begin{equation}\label{eq:uptwo}
\begin{split}
\sum_{\mx \in \ka} \mu(\mx) \log \left( \prod_{\{(x,j) \in \bX \times [1,k]~|~\Xxj=1\}} \madxj \right)&=\sum_{\mx \in \ka} \mu(\mx) \sum_{\{(x,j) \in \bX \times [1,k]~|~\Xxj=1\}}\log  \madxj~,\\
&=\sum_{(x,j) \in \bX \times [1,k]} \log  \madxj \sum_{\{\mx \in \ka~|~\Xxj=1 \}} \mu(\mx) ~,\\
&=\sum_{(x,j) \in \bX \times [1,k]} \Pxj\log  \madxj~. \\
\end{split}
\end{equation}

Combining \Cref{eq:upone} and \Cref{eq:uptwo} together we get,
\begin{equation}
\expt{\mx\sim \mu}{\log \mu(\mx)}=\log \left(\prod_{j \in [1,k]} \phi_{j}! \right)-\log \perm(\ma) +\sum_{(x,j) \in \bX \times [1,k]} \Pxj\log  \madxj~.
\end{equation}

We next define a different distribution $\nu$ on $\kr$ using the following sampling procedure: For each $x \in \bX$, pick a column $j \in [1,k]$ independently with probability $\Pxj$. Note that this is a valid sampling procedure because for each $x \in \bX$, $\sum_{j \in [1,k]}\Pxj=1$. The description of distribution $\nu$ is as follows: for each $\mx \in \kr$,
\begin{equation}
\nu(\mx)\defeq \prod_{\{(x,j) \in \bX \times [1,k]~|~\Xxj=1\}} \Pxj
\end{equation}
{\bf Remark}: Note that $\sum_{X \in \kr}\nu(\mx)=\prod_{x\in \bX} (\sum_{j \in [1,k]} \Pxj) =1$ and $\nu$ is a valid distribution. 

We next calculate the expectation of $\log(\nu(\mx))$ with respect to distribution $\mu$ and express it in terms of matrix $\bP$. Note that $\sum_{\mx \in \kr} \mu(\mx)\log \nu(\mx)=\sum_{\mx \in \ka} \mu(\mx)\log \nu(\mx)$ because $\mu(\mx)=0$ for all $\mx \in \kr \backslash \ka$ and we get,
\begin{align*}
\expt{\mx\sim \mu}{\log \nu(\mx)}&=\sum_{\mx \in \ka} \mu(\mx)\log \nu(\mx)
=\sum_{\mx \in \ka} \mu(\mx)\log  \left( \prod_{\{(x,j) \in \bX \times [1,k]~|~\Xxj=1\}} \Pxj \right)\\
&=\sum_{\mx \in \ka} \mu(\mx) \sum_{\{(x,j) \in \bX \times [1,k]|\Xxj=1\}}\log  \Pxj
=\sum_{(x,j) \in \bX \times [1,k]} \log  \Pxj \sum_{\{\mx \in \ka|\Xxj=1 \}} \mu(\mx) \\
&=\sum_{(x,j) \in \bX \times [1,k]} \Pxj\log  \Pxj \\
\end{align*}
We now calculate the KL divergence $\KL{\mu\|\nu}$ between distributions $\mu$ and $\nu$.
\begin{align*}
\KL{\mu\|\nu}&=\expt{\mx\sim \mu}{\log \mu(\mx)}-\expt{\mx\sim \mu}{\log \nu(\mx)}\\
&=\log \left(\prod_{j \in [1,k]} \phi_{j}! \right)-\log \perm(\ma) +\sum_{(x,j) \in \bX \times [1,k]} \Pxj\log  \madxj-\sum_{(x,j) \in \bX \times [1,k]} \Pxj\log  \Pxj
\end{align*}
Using \Cref{lem:kldiv}, we have $\KL{\mu\|\nu} \geq 0$, that further implies,
\begin{equation}\label{perm:up}
\begin{split}
\log \perm(\ma) &\leq \log \left(\prod_{j \in [1,k]} \phi_{j}! \right) +\sum_{(x,j) \in \bX \times [1,k]} \Pxj\log  \frac{\madxj}{\Pxj}\\
&\leq \sum_{j \in [1,k]}O(\log \phi_{j})+ \sum_{j\in[1,k]}\phi_{j}\log \phi_{j} -\sum_{j\in[1,k]}\phi_{j} +\sum_{(x,j) \in \bX \times [1,k]} \Pxj\log  \frac{\madxj}{\Pxj}\\
&\leq O(k \log \frac{\Ns}{k})+ \sum_{j\in[1,k]}\phi_{j}\log \phi_{j} -\sum_{j\in[1,k]}\phi_{j} +\sum_{(x,j) \in \bX \times [1,k]} \Pxj\log  \frac{\madxj}{\Pxj}
\end{split}
\end{equation}
In the second inequality we used \Cref{lem:ster} on each $\phi_{j}$ and further in the third inequality we used $\sum_{j \in [1,k]}\phi_{j}= \Ns$ and the fact that the function $\sum_{j \in [1,k]}\log \phi_{j}$ is always upper bounded by $O(k \log \frac{\Ns}{k})$. Further using the definition of $\fnp(\ma,\bP)$ (See \Cref{eq:fnp}), we conclude the proof.
\end{proof}

We provided an upper bound to the permanent of matrix $\ma$ and all that remains is to relate this upper bound to the scaled Sinkhorn permanent of matrix $\ma$. Our next lemma will serve this purpose.

\begin{lemma}\label{thm:ds}
	For any matrix $\bP \in \R_{\geq 0}^{\bX \times [1,k]}$ that satisfies,
	\begin{equation}
	\sum_{j \in [1,k]}\Pxj=1 \text{ for all } x\in \bX \quad \text{ and } \quad  \sum_{x \in \bX}\Pxj=\phi_{j} \text{ for all } j \in [1,k]~.
	\end{equation}
	there exists a doubly stochastic matrix $\bQ \in \R_{\geq 0}^{\bX \times \bX}$ such that,
	\begin{equation}
	\fnp(\ma, \bP) = \frst(\ma,\bQ)-\Ns~.
	\end{equation}
\end{lemma}
\begin{proof}
	Define matrix $\bQ \in \R^{\bX \times \bX}$ as follows,
	$$\bQ_{x,y}\defeq \frac{\Pxj}{\phi_{j}}$$
	where in the definition above $j$ is such that $\may=\cj$. Now we verify the row and column sums of matrix $\bQ$. For each $x \in \bX$,
\begin{equation}\label{eq:rowsum}
\begin{split}
\sum_{y \in \bX}\bQ_{x,y}&= \sum_{j\in [1,k]} \sum_{\{y \in \bX~|~\may=\cj\}}\frac{\Pxj}{ \phi_{j}}=\sum_{j \in [1,k]}\frac{\Pxj}{\phi_{j}} \sum_{\{y \in \bX~|~\may=\cj\}}1\\
&=\sum_{j \in [1,k]} \frac{\Pxj}{\phi_{j}} \cdot \phi_{j}=\sum_{j \in [1,k]} \Pxj=1\\
\end{split}
\end{equation}
We next evaluate the column sums. For each $y\in \bX$, let $j$ \footnote{Note that $j$ is a function of $y$. For convenience, in our notation we don't capture its dependence on $y$.} be such that $\may=\cj$, then
	\begin{equation}
	\sum_{x \in \bX}\bQ_{x,y}=\sum_{x \in \bX}\frac{\Pxj}{\phi_{j}}=\frac{1}{\phi_{j}}\sum_{x \in \bX}\Pxj=\frac{1}{\phi_{j}}\phi_{j}=1~.
	\end{equation}
	Therefore the matrix $\bQ$ is doubly stochastic and we next relate $\frst(\ma,\bQ)$ with $\fnp(\ma, \bP)$. Recall the definition of $\frst(\ma,\bQ)$ (\Cref{eq:ftsd}), 
	\begin{equation}\label{eq:zzz}
	\frst(\ma,\bQ)=\sum_{(x,y) \in \bX \times \bX} \bQ_{x,y}\log (\frac{{\ma}_{x,y}}{\bQ_{x,y}})~.
	\end{equation}
	We analyze the above term and express it in terms of entries of matrices $\bP$ and $\mad$.
\begin{equation}\label{eq:ooo}
\begin{split}
\sum_{(x,y) \in \bX \times \bX} \bQ_{x,y}\log (\frac{{\ma}_{x,y}}{\bQ_{x,y}})&=\sum_{x \in \bX} \sum_{j \in [1,k]}\left[ \sum_{\{y \in \bX~|~\may=\cj\}} \bQ_{x,y}\log (\frac{{\ma}_{x,y}}{\bQ_{x,y}})\right] \\
&=\sum_{x \in \bX} \sum_{j \in [1,k]}\left[ \sum_{\{y \in \bX~|~\may=\cj\}} \frac{\Pxj}{  \phi_{j}}\log (\frac{\madxj  \phi_{j}}{\Pxj})\right] \\
&=\sum_{x \in \bX} \sum_{j \in [1,k]}\left[  \phi_{j} \cdot \frac{\Pxj}{  \phi_{j}}\log (\frac{\madxj  \phi_{j}}{\Pxj})\right] 
= \sum_{x \in \bX} \sum_{j \in [1,k]}\left[\Pxj\log (\frac{\madxj  \phi_{j}}{\Pxj})\right]
\end{split}
\end{equation}
The first equality follows because $\cj$ for all $j\in [1,k]$ are distinct. The second equality follows because for each $x\in \bX$, consider any $y\in \bX$ such that $\may=\cj$ and note that for all such $y$'s, $\ma_{x,y}=\madxj$ and $\bQ_{x,y}=\frac{\Pxj}{  \phi_{j}}$. The third equality follows because $\sum_{\{y \in \bX~|~\may=\cj\}}1=|\{y \in \bX~|~\may=\cj\}|=\phi_j$. 

We further simplify the final term in the above derivation.
\begin{equation}\label{eq:ttt}
\begin{split}
\sum_{x \in \bX} \sum_{j \in [1,k]}\left[\Pxj\log (\frac{\madxj  \phi_{j}}{\Pxj})\right]&= \sum_{x \in \bX} \sum_{j \in [1,k]}\left[\Pxj\log (\frac{\madxj}{\Pxj})\right]+ \sum_{x \in \bX} \sum_{j \in [1,k]}\Pxj\log \phi_{j}\\
&= \sum_{x \in \bX} \sum_{j \in [1,k]}\left[\Pxj\log (\frac{\madxj}{\Pxj})\right]+ \sum_{j \in [1,k]}\log \phi_{j} \sum_{x \in \bX} \Pxj\\
&= \sum_{x \in \bX} \sum_{j \in [1,k]}\left[\Pxj\log (\frac{\madxj}{\Pxj})\right]+ \sum_{j \in [1,k]}\phi_{j}\log \phi_{j}~.
\end{split}
\end{equation}
Combining \Cref{eq:ooo}, \Cref{eq:ttt} and further substituting back in \Cref{eq:zzz} we get,
\begin{equation}
\begin{split}
\frst(\ma,\bQ) &=\sum_{x \in \bX} \sum_{j \in [1,k]}\left[\Pxj\log (\frac{\madxj}{\Pxj})\right]+ \sum_{j \in [1,k]}\phi_{j}\log \phi_{j}\\
& = \fnp(\ma,\bQ) + \Ns~.
\end{split}
\end{equation}
In the final expression, we used the definition of $\fnp(\ma,\bQ)$ and combined it with $\Ns=\sum_{j\in[1,k]}\phi_j$.
\end{proof}

We are now ready to prove our main lemma of this section and is restated for convenience.
\mainthm*
\begin{proof}
Consider the matrix $\bP$ defined in \Cref{eq:mp}. By \Cref{lem:condip}, matrix $\bP$ satisfies the conditions of \Cref{thm:ds}; therefore, there exists a doubly stochastic matrix $\bQ \in \dstoc$ such that $\fnp(\ma,\bP)=\frst(\ma,\bQ)-\Ns$. Combining it with \Cref{lem:upperm} we get $\log \perm(\ma)	\leq  O( k\log \frac{\Ns}{k} )+ \frst(\ma, \bQ) -\Ns$, which further implies $\perm(\ma) \leq \exps{O( k\log \frac{\Ns}{k} )} \ssink(\ma)$. The lower bound for the $\perm(\ma)$ follows from \Cref{cor:ssinklb} and we conclude the proof.
\end{proof}
We next state another interesting property of the matrix $\bP$ defined in \Cref{eq:mp}. This property will be useful for the purposes of PML (\Cref{sec:pml}). 
\begin{theorem}\label{thm:sameprob}
For matrix $\ma \in \kds$, the matrix $\bP$ defined in \Cref{eq:mp} satisfies the following: 
$\text{If }x,y \in \bX\text{ are such that }\maxt=\mayt\text{ then, for all }j\in [1,k] \text{ we have }\Pxj=\Pyj~.$
\end{theorem}
\begin{proof}
For any $j\in [1,k]$, recall by the definitions of terms $\Pxj$ and $\Pyj$,
\begin{equation}\label{eq:pxj}
\begin{split}
\Pxj&=\sum_{\{\mx \in \ka~|~\Xxj=1\}}\left(\prod_{j' \in [1,k]} \phi_{j'}! \right) \left( \prod_{(z,j') \in \bX \times [1,k]} \madzjp^{\Xzjp} \right)\left(\frac{1}{\perm(\ma)}\right),\\
&=\left(\prod_{j' \in [1,k]} \phi_{j'}! \right)\left(\frac{1}{\perm(\ma)}\right)\sum_{\{\mx \in \ka~|~\Xxj=1\}}  \prod_{(z,j') \in \bX \times [1,k]} \madzjp^{\Xzjp}~.
\end{split}
\end{equation}
\begin{equation}\label{eq:pyj}
\begin{split}
\Pyj=\left(\prod_{j' \in [1,k]} \phi_{j'}! \right)\left(\frac{1}{\perm(\ma)}\right)\sum_{\{\mx' \in \ka~|~\Xyj'=1\}}  \prod_{(z,j') \in \bX \times [1,k]} \madzjp^{\Xzjp'}~.
\end{split}
\end{equation}
For any $\mY \in  \{\mx \in \ka~|~\Xxj=1\}$ we next construct a unique $\mY' \in \{\mx' \in \ka~|~\Xyj'=1\}$ (and vice versa) such that,
$$ \prod_{(z,j') \in \bX \times [1,k]} \madzjp^{\Yzjp}=\prod_{(z,j') \in \bX \times [1,k]} \madzjp^{\Yzjp'}$$
Each $\mY \in \ka$ corresponds to a bipartite graph where vertices correspond to set $\bX$ on left side and $[1,k]$ on the other, such that, degree of every left vertex $x \in \bX$ is 1 and degree of every right vertex $j \in [1,k]$ is $\phi_j$.

Consider $\mY \in  \{\mx \in \ka~|~\Xxj=1\}$, we divide the analysis into the following two cases,
\begin{enumerate}
\item If $\mY_{y,j}=1$, meaning both vertices $x,y \in \bX$ are connected to $j \in [1,k]$ in our bipartite graph representation. Then, $\mY'\defeq\mY$.
\item If $\mY_{y,j}=0$, meaning vertex $x$ is connected to $j$ and $y$ to some other vertex $j'\neq j$. In this case we swap the edges, meaning we remove edges $(x,j), (y,j')$  and add $(x,j'),(y,j)$ to construct $\mY'$. We formally define $\mY'$ next,
\begin{equation}
\mY'_{z,j''}\defeq 
\begin{cases}
1 \text{ if }z=y \text{ and }j''=j,\\
0 \text{ if }z=y \text{ and }j''=j',\\
1 \text{ if }z=x \text{ and }j''=j',\\
0 \text{ if }z=x \text{ and }j''=j,\\
\mY_{z,j'} \text{ otherwise}~.
\end{cases}
\end{equation}	
\end{enumerate}
In both cases, clearly $\mY' \in \{\mx' \in \ka~|~\Xyj'=1\}$. Further, $\maxt=\mayt$ implies $\mad_{x,j'}=\mad_{y,j}$ for all $j'\in [1,k]$ and the following equality holds,
$$ \prod_{(z,j') \in \bX \times [1,k]} \madzjp^{\Yzjp}=\prod_{(z,j') \in \bX \times [1,k]} \madzjp^{\Yzjp'}$$
The same analysis also holds when we start with a $\mY' \in \{\mx' \in \ka~|~\Xyj'=1\}$ and construct $\mY \in  \{\mx \in \ka~|~\Xxj=1\}$. We have a one to one correspondence between elements $\mY$ and $\mY'$ in the sets $\{\mx \in \ka~|~\Xxj=1\}$ and $\{\mx' \in \ka~|~\Xyj'=1\}$ respectively, satisfying,
$$ \prod_{(z,j') \in \bX \times [1,k]} \madzjp^{\Yzjp}=\prod_{(z,j') \in \bX \times [1,k]} \madzjp^{\Yzjp'}~.$$
Therefore, $\Pxj=\Pyj$ and we conclude the proof.
\end{proof}
\subsection{Generalization to low non-negative rank matrices}\label{sec:lowrank}
Here we prove our main result for the scaled Sinkhorn permanent of low non-negative rank matrices (\Cref{thm:mainlowrank}). To prove this result, we use the performance result of the scaled Sinkhorn permanent for non-negative matrices with $k$ distinct columns. The following lemma relates the permanent of a matrix $\ma$ of non-negative rank $k$ to matrices with at most $k$ distinct columns and will be crucial for our analysis. 

\newcommand{\va}{\textbf{V}^{\alpha}}
\newcommand{\ua}{\textbf{U}^{\alpha}}
\newcommand{\uvec}{\textbf{u}}
\begin{lemma}[\cite{Bar96}]
	For any matrix $\ma \in \R_{\geq 0}^{\bX \times \bX}$ of non-negative rank $k$. If $\ma \defeq \sum_{j \in [k]} \vj \uj^{\top}$ for $\vj,\uj \in \R_{\geq 0}^{\bX}$, then 
	$$\perm(\ma)= \sum_{\{\alpha \subseteq \Z^{k}|\sum_{j \in [k]}\alpha_{j}=N\}}  \frac{1}{\prod_{j \in [k]}\alpha_{j}!}\perm(\va)\perm(\ua),$$
	where $\va \defeq [\underbrace{\vvec_{1} \dots \vvec_{1}}_{\alpha_{1}}~|~\underbrace{\vvec_{2} \dots \vvec_{2}}_{\alpha_{2}} ~|~\dots~| \underbrace{\vvec_{k} \dots \vvec_{k}}_{\alpha_{k}}]$, $\ua \defeq [\underbrace{\uvec_{1} \dots \uvec_{1}}_{\alpha_{1}}~|~\underbrace{\uvec_{2} \dots \uvec_{2}}_{\alpha_{2}} ~|~\dots~| \underbrace{\uvec_{k} \dots \uvec_{k}}_{\alpha_{k}}]$.
	\end{lemma}
As the number of terms in the above summation is low, the maximizing term is a good approximation to the permanent of $\ma$.
\begin{cor}
	Given a non-negative matrix $\ma \in \R_{\geq 0}^{\bX \times \bX}$, let $k$ denote the non-negative rank of the matrix. If $\ma = \sum_{j \in [k]} \vj \uj^{\top}$ for $\vj,\uj \in \R_{\geq 0}^{\bX}$ is any non-negative matrix factorization of $\ma$, then 
	\begin{equation}\label{eq:maxalpha}
		\perm(\ma) \leq  \expo{O(k\log \frac{N}{k})} \max_{\{\alpha \subseteq \Z^{k}|\sum_{j \in [k]}\alpha_{j}=N\}}  \frac{1}{\prod_{j \in [k]}\alpha_{j}!}\perm(\va)\perm(\ua)~.
	\end{equation}
\end{cor}
\begin{proof}
	The number of feasible $\alpha$'s in the set $\{\alpha \subseteq \Z^{k}|\sum_{j \in [k]}\alpha_{j}=N\}$ is at most $\binom{N+k-1}{k-1} \in \expo{O(k\log \frac{N}{k})}$ and we conclude the proof.
\end{proof}
\newcommand{\qv}{\textbf{\bQ}'}
\newcommand{\qu}{\textbf{\bQ}''}
\newcommand{\pv}{\textbf{\bP}'}
\newcommand{\pu}{\textbf{\bP}''}
\renewcommand{\1}{\textbf{1}}
\newcommand{\bjxy}{\beta^{j}_{x,y}}
\begin{lemma}\label{lem:stoc}
	Let $\qv,\qu \in \kds$ be any doubly stochastic matrices. Then $\bQ \defeq \qv \qu$ is a doubly stochastic matrix.
	\end{lemma}
\begin{proof}
	We first consider the row sums,
	$$\bQ \1=\qv \qu \1=\qv \1=\1~.$$
	Therefore the matrix $\bQ$ is row stochastic. In the above derivation, the second and third equalities follow because $\qu$ and $\qv$ are row stochastic matrices respectively. We now consider the column sums, 
	$$\bQ^{\top} \1=\qu^{\top} \qv^{\top} \1=\qu^{\top} \1=\1~.$$
	The above derivation follows because $\qv$ and $\qu$ are column stochastic and therefore the matrix $\bQ$ is column stochastic. As the matrix $\bQ$ is both row and column stochastic we conclude the proof.
\end{proof}
We are now ready to prove our main result of this section and we restate it for convenience.
\thmmainlowrank*
\begin{proof}
	Let $\alpha$ be the maximizer of the optimization problem \ref{eq:maxalpha}, then
	\begin{equation}\label{eq:xone}
	\perm(\ma) \leq  \expo{O(k\log \frac{N}{k})}  \frac{1}{\prod_{j \in [k]}\alpha_{j}!}\perm(\va)\perm(\ua)~.
	\end{equation}
	Recall to prove the theorem, we need to construct a doubly stochastic witness $\bQ$ that satisfies:
	$$\log \perm(\ma) \leq {O(k\log \frac{N}{k})} + \frst(\ma,\bQ)-N~.$$
	We construct such a witness $\bQ$ from the doubly stochastic witnesses for matrices $\va$ and $\ua$. For all $j \in [k]$ define $S_{j}\defeq \{y \in \bX~|~\va_{:y}=\vj\}$, equivalently $S_{j}= \{y \in \bX~|~\ua_{:y}=\uj\}$ and note that $\alpha_j=|S_j|$. Let $\qv$ and $\qu$ be the doubly stochastic matrices that maximize the scaled Sinkhorn permanent for matrices $\va$ and $\ua$ respectively. Therefore by \Cref{thm:main} the following inequalities hold,
	\begin{equation}\label{eq:xtwo}
	\log \perm(\va) \leq O(k\log \frac{N}{k})+\frst(\va,\qv)-N~,
	\end{equation}
	\begin{equation}\label{eq:xth}
	\log \perm(\ua) \leq O(k\log \frac{N}{k})+\frst(\ua,\qu)-N~,
	\end{equation}
	where recall $\frst(\va,\qv)=\sum_{x,y \in \bX \times \bX}\qv_{x,y} \log \frac{\va_{x,y}}{\qv_{x,y}}$ and $\frst(\ua,\qu)=\sum_{x,y \in \bX \times \bX}\qu_{x,y} \log \frac{\ua_{x,y}}{\qu_{x,y}}$. Without loss of generality by the symmetry (with respect to columns within $S_{j}$) and concavity of the scaled Sinkhorn objective, we can assume that the maximizing matrices $\qv$ and $\qu$ satisfy the following: for all $x\in \bX$ and $j \in [k]$,
	\begin{equation}\label{eq:symmetry}
	\qv_{x,y}=\qv_{x,y'} \text{ and } \qu_{x,y}=\qu_{x,y'} \text{ for all }y,y' \in S_{j} \text{ and }x\in \bX~.
	\end{equation}
	Note that the doubly stochastic matrix that we constructed for the proof of \Cref{thm:main} also satisfies the above collection of equalities. Now combining \Cref{eq:xone,eq:xtwo,eq:xth}  we get,
	\begin{equation}\label{eq:task}
	\begin{split}
	\log \perm(\ma) & \leq  {O(k\log \frac{N}{k})}  - \log \prod_{j \in [k]} \alpha_{j}! + \frst(\va,\qv)-N+ \frst(\ua,\qu)-N~,\\
	& \leq {O(k\log \frac{N}{k})} - \sum_{j \in [k]}\left(\alpha_{j}\log\alpha_{j} - \alpha_{j}  \right)+ \frst(\va,\qv)-N+ \frst(\ua,\qu)-N~,\\
	& \leq {O(k\log \frac{N}{k})} - \sum_{j \in [k]}\alpha_{j}\log\alpha_{j} + \frst(\va,\qv)+ \frst(\ua,\qu)-N~.
	\end{split}
	\end{equation}
	In the second inequality we use the Stirling's approximation (\Cref{lem:ster}) and the error term due to this approximation is upper bounded by $O(k\log \frac{N}{k})$. In the third inequality we used $\sum_{j \in [k]}\alpha_{j}=N$.
	
	Let $\bQ=\qv \qu^{\top}$, then by \Cref{lem:stoc} the matrix $\bQ$ is doubly stochastic. In the remainder of the proof we show that,
	\begin{equation}\label{eq:mainhere}
	- \sum_{j \in [k]}\alpha_{j}\log\alpha_{j} + \frst(\va,\qv)+ \frst(\ua,\qu) \leq \frst(\ma,\bQ)~,
	\end{equation}
	where recall $\frst(\ma,\bQ)=\sum_{x,y \in \bX \times \bX} \bQ_{x,y} \log \frac{\ma_{x,y}}{\bQ_{x,y}}$. As matrix $\bQ$ is doubly stochastic, the above inequality combined with \Cref{eq:task} concludes the proof. Therefore in the remainder we focus our attention to prove \Cref{eq:mainhere} and we start by simplifying the above expression. Define,
	\begin{equation}\label{eq:betadef}
	\bjxy\defeq \frac{1}{\bQ_{x,y}}\sum_{z \in S_j} \qv_{x,z} \qu_{y,z} \quad \text{ for all }x\in \bX, y\in \bX \text{ and for all }j\in [k]~.
	\end{equation}
	For all $x\in \bX$ and $y\in \bX$ the variables defined above satisfy the following,
	\begin{equation}\label{eq:betastoc}
	\sum_{j \in [k]}\bjxy=\frac{1}{\bQ_{x,y}}\sum_{j \in [k]}\sum_{z \in S_j} \qv_{x,z} \qu_{y,z}=\frac{1}{\bQ_{x,y}}\sum_{z \in \bX} \qv_{x,z} \qu_{y,z}=\frac{1}{\bQ_{x,y}}\bQ_{x,y}=1~,
	\end{equation}
	where in the third inequality we used the definition of $\bQ=\qv \qu^{\top}$. We next simplify and lower bound the term $\frst(\ma,\bQ)$ in terms of these newly defined variables. 
	\begin{align}
	\log \ma_{x,y}=\log \sum_{j\in [k]}\vj(x)\uj(y)\geq \log \prod_{j \in [k]} \left(\frac{\vj(x)\uj(y)}{\bjxy}\right)^{\bjxy}= \sum_{j \in [k]}{\bjxy} \log \left(\frac{\vj(x)\uj(y)}{\bjxy}\right)~,
	\end{align}
	where in the first equality we used $\ma = \sum_{j \in [k]} \vj \uj^{\top}$. In the second inequality we used weighted AM-GM inequality. Now consider the term $\bQ_{x,y}\log \ma_{x,y}$ and substitute the above lower bound,
	\begin{align}
	\bQ_{x,y}\log \ma_{x,y}&\geq \bQ_{x,y} \sum_{j \in [k]}{\bjxy} (\log \vj(x) + \log \uj(y)) - \bQ_{x,y}\sum_{j \in [k]}{\bjxy}\log {\bjxy}~.
	\end{align}
	Summing over all the $(x,y)$ pairs we get,
	\begin{equation}\label{eq:kkmain}
	\begin{split}
	\sum_{x,y \in \bX \times \bX}\bQ_{x,y}\log \ma_{x,y}& \geq\sum_{x \in \bX}\sum_{j \in [k]} \log \vj(x) \big(\sum_{y \in \bX} \bQ_{x,y} \bjxy \big) + \sum_{y \in \bX}\sum_{j \in [k]} \log \vj(y) \big(\sum_{x \in \bX} \bQ_{x,y} \bjxy \big)~,\\
	& \quad - \sum_{x,y \in \bX \times \bX} \bQ_{x,y}\sum_{j \in [k]}{\bjxy}\log {\bjxy}~.
		\end{split}
	\end{equation}	
	In the above expression the following terms simplify, 
	\begin{equation}\label{eq:kkone}
	\sum_{y \in \bX} \bQ_{x,y} \bjxy= \sum_{y \in \bX} \bQ_{x,y} \frac{1}{\bQ_{x,y}}\sum_{z \in S_j} \qv_{x,z} \qu_{y,z}=\sum_{z \in S_j} \qv_{x,z} \sum_{y \in \bX}  \qu_{y,z}=\sum_{z \in S_j} \qv_{x,z}~.
	\end{equation}
	Similarly,
	\begin{equation}\label{eq:kktwo}
	\sum_{x \in \bX} \bQ_{x,y} \bjxy= \sum_{z \in S_j} \qu_{y,z}~.
	\end{equation}
	Also note that,
	\begin{equation}\label{eq:kkthree}
	\begin{split}
	\sum_{x,y \in \bX \times \bX} \bQ_{x,y}\sum_{j \in [k]}{\bjxy}\log {\bjxy} & = \sum_{x,y \in \bX \times \bX} \bQ_{x,y}\sum_{j \in [k]}{\bjxy}\log \frac{\bjxy\bQ_{x,y}}{\bQ_{x,y}},\\
	& =\sum_{x,y \in \bX \times \bX} \bQ_{x,y}\sum_{j \in [k]}{\bjxy}\log (\bjxy\bQ_{x,y}) - \sum_{x,y \in \bX \times \bX} \bQ_{x,y}\sum_{j \in [k]}{\bjxy} \log {\bQ_{x,y}},\\
	& =\sum_{x,y \in \bX \times \bX} \sum_{j \in [k]}{\bjxy}\bQ_{x,y}\log (\bjxy\bQ_{x,y}) - \sum_{x,y \in \bX \times \bX} \bQ_{x,y} \log {\bQ_{x,y}}~,\\
	& =\sum_{x,y \in \bX \times \bX} \sum_{j \in [k]}\big(\sum_{z \in S_j} \qv_{x,z} \qu_{y,z} \big) \log \big( \sum_{z \in S_j} \qv_{x,z} \qu_{y,z} \big) - \sum_{x,y \in \bX \times \bX} \bQ_{x,y} \log {\bQ_{x,y}}~.\\
	\end{split}
	\end{equation}
	In the third and fourth inequality we used \Cref{eq:betastoc} and \Cref{eq:betadef} respectively. Substituting \Cref{eq:kkone,eq:kktwo,eq:kkthree} in \Cref{eq:kkmain} we get,
	\begin{equation}
	\begin{split}
	\sum_{x,y \in \bX \times \bX}\bQ_{x,y}\log \ma_{x,y}& \geq\sum_{x \in \bX}\sum_{j \in [k]} \log \vj(x) \big(\sum_{z \in S_j} \qv_{x,z}\big) + \sum_{y \in \bX}\sum_{j \in [k]} \log \vj(y) \big(\sum_{z \in S_j} \qu_{y,z} \big)\\
	& \quad - \sum_{x,y \in \bX \times \bX} \sum_{j \in [k]}\big(\sum_{z \in S_j} \qv_{x,z} \qu_{y,z} \big) \log \big( \sum_{z \in S_j} \qv_{x,z} \qu_{y,z} \big) + \sum_{x,y \in \bX \times \bX} \bQ_{x,y} \log {\bQ_{x,y}}~.
	\end{split}
	\end{equation}	
	By rearranging terms the above expression can be equivalently written as,
	\begin{equation}\label{eq:inter}
	\begin{split}
	\frst(\ma,\bQ)=\sum_{x,y \in \bX \times \bX}\bQ_{x,y}\log \frac{\ma_{x,y}}{\bQ_{x,y}}& \geq\sum_{x \in \bX}\sum_{j \in [k]} \log \vj(x) \big(\sum_{z \in S_j} \qv_{x,z}\big) + \sum_{y \in \bX}\sum_{j \in [k]} \log \uj(y) \big(\sum_{z \in S_j} \qu_{y,z} \big)\\
	& \quad - \sum_{x,y \in \bX \times \bX}\sum_{j \in [k]} \big(\sum_{z \in S_j} \qv_{x,z} \qu_{y,z} \big) \log \big( \sum_{z \in S_j} \qv_{x,z} \qu_{y,z} \big)~.
	\end{split}
	\end{equation}	
	In the above expression we have a lower bound for the term $\frst(\ma,\bQ)$ and we relate it to terms $\frst(\va,\qv)$ and $\frst(\ua,\qu)$. Consider the following term,
	\begin{equation}\label{eq:ione}
	\begin{split}	 
	\sum_{x,y \in \bX \times \bX}\qv_{x,y} \log \va_{x,y} &=\sum_{x \in \bX} \sum_{j \in [k]} \sum_{y \in S_{j}} \qv_{x,y}\log \va_{x,y} = \sum_{x \in \bX} \sum_{j \in [k]} \sum_{y \in S_{j}} \qv_{x,y}\log \vj(x)~,\\
	& = \sum_{x \in \bX} \sum_{j \in [k]} \log \vj(x) \big(\sum_{y \in S_{j}} \qv_{x,y}\big)=\sum_{x \in \bX}\sum_{j \in [k]} \log \vj(x) \big(\sum_{z \in S_j} \qv_{x,z}\big)~,\\
	\end{split}
	\end{equation}
	In the final equality we renamed the variables and the rest of equalities are straightforward. Carrying out similar derivation we also get,
	\begin{equation}\label{eq:itwo}
	\sum_{x,y \in \bX \times \bX}\qu_{x,y} \log \ua_{x,y} = \sum_{x \in \bX} \sum_{j \in [k]} \log \uj(x) \big(\sum_{y \in S_{j}} \qu_{x,y}\big)=\sum_{y \in \bX}\sum_{j \in [k]} \log \uj(y) \big(\sum_{z \in S_j} \qu_{y,z} \big)~.
	\end{equation}
	As before in the final equality we renamed variables. Substituting \Cref{eq:ione,eq:itwo} in \Cref{eq:inter} we get,
	\begin{equation}\label{eq:internew}
	\begin{split}
	\frst(\ma,\bQ) & \geq \sum_{x,y \in \bX \times \bX}\qv_{x,y} \log \va_{x,y} + \sum_{x,y \in \bX \times \bX}\qu_{x,y} \log \ua_{x,y} 
	- \sum_{x,y \in \bX \times \bX} \big(\sum_{z \in S_j} \qv_{x,z} \qu_{y,z} \big) \log \big( \sum_{z \in S_j} \qv_{x,z} \qu_{y,z} \big)\\
	& = \frst(\va,\qv) + \frst(\ua,\qu) +\sum_{x,y \in \bX \times \bX} \qv_{x,y} \log \qv_{x,y}  + \sum_{x,y \in \bX \times \bX} \qu_{x,y} \log \qu_{x,y}\\ 
	& \quad - \sum_{x,y \in \bX \times \bX} \sum_{j \in [k]}\big(\sum_{z \in S_j} \qv_{x,z} \qu_{y,z} \big) \log \big( \sum_{z \in S_j} \qv_{x,z} \qu_{y,z} \big).
	\end{split}
	\end{equation}	
	\newcommand{\rv}{\textbf{R}'}
	\newcommand{\ru}{\textbf{R}''}
	Therefore to prove \Cref{eq:mainhere}, all that remains is to show that,
	\begin{equation}\label{eq:alphamain}
		\sum_{x,y \in \bX \times \bX} \big(\qv_{x,y} \log \qv_{x,y}  + \qu_{x,y} \log \qu_{x,y}\big)  - \sum_{x,y \in \bX \times \bX} \big(\sum_{z \in S_j} \qv_{x,z} \qu_{y,z} \big) \log \big( \sum_{z \in S_j} \qv_{x,z} \qu_{y,z} \big) \geq - \sum_{j \in [k]}\alpha_{j}\log\alpha_{j}~.
	\end{equation}
	To prove the above inequality we use the symmetry in the solutions $\qv$ and $\qu$. Recall from \Cref{eq:symmetry}, for all $x\in \bX$ and $j \in [k]$ we have $\qv_{x,y}=\qv_{x,y'} \text{ and } \qu_{x,y}=\qu_{x,y'} \text{ for all }y,y' \in S_{j} \text{ and }x\in \bX$. Define $\rv_{x,j}=\qv_{x,y}$ and $\ru_{x,j}=\qu_{x,y}$ for any $y \in S_j$. We next substitute these definitions and simplify terms on the left hand side of \Cref{eq:alphamain},
	\begin{equation}\label{eq:hhone}
	\begin{split}
	\sum_{x,y \in \bX \times \bX} \qv_{x,y} \log \qv_{x,y}& =\sum_{x \in \bX}\sum_{j \in [k]}\sum_{y \in S_j} \qv_{x,y} \log \qv_{x,y}=\sum_{x \in \bX}\sum_{j \in [k]}\sum_{y \in S_j} \rv_{x,j} \log \rv_{x,j},\\
	&=\sum_{x \in \bX}\sum_{j \in [k]}\alpha_j \rv_{x,j} \log \rv_{x,j}~.
	\end{split}
	\end{equation}
	In the final equality we used $|S_j|=\alpha_j$ and the rest of the equalities are straightforward. Similar argument as above also gets us,
	\begin{equation}\label{eq:hhtwo}
	\begin{split}
	\sum_{x,y \in \bX \times \bX} \qu_{x,y} \log \qu_{x,y}& =\sum_{x \in \bX}\sum_{j \in [k]}\alpha_j \ru_{x,j} \log \ru_{x,j}=\sum_{y \in \bX}\sum_{j \in [k]}\alpha_j \ru_{y,j} \log \ru_{y,j}~.
	\end{split}
	\end{equation}
	Note in the final equality we renamed variables. Finally, 
	\begin{equation}\label{eq:hhthree}
	\begin{split}
	\sum_{x,y \in \bX \times \bX} \sum_{j \in [k]}\big(\sum_{z \in S_j} \qv_{x,z} \qu_{y,z} \big)&  \log \big( \sum_{z \in S_j} \qv_{x,z} \qu_{y,z} \big) = \sum_{x,y \in \bX \times \bX} \sum_{j \in [k]} \alpha_j \rv_{x,j} \ru_{y,j}\log \alpha_j \rv_{x,j} \ru_{y,j}~,\\
	& = \sum_{x,y \in \bX \times \bX} \sum_{j \in [k]}\alpha_j \rv_{x,j} \ru_{y,j} \big(\log \alpha_j + \log \rv_{x,j}+\log  \ru_{y,j} \big)~,\\
	\end{split}
	\end{equation}
	Again each of the terms in the parenthesis further simplify as follows,
	\begin{align*}
	\sum_{x,y \in \bX \times \bX} \sum_{j \in [k]}\alpha_j \rv_{x,j} \ru_{y,j} \log \alpha_j&=\sum_{j \in [k]}\alpha_j  \log \alpha_j \sum_{x,y \in \bX \times \bX}\rv_{x,j} \ru_{y,j}=\sum_{j \in [k]}\alpha_j  \log \alpha_j \sum_{x\in \bX} \rv_{x,j}\sum_{y\in \bX} \ru_{y,j},\\
	&=\sum_{j \in [k]}\alpha_j  \log \alpha_j~.
	\end{align*}
	\begin{align*}
	\sum_{x,y \in \bX \times \bX} \sum_{j \in [k]}\alpha_j \rv_{x,j} \ru_{y,j} \log \rv_{x,j}=  \sum_{x \in \bX} \sum_{j \in [k]}\alpha_j \rv_{x,j} \log \rv_{x,j} \sum_{y \in \bX}\ru_{y,j} =\sum_{x \in \bX} \sum_{j \in [k]}\alpha_j \rv_{x,j} \log \rv_{x,j}~.
	\end{align*}
	Similarly,
	\begin{align*}
	\sum_{x,y \in \bX \times \bX} \sum_{j \in [k]}\alpha_j \rv_{x,j} \ru_{y,j} \log \ru_{y,j}=\sum_{y \in \bX} \sum_{j \in [k]}\alpha_j \ru_{y,j} \log \ru_{y,j}~.
	\end{align*}
	Substituting back all the above three expressions in \Cref{eq:hhthree} we get,
	\begin{equation}\label{eq:hhthreenew}
	\begin{split}
	\sum_{x,y \in \bX \times \bX} \sum_{j \in [k]}\big(\sum_{z \in S_j} \qv_{x,z} \qu_{y,z} \big)  \log \big( \sum_{z \in S_j} \qv_{x,z} \qu_{y,z} \big)& =\sum_{x \in \bX} \sum_{j \in [k]}\alpha_j \rv_{x,j} \log \rv_{x,j}+ \sum_{y \in \bX} \sum_{j \in [k]}\alpha_j \ru_{y,j} \log \ru_{y,j}\\
	& \quad + \sum_{j \in [k]}\alpha_j  \log \alpha_j~.
		\end{split}	
	\end{equation}
	Further substituting \Cref{eq:hhone,eq:hhtwo,eq:hhthreenew} in the derivation below we get,
	$$\sum_{x,y \in \bX \times \bX} \big(\qv_{x,y} \log \qv_{x,y}  + \qu_{x,y} \log \qu_{x,y}\big)  - \sum_{x,y \in \bX \times \bX} \big(\sum_{z \in S_j} \qv_{x,z} \qu_{y,z} \big) \log \big( \sum_{z \in S_j} \qv_{x,z} \qu_{y,z} \big) = - \sum_{j \in [k]}\alpha_{j}\log\alpha_{j}~.$$
	Therefore the above derivation proves \Cref{eq:alphamain} and we further substitute it in \Cref{eq:internew} to get,
	\begin{equation}
	\begin{split}
	\frst(\ma,\bQ)  \geq \frst(\va,\qv) + \frst(\ua,\qu) - \sum_{j \in [k]}\alpha_{j}\log\alpha_{j}~.
	\end{split}
	\end{equation}
	The above expression combined with \Cref{eq:task} gives the following upper bound on the log of permanent,
	\begin{align}
	\log \perm(\ma) \leq {O(k\log \frac{N}{k})} + \frst(\ma,\bQ)-N~.
	\end{align}
	The above expression combined with definition of the scaled Sinkhorn permanent concludes the proof.
\end{proof}

\section{Lower bound for Bethe and scaled Sinkhorn permanent approximations}\label{sec:lbbethesink}
Here we provide the proof of \Cref{thm:lb} that is stated below for convenience.
\thmlb*
\begin{proof}
	\renewcommand{\mE}{\ma}
Assume $\Ns$ is divisible by $k$. Let $\mone$ and $\mzero$ be $\frac{\Ns}{k} \times \frac{\Ns}{k}$ all ones and all zeros matrices respectively. Note that $\bethe(\mone) \leq \frac{\Ns}{k}\log \frac{\Ns}{k}- \frac{\Ns}{k}+1$ and the proof for this statement follows because $\frac{k}{\Ns} \mone$ is the maximizer of the optimization problem $\max_{\bQ} \betheq(\mone,\bQ)$ over all doubly stochastic matrices $\bQ$. On the other hand $\log \perm (\mone) = \log \frac{\Ns}{k}! \geq \frac{\Ns}{k} \log \frac{\Ns}{k} - \frac{\Ns}{k} + \Omega (\log \frac{\Ns}{k})$, where in the last inequality we used the Stirling's approximation. Now consider the following matrix,
$$\mE \defeq 
\begin{bmatrix} 
\mone & \mzero &\dots \mzero \\
\mzero & \mone & \dots \mzero \\
\vdots & \dots & \ddots \\
\mzero & \mzero &\dots \mone \\
\end{bmatrix}$$
In the above definition, $\mE$ is a $\Ns \times \Ns$ matrix, with $k\times k$ blocks. For the matrix $\mE$ we have, $\log \perm(\mE)=k \cdot \log \perm(\mone) \geq k  \left(\frac{\Ns}{k} \log \frac{\Ns}{k} - \frac{\Ns}{k} + \Omega (\log \frac{\Ns}{k})\right)$ and $\bethe(\mE) =k \cdot \bethe (\mone) \leq k \left( \frac{\Ns}{k}\log \frac{\Ns}{k}- \frac{\Ns}{k}+1 \right)$. Therefore $\log \perm(\mE)-\bethe(\mE) \geq \Omega(k \log \frac{\Ns}{k})$.

The proof for the case when $\Ns$ is not divisible by $k$ is similar. Here matrix $\mE$ is the $\Ns \times \Ns$ block diagonal matrix where the first $k$ blocks correspond to $\floor{\frac{\Ns}{k}} \times \floor{\frac{\Ns}{k}}$ all ones matrix and the final block corresponds to $r \times r$ all ones matrix, where $r\defeq \Ns- k \floor{\frac{\Ns}{k}}$. For this definition of matrix $\mE$ we have, $\log \perm(\mE)=k \cdot \log \floor{\frac{\Ns}{k}}! + \log r!  \geq k  \left(\floor{\frac{\Ns}{k}} \log \floor{\frac{\Ns}{k}} - \floor{\frac{\Ns}{k}} + \Omega (\log \frac{\Ns}{k})\right)+r \log r -r +\Omega(\log r)$ and $\bethe(\mE) =k \cdot \bethe (\mone) \leq k \left(\floor{\frac{\Ns}{k}} \log \floor{\frac{\Ns}{k}} - \floor{\frac{\Ns}{k}} +1 \right)+r \log r -r +1$. Therefore $\log \perm(\mE)-\bethe(\mE) \geq \Omega(k \log \frac{\Ns}{k})$. The first condition of the theorem follows by taking exponentials on both sides of the previous inequality. 

The second inequality in the theorem follows by using $\bethe(\ma) \geq \ssink(\ma)$ (See \Cref{cor:ssinklb}). As the matrix $\mE$ constructed here is of non-negative rank $k$,  we conclude the proof.
\renewcommand{\mE}{\textbf{E}}
\end{proof}
\section{Improved approximation to profile maximum likelihood}\label{sec:pml}
In this section, we provide an efficient algorithm to compute an $\expo{-O(\sqrt{n}\log n)}$-approximate PML distribution. We first introduce the setup and some new notation. For convenience, we also recall some definitions from \Cref{sec:prelims}.

We are given access to $n$ independent samples from a hidden distribution $\bp \in \simplex$ supported on domain $\bX$. Let $x^n$ be this length $n$ sequence and $\phi=\Phi(x^n)$ be its corresponding profile. Let $\bff(x^n,y)$ be the frequency for domain element $y\in \bX$ in sequence $x^n$. Let $k$ be the number of non-zero distinct frequencies and we use 
$\{\mo,\dots \mj,\dots \mk \}$ 
to denote these distinct frequencies. Note that the number of non-zero distinct frequencies $k$ is always upper bounded by $\sqrt{n}$. For $j\in [1,k]$, we define $\phi_{j}\defeq |\{y\in \bX~|~ \bff(x^n,y)=\mj \} |$. Let $\pml$ be the PML distribution with respect to profile $\phi$ and is formally defined as follows,
$$\pml \in  \argmin_{\bp \in \simplex}\probpml(\bp,\phi)~.$$
Recall the definition of profile probability matrix $\maqphi$ with respect to profile $\phi$ and distribution $\bp$,
\begin{equation}\label{eq:matrixqphi}
\mapphi_{x,y}\defeq \bp_{x}^{\my} \text{ for all } x,y \in \bX,
\end{equation}
where $\my\defeq \bff(x^n,y)$ is the frequency of domain element $y\in \bX$ in the observed sequence $x^n$ and recall $\Phi(x^n)=\phi$. Note that the number of distinct columns is equal to number of distinct observed frequencies plus one (for the unseen) and therefore it is $k+1$.

From \Cref{eqpml2}, {the} probability of profile $\phi$ with respect to distribution $\bp$ is,
\begin{equation}\label{eqpml2new}
\probpml(\bp,\phi)=\cphi \cdot \left(\prod_{j\in [0,k]}\frac{1}{\phi_{j}!}\right) \cdot \perm(\mapphi)~,
\enspace \text{ where } \enspace
\cphi= \frac{n!}{\prod_{j\in [1,k]}(\mj!)^{\phi_{j}}}~.
\end{equation}
$\phi_{0}$ here denotes the number of unseen domain elements and note that it is not part of the profile. Given a distribution $\bp$ we know its domain $\bX$ therefore the unseen domain elements. Also, note that $\cphi$ is independent of the term $\phi_{0}$, therefore it depends just on the profile $\phi$ and not on the underlying distribution $\bp$. 

We now provide the motivation behind the techniques used in this section. Recall that the goal of this section is to compute an approximate PML distribution and we wish to do this using the results from the previous section. A first attempt would be to use the scaled Sinkhorn (or the Bethe) permanent as a proxy for the term $\perm(\mapphi)$ in \Cref{eqpml2new} and solve the following optimization problem:
$$\max_{\bp \in \simplex} \cphi \cdot \left(\prod_{j\in [0,k]}\frac{1}{\phi_{j}!}\right) \cdot \ssink(\mapphi)~.$$

The above optimization problem is indeed a good proxy for the PML objective and recall the above optimization problem is equivalent to the following:
$$\max_{\bp \in \simplex} \cphi \cdot \left(\prod_{j\in [0,k]}\frac{1}{\phi_{j}!}\right) \cdot \max_{\bQ \in \zrc} \expo{\frst(\mapphi,\bQ)}~.$$
Taking log and ignoring the constants we get the following equivalent optimization problem,
$$\max_{\bp \in \simplex} \max_{\bQ \in \zrc} \left(\log \frac{1}{\phi_{0}!} + \frst(\mapphi,\bQ) \right)$$
Interestingly, the function $\frst(\mapphi,\bQ)$, is concave with respect to $\bp$ for a fixed $\bQ$ and concave with respect to $\bQ$ for a fixed $\bp$ (See \cite{Von14}). However, unfortunately the function $\frst(\mapphi,\bQ)$ in general is not a concave function with respect to $\bp$ and $\bQ$ simultaneously~\cite{Von14} and we do not know how to solve the above optimization problem. Vontobel~\cite{Von14} proposed an alternating maximization algorithm to solve the above optimization problem, and studied its implementation and convergence to a stationary point; see \cite{Von14} for empirical performance of this approach. Using the Bethe permanent as a proxy in the above optimization problem has similar issues; see~\cite{Von12,Von14} for further details.

To address the above issue we use the idea of probability discretization from \cite{CSS19}, meaning we assume distribution takes all its probability values from some fixed predefined set. We use this idea in a different way than \cite{CSS19} and further exploit the structure of optimal solution $\bQ$ to write a convex optimization problem that approximates the PML objective. The solution of this convex optimization problem returns a fractional representation of the distribution that we later round to return the approximate PML distribution with desired guarantees. Surprisingly, the final convex optimization problem we write is exactly same as the one in \cite{CSS19} and our work gives a better analysis of the same convex program by a completely different approach.

The rest of this section is organized as follows. In \Cref{sec:prob_discrete}, we study the probability discretization. In the same section, we also study the application of results from \Cref{sec:bethe} for approximating the permanent of profile probability matrix ($\mapphi$). We further provide the convex optimization problem at the end of this section that can be solved efficiently and returns a fractional representation of the approximate PML distribution. In \Cref{subsec:algmain}, we provide the rounding algorithm that returns our final approximate PML distribution. 
Till this point, all our results are independent of the choice of the probability discretization set. Later in \Cref{subsec:combine}, we choose an appropriate probability discretization set and further combine analysis from all the previous sections. In this section, we state and analyze our final algorithm that returns a $\expo{-O(\sqrt{n}\log n)}$-approximate PML distribution. Note that our rounding algorithm is technical and for the continuity of reading we defer all the proofs for results in \Cref{subsec:algmain} to \Cref{subsec:missing}.
\subsection{Probability discretization}\label{sec:prob_discrete}
\renewcommand{\boo}{\ell}
\renewcommand{\epso}{\epsilon}
Here we study the idea of probability discretization that is also used in \cite{CSS19}. We combine this with other ideas from \Cref{sec:bethe} to provide a convex program that approximates the PML objective.

Let $\bR \subseteq [0,1]_{\R}$ be some discretization of the probability space and in this section we consider distributions that take all its probability values in set $\bR$. All results in this section hold for finite set $\bR$ and we specify the exact definition of $\bR$ in \Cref{subsec:combine}.

The discretization introduces a technicality of probability values not summing up to one and we redefine pseudo-distribution and discrete pseudo-distribution from \cite{CSS19} to deal with these.

\begin{defn}[Pseudo-distribution] $\bq \in [0,1]^{\bX}_{\R}$ is a \emph{pseudo-distribution} if $\|\bq\|_1 \leq 1$ and a \emph{discrete pseudo-distribution} with respect to $\bR$ if all its entries are in $\bR$ as well. We use $\psimplex$ and $\dsimplex$ to denote the set of all such pseudo-distributions and discrete pseudo-distributions with respect to $\bR$ respectively.
\end{defn}

We extend and use the following definition for $\bbP(\vvec,y^n)$ for any vector $\vvec \in \R_{\geq 0}^{\bX}$ and therefore for pseudo-distributions as well,
$$\bbP(\vvec,y^n) \defeq \prod_{x \in \bX}\vvec_x^{\bff(y^n,x)}~.$$
Further, for any probability terms defined involving $\bp$, we define those terms for any vector $\vvec \in \R_{\geq 0}^{\bX}$ just by replacing $\bp_{x}$ by $\vvec_{x}$ everywhere. For convenience we refer to $\probpml(\bq,\phi)$ for any pseudo-distribution $\bq$ as the ``probability'' of profile $\phi$ with respect to $\bq$.

For a scalar $c$ and set $\bS$, define $\floor{c}_{\bS}$ and $\ceil{c}_{\bS}$ as follows:\\
$$ \floor{c}_{\bS}\defeq\sup_{s \in \bS: s \leq c}s \quad \text{ and } \quad \ceil{c}_{\bS}\defeq\inf_{s \in \bS: s \geq c}s$$
\begin{defn}[Discrete pseudo-distribution] For any distribution $\bp \in \simplex$, its \emph{discrete pseudo-distribution} $\bq=\disc(\bp) \in \dsimplex$ with respect to $\bR$ is defined as:
	$$\bq_x\defeq\floor{\bp_x}_{\bR} \quad \forall x \in \bX$$
\end{defn}

\renewcommand{\bp}{\textbf{q}}

We now define some additional definitions and notation that will help us lower and upper bound the permanent of profile probability matrix by a convex optimization problem.
\begin{itemize}
	\item  Let $\ell\defeq |\bR|$ be the cardinality of set $\bR$ and $\ri$ be the $i$'th element of set $\bR$.
	\item For any discrete pseudo-distribution $\bp$ with respect to $\bR$, that is $\bq \in \dsimplex$, we let $\lpi \defeq |\{y\in \bX~|~\bp_y=\ri \}|$, be the number of domain elements with probability $\ri$.
	\item Let $\za \subseteq \zrm$ be the set of non-negative matrices such that, for any $\bS \in \za$ the following holds:
\begin{align}
\sum_{j \in [0,k]}\Sij=\lpi \text{ for all } i \in [1,\ell] \quad  \text{and} \quad \sum_{i \in [1,\ell]}\Sij=\phi_{j} \text{ for all }j \in [0,k]~,
\end{align}
where $\phi_{0}$\footnote{$\phi_{0}$ is not part of the profile and is not given to us. Later in this section, we get rid of this dependency on $\phi_{0}$.} is the number of unseen domain elements and we use $\mz \defeq 0$ to denote the corresponding frequency element.
	\item For any $\bS \in \R_{\geq0}^{\ell \times (k+1)}$ define, 
\begin{equation}
\fng(\bS)=\sum_{i \in [1,\ell]} \sum_{j \in [0,k]}\left[\Sij\log (\frac{\ri^{\mj}}{\Sij})\right]+\sum_{i \in [1,\ell]}\left(\sum_{j \in [0,k]}\Sij \right)\log \left(\sum_{j \in [0,k]}\Sij \right)+\sum_{j \in [0,k]}\phi_{j}\log \phi_{j}-\sum_{j \in [0,k]}\phi_{j}~.
\end{equation}
\item Throughout this section, for convenience unless stated otherwise we abuse notation and use $\ma$ to denote the matrix $\maqphi$. The underlying pseudo-distribution $\bq$ and profile $\phi$ with respect to matrix $\ma$ will be clear from the context.
\end{itemize}
The first half of this section is dedicated to bound the $\perm(\ma)$ in terms of function $\fng(\bS)$. For any fixed discrete pseudo-distribution $\bq$ and profile $\phi$, we will show that,
$$\max_{\bS \in \za}\fng(\bS) \leq \log \perm(\maqphi) \leq O(k\log \frac{\Ns}{k}) + \max_{\bS \in \za}\fng(\bS)~.$$
Later in the second half, we use the above inequality to maximize over all the discrete pseudo-distributions to find the approximate PML distribution and the summary of which is stated later. We start by showing the lower bound first and later in \Cref{thm:newupperbound} we prove the upper bound. 
\begin{theorem}\label{thm:newlowerbound}
	For any discrete pseudo-distribution $\bq$ with respect to $\bR$ and profile $\phi$, let $\ma$ be the matrix defined (with respect to $\bq$ and $\phi$) in \Cref{eq:matrixqphi}, then the following holds,
	\begin{equation}
	\log \perm(\ma) \geq \max_{\bS \in \za}\fng(\bS)~.
	\end{equation}
\end{theorem}
\begin{proof}
	For any matrix $\bS \in \za$, define matrix $\bQ \in \R^{\bX \times \bX}$ as follows,
	$$\bQ_{x,y}\defeq \frac{\Sij}{\lpi\phi_{j}}$$
	where in the definition above $i$ and $j$ are such that $\bp_x=\ri$ and $\my=\mj$. We now establish that matrix $\bQ$ is doubly stochastic. For each $x \in \bX$, let $i$ be such that $\bp_x=\ri$, then
	\begin{equation}
	\begin{split}
	\sum_{y \in \bX}\bQ_{x,y}&= \sum_{j \in [0,k]} \sum_{\{y \in \bX~|~\my=\mj\}}\frac{\Sij}{\lpi \phi_{j}}=\sum_{j \in [0,k]}\frac{\Sij}{\lpi\phi_{j}} \sum_{\{y \in \bX~|~\my=\mj\}}1\\
	&=\sum_{j \in [0,k]} \frac{\bS_{x,\mj}}{\lpi \phi_{j}} \cdot \phi_{j}=\frac{1}{\lpi}\sum_{j \in [0,k]} \bS_{x,\mj}=1~.\\
	\end{split}
	\end{equation}
	For each $y\in \bX$, let $j$ be such that $\my=\mj$, then
	\begin{equation}
	\begin{split}
	\sum_{x \in \bX}\bQ_{x,y}&= \sum_{i\in [1,\ell]} \sum_{\{x \in \bX~|~\bp_x=\ri\}}\frac{\Sij}{\lpi \phi_{j}}=\sum_{i \in [1,\ell]}\frac{\Sij}{\lpi\phi_{j}} \sum_{\{x \in \bX~|~\bp_x=\ri\}}1\\
	&=\sum_{i \in [1,\ell]} \frac{\bS_{x,\mj}}{\lpi \phi_{j}} \cdot \lpi=\frac{1}{\phi_j}\sum_{i \in [1,\ell]} \bS_{x,\mj}=1~.\\
	\end{split}
	\end{equation}
	Since matrix $\bQ$ is doubly stochastic, by the definition of the scaled Sinkhorn permanent (See \Cref{def:ssink}) and \Cref{cor:ssinklb} we have $\log \perm(\ma) \geq \frst(\ma,\bQ)-\Ns$. To conclude the proof we show that $\frst(\ma,\bQ) -\Ns = \fng(\bS)$.
\begin{equation}\label{eq:onee}
\begin{split}
\frst(\ma,\bQ)&=\sum_{(x,y) \in \bX \times \bX} \bQ_{x,y}\log (\frac{{\ma}_{x,y}}{\bQ_{x,y}}) =\sum_{i \in [1,\ell]} \sum_{j \in [0,k]}  \lpi \phi_{j} \cdot \frac{\Sij}{\lpi   \phi_{j}}\log (\frac{\ri^{\mj} \lpi  \phi_{j}}{\Sij})\\
&= \sum_{i \in [1,\ell]} \sum_{j \in [0,k]}\Sij\log (\frac{\ri^{\mj} \lpi  \phi_{j}}{\Sij})~.
\end{split}
\end{equation}
We consider the final expression above and simplify it. First note that, 
$$\sum_{i \in [1,\ell]} \sum_{j \in [0,k]}\Sij\log \lpi= \sum_{i \in [1,\ell]}\log \lpi \sum_{j \in [0,k]} \Sij=\sum_{i \in [1,\ell]}\lpi \log \lpi~.$$ 
Similarly,
$$\sum_{i \in [1,\ell]} \sum_{j \in [0,k]}\Sij\log \phi_{j}=\sum_{j \in [0,k]}\log \phi_{j} \sum_{i \in [1,\ell]} \Sij=\sum_{j \in [0,k]}\phi_{j}\log \phi_{j}~.$$
Using the above two expressions, the final expression of \Cref{eq:onee} can be equivalently written as,
\begin{equation}\label{eq:twoo}
\sum_{i \in [1,\ell]} \sum_{j \in [0,k]}\Sij\log (\frac{\ri^{\mj} \lpi  \phi_{j}}{\Sij})= \sum_{i \in [1,\ell]} \sum_{j \in [0,k]}\left[\Sij\log (\frac{\ri^{\mj}}{\Sij})\right]+ \sum_{i \in [1,\ell]}\lpi \log \lpi+\sum_{j \in [0,k]}\phi_{j}\log \phi_{j}~.
\end{equation}
Combining \Cref{eq:onee}, \Cref{eq:twoo} and substituting $\Ns=\sum_{j \in [0,k]}\phi_{j}$, we get:
$$\frst(\ma,\bQ) -\Ns = \sum_{i \in [1,\ell]} \sum_{j \in [0,k]}\Sij\log (\frac{\ri^{\mj}}{\Sij})+ \sum_{i \in [1,\ell]}\lpi \log \lpi+\sum_{j \in [0,k]}\phi_{j}\log \phi_{j}-\sum_{j \in [0,k]}\phi_{j}=\fng(\bS)~.$$
In the above equality we used $\sum_{j \in [0,k]}\Sij=\levelq_{i}$ for all $i \in [1,\ell]$ and for any $\bS \in \za$. Combining the above inequality with $\log \perm(\ma) \geq \frst(\ma,\bQ) -\Ns$ we get,
$$\log \perm(\ma) \geq \fng(\bS)~.$$
The above inequality holds for any $\bS \in \za$ (and therefore holds for the maximizer as well) and we conclude the proof.
\end{proof}

We next give an upper bound for the log of permanent of $\ma$ in terms of $\fng(\bS)$. 
\begin{theorem}\label{thm:newupperbound}
		For any discrete pseudo-distribution $\bq$ with respect to $\bR$ and profile $\phi$, let $\ma$ be the matrix defined (with respect to $\bq$ and $\phi$) in \Cref{eq:matrixqphi}. Then,
	$$\log \perm(\ma)	\leq O(k \log \frac{\Ns}{k})+ \max_{\bS \in \za}\fng(\bS)~.$$
\end{theorem}
\begin{proof}
	Here we construct a particular matrix $\bS \in \za$ such that $\log \perm(\ma) \leq O(k \log \frac{\Ns}{k})+ \fng(\bS)$, which immediately implies the theorem. Recall by \Cref{lem:upperm,lem:condip}, there exists a matrix $\bP \in \R_{\geq 0}^{\bX \times (k+1)}$ such that, $\sum_{j \in [0,k]}\Pxj=1 \text{ for all } x\in \bX$ and $\sum_{x \in \bX}\Pxj=\phi_{j} \text{ for all } j \in [0,k]$, 
	and satisfies $\log \perm(\ma) \leq O(k \log \frac{\Ns}{k}) + \fnp(\ma,\bP)$ \footnote{The inequality holds because matrix $\ma$ has $k+1$ distinct columns and $O((k+1) \log \frac{\Ns}{k+1})$ is asymptotically same as $O(k \log \frac{\Ns}{k})$.}. Further using the definition of $\fnp(\ma,\bP)$ we get,
	\begin{equation}\label{eq:oooo}
	\log \perm(\ma)	\leq O(k \log \frac{\Ns}{k})+ \sum_{j \in [0,k]}\phi_{j}\log \phi_{j} -\sum_{j \in [0,k]}\phi_{j} +\sum_{(x,j) \in \bX \times [0,k]} \Pxj\log  \frac{\madxj}{\Pxj}~,
	\end{equation}
	where for the matrix $\ma$ defined (with respect to $\bq$ and $\phi$) in \Cref{eq:matrixqphi}, we have,
	$$\madxj = \bq_{x}^{\mj}~.$$ 
	We now define the matrix $\bS$ that satisfies the conditions of the lemma. 
	$$\Sij\defeq\sum_{\{x \in \bX~|~\bp_{x}=\ri \}} \Pxj$$
	By \Cref{thm:sameprob}, for any fixed $j \in [0,k]$, all $x\in \bX$ such that $\bp_{x}=\ri$, share the same probability value $\Pxj$ and we use the notation $\bP_{i,j}$ to denote this value. Using this definition, we have:
	\begin{equation}\label{eq:sij}
	\Sij = \level^{\bp}_{i}\bP_{i,j}
	\end{equation}
	Further note that for any $i\in [1,\ell]$, if $x\in \bX$ is any element such that $\bp_{x}=\ri$, then
	$$\sum_{j \in [0,k]}\bP_{i,j}=\sum_{j \in [0,k]}\Pxj=1$$
	We wish to show that $\bS\in \za$. We first analyze the row sum constraint. For each $i \in [1,\ell]$, 
	$$\sum_{j \in [0,k]}\Sij= \sum_{j \in [0,k]}\level^{\bp}_{i}\bP_{i,j}=\level^{\bp}_{i}$$
	We now analyze the column constraint. For each $j \in [0,k]$,
	$$\sum_{i\in [1,\ell]}\Sij=\sum_{i\in [1,\ell]} \sum_{\{x \in \bX~|~\bp_{x}=\ri \}} \Pxj=\sum_{x \in \bX}\Pxj=\phi_{j}$$
	In the remainder of the proof we show that the matrix $\bS$ defined earlier satisfies $\log \perm(\ma)	\leq O(k \log \frac{\Ns}{k})+ \fng(\bS)$. We start by simplifying the term $\sum_{(x,j) \in \bX \times [0,k]} \Pxj\log  \frac{\madxj}{\Pxj}$ in \Cref{eq:oooo},
	\begin{equation}\label{eq:tttt}
	\begin{split}
	\sum_{(x,j) \in \bX \times [0,k]} \Pxj\log  \frac{\madxj}{\Pxj}&=	\sum_{j \in [0,k]}\sum_{i\in [1,\ell]}\sum_{\{x\in \bX~|~\bp_{x}=\ri\}} \Pxj\log  \frac{\madxj}{\Pxj}
	=\sum_{j \in [0,k]}\sum_{i\in [1,\ell]}\sum_{\{x\in \bX~|~\bp_{x}=\ri\}} \bP_{i,j}\log  \frac{\ri^{\mj}}{\bP_{i,j}}\\
	&=\sum_{j \in [0,k]}\sum_{i\in [1,\ell]}\level^{\bp}_{i} \bP_{i,j}\log  \frac{\ri^{\mj}}{\bP_{i,j}}
	=\sum_{i\in [1,\ell]}\sum_{j \in [0,k]} \Sij\log \frac{\ri^{\mj}\level^{\bp}_{i}}{\Sij}\\
	&=\sum_{i\in [1,\ell]}\sum_{j \in [0,k]} \Sij\log \frac{\ri^{\mj}}{\Sij}+\sum_{i\in [1,\ell]}\level^{\bp}_{i}\log \level^{\bp}_{i}
	\end{split}
	\end{equation}
	In the second equality, we used $\madxj=\ri^{\mj}$ and further by the definition of $\bP_{i,j}$ we have $\bP_{x,j}=\bP_{i,j}$ for all $x\in \bX$ that satisfy $\bq_{x}=\ri$. In the third equality, we used $\sum_{\{x\in \bX~|~\bp_{x}=\ri\}}1=\level^{\bp}_{i}$. In the fourth equality we used \Cref{eq:sij}. In the final equality, we used 	$\sum_{i\in [1,\ell]}\sum_{j \in [0,k]} \Sij\log \frac{\ri^{\mj}\level^{\bp}_{i}}{\Sij}=\sum_{i\in [1,\ell]}\sum_{j \in [0,k]} \Sij\log \frac{\ri^{\mj}}{\Sij}+\sum_{i\in [1,\ell]}\sum_{j \in [0,k]} \Sij\log \level^{\bp}_{i}$ and the final term further simplifies to the following, $\sum_{i\in [1,\ell]}\sum_{j \in [0,k]} \Sij\log \level^{\bp}_{i}=\sum_{i\in [1,\ell]}\log \level^{\bp}_{i}\sum_{j \in [0,k]} \Sij =\sum_{i\in [1,\ell]}\level^{\bp}_{i}\log \level^{\bp}_{i}$.
	
	We conclude the proof by combining equations \ref{eq:oooo} and \ref{eq:tttt} and using $\sum_{j \in [0,k]}\Sij=\levelq_{i}$ for any $\bS \in \za$.
	\end{proof}
Note using Theorem~\ref{thm:newlowerbound} and \ref{thm:newupperbound}, for matrix $\ma$ defined (with respect to $\bq$ and $\phi$) in \Cref{eq:matrixqphi}, we showed the following,
\begin{equation}\label{eq:permlbub}
\max_{\bS \in \za} \fng(\bS) \leq \log \perm(\ma) \leq O(k\log \frac{\Ns}{k})+\max_{\bS \in \za} \fng(\bS)~.
\end{equation}

\renewcommand{\boo}{\ell}
\renewcommand{\btt}{k}
\renewcommand{\bttpo}{\btt}
\renewcommand{\ztbtt}{[\btt]}
\renewcommand{\pvec}{\textbf{r}}

Our final goal of this section is to maximize $\probpml(\bq,\phi) \propto \frac{1}{\phi_{0}!}\perm(\ma)$ over discrete pseudo-distributions $\bq$ but let us take a step back and just focus on writing an upper bound. Consider the term $\max_{\bS \in \za} \fng(\bS)$ in the expression above, it depends on discrete pseudo-distribution $\bq$ at two different places. The first is the constraint set $\za$ and the second is the function $\fng(\bS)$ (because it contains the $\phi_{0}$ term in its expression). We address the first issue by defining the following new set that encodes the constraint set $\za$ for all discrete pseudo-distributions simultaneously.
\renewcommand{\bZ}{\textbf{Z}^{\phi}_{\bR}}
\begin{defn}
Let $\bZ\subset \R_{\geq 0}^{\ell \times (k+1)}$ be the set of non-negative matrices, such that any $\bS \in \bZ$ satisfies,
\begin{equation}\label{eq:zrphi}
\sum_{i\in[1,\ell]}\bS_{i,j}=\phi_{j} \text { for all }j \in [1,k], \sum_{j\in[0,k]}\bS_{i,j} \in \Z \text { for all }i \in [1,\ell] \text{ and } \sum_{i \in [1,k]} \ri \sum_{j \in [0,k]}\bS_{i,j}\leq 1~.
\end{equation}
\end{defn}
Note in the definition of $\bZ$ we removed the constraint related to $\phi_{0}$ and recall $\phi_{0}$ denotes the number of unseen domain elements. Not having constraint with respect to $\phi_{0}$ helps us encode discrete pseudo-distributions (with respect to $\bR$) of different domain sizes. Further for any $\bS \in \bZ$, there is a discrete pseudo-distribution associated with it and we define it next.
\begin{defn}\label{defn:distS}
For any $\bS \in \bZ$, the discrete pseudo-distribution $\bqS$ associated with $\bS$ is defined as follows: For any arbitrary $\sum_{j\in[0,k]}\bS_{i,j}$ number of domain elements assign probability $\ri$.
\end{defn}
Note in the definition above $\bqS$ is a valid pseudo-distribution because of the third condition in \Cref{eq:zrphi}. Further note that, for any discrete pseudo-distribution $\bq$ and $\bS \in \za$, the distribution $\bqS$ associated with respect to $\bS$ is a permutation of distribution $\bq$. Since the probability of a profile is invariant under permutations of distribution, we treat all these distributions the same and do not distinguish between them.

We now handle the second issue that corresponds to removing the dependency of discrete pseudo-distribution $\bq$ from the function $\fng(\bS)$. To handle this issue, we define a new function $\bg(\bS)$ that when maximized over the set $\za$ and $\bZ$ approximates the value $\probpml(\bp,\phi)$ and $\max_{\bq \in \dsimplex} \probpml(\bp,\phi)$ respectively (See next theorem for the formal statement). For any $\bS \in \R_{\geq 0}^{\ell \times (k+1)}$, the function $\bg(\bS)$ is defined as follows,
\begin{equation}
\bg(\bS)\defeq \expo{\sum_{i \in [1,\ell]} \sum_{j \in [0,k]}\left[\Sij\log (\frac{\ri^{\mj}}{\Sij})\right]+\sum_{i \in [1,\ell]}\left(\sum_{j \in [0,k]}\Sij \right)\log \left(\sum_{j \in [0,k]}\Sij \right)}~.
\end{equation}
Note that we switch gears and define the function $\bg(\bS)$ as an exponential function. $\bg(\bS)$ approximates the value $\probpml(\bp,\phi)$ instead of log of it and it helps with proof readability. The following theorem summarizes this result.
\begin{theorem}\label{pmlprob:approx}
	Let $\bR$ be a probability discretization set. Given a profile $\phi$ and discrete pseudo-distribution $\bq$ with respect to $\bR$. The following inequality holds,
	\begin{equation}
\expo{-O(k \log (\Ns+n))}\cdot	\cphi \cdot	\max_{\bS \in \za}\bg(\bS) \leq \probpml(\bp,\phi) \leq \expo{\bigO{k \log \frac{\Ns}{k}}} \cdot \cphi \cdot \max_{\bS \in \za}\bg(\bS)
	\end{equation}
	Further,
	\begin{equation}
\expo{-O(k \log (\Ns+n))}\cdot	\cphi \cdot	\max_{\bS \in \bZ}\bg(\bS) \leq \max_{\bq \in \dsimplex}\probpml(\bp,\phi) \leq \expo{\bigO{k \log \frac{\Ns}{k}}} \cdot \cphi \cdot \max_{\bS \in \bZ}\bg(\bS)
	\end{equation}
\end{theorem}
\begin{proof}
	For any discrete pseudo-distribution $\bq$ with respect to $\bR$ and profile $\phi$, let $\ma$ be the matrix defined (with respect to $\bq$ and $\phi$) in \Cref{eq:matrixqphi}. Then, by \Cref{eq:permlbub} we have,
	$$\max_{\bS \in \za} \fng(\bS) \leq \log \perm(\ma) \leq O(k\log \frac{\Ns}{k})+\max_{\bS \in \za} \fng(\bS)~.$$
	Further by \Cref{eqpml2} we have,
	$$\probpml(\bp,\phi)=\cphi \cdot \left(\prod_{j\in [0,k]}\frac{1}{\phi_{j}!}\right) \cdot \perm(\maqphi)~.$$
	Combining the above two equations we have,
	\begin{equation}\label{eq:lbubone}
		\cphi \cdot \left(\prod_{j\in [0,k]}\frac{1}{\phi_{j}!}\right) \cdot \max_{\bS \in \za}\expo{\fng(\bS)} \leq \probpml(\bp,\phi)\leq  \expo{\bigO{k \log \frac{\Ns}{k}}} \cdot \cphi \cdot \left(\prod_{j\in [0,k]}\frac{1}{\phi_{j}!}\right) \cdot \max_{\bS \in \za}\expo{\fng(\bS)}
	\end{equation}
	We now simplify the term $\left(\prod_{j\in [0,k]}\frac{1}{\phi_{j}!}\right) \cdot \expo{\fng(\bS)}$ in the above expression. First note that for any $\bS \in \za$, 
	$$\expo{\fng(\bS)}=\bg(\bS)\cdot \expo{\sum_{j \in [0,k]}\phi_{j}\log \phi_{j}-\sum_{j \in [0,k]}\phi_{j}}~.$$
Therefore,
	\begin{equation}\label{eq:lbubtwo}
	\begin{split}
	\left(\prod_{j\in [0,k]}\frac{1}{\phi_{j}!}\right) \cdot \expo{\fng(\bS)}&=\left(\prod_{j\in [0,k]}\frac{1}{\phi_{j}!}\right) \cdot \bg(\bS)\cdot \expo{\sum_{j \in [0,k]}\phi_{j}\log \phi_{j}-\sum_{j \in [0,k]}\phi_{j}}~.
	\end{split}
	\end{equation}
	By \Cref{lem:ster} (Stirling's approximation) we have,
	\begin{equation}\label{eq:lbubthree}
	\expo{-O\left(k \log (\Ns+n) \right)} \leq \left(\prod_{j\in [0,k]}\frac{1}{\phi_{j}!}\right) \cdot \expo{\sum_{j \in [0,k]}\phi_{j}\log \phi_{j}-\sum_{j \in [0,k]}\phi_{j}} \leq 1~.
	\end{equation}
The first inequality follows because for each $j \in [0,k]$, we have $\frac{1}{\phi_{j}!} \expo{\phi_{j}\log \phi_{j}-\phi_{j}} \geq \Omega(\frac{1}{\sqrt{\phi_{j}+1}})$, which by using $\phi_{j} \leq \Ns+n$ is further lower bounded by $\Omega(\frac{1}{\sqrt{\Ns+n}}) \geq \expo{-O(\log (\Ns+n))}$. \Cref{eq:lbubthree} follows by taking product over all $j \in [0,k]$. Now combining \Cref{eq:lbubthree} and \Cref{eq:lbubtwo} we have,
\begin{equation}\label{eq:lbubfour}
\expo{-O(k \log (\Ns+n))} \cdot \bg(\bS) \leq \left(\prod_{j\in [0,k]}\frac{1}{\phi_{j}!}\right) \cdot \expo{\fng(\bS)} \leq \bg(\bS)~.
\end{equation}
The first statement of the lemma follows by combining the above \Cref{eq:lbubfour} with \Cref{eq:lbubone}, that is we have,
\begin{equation}\label{eq:caseone}
\expo{-O(k \log (\Ns+n))}\cdot \cphi \cdot	\max_{\bS \in \za}\bg(\bS) \leq \probpml(\bp,\phi) \leq \expo{\bigO{k \log \frac{\Ns}{k}}} \cdot \cphi \cdot \max_{\bS \in \za}\bg(\bS)~.
\end{equation}
Given a profile $\phi$, for any discrete pseudo-distribution $\bq \in \dsimplex$ we have $\za \subseteq \bZ$ and further combining it with above inequality we get,
$$\max_{\bq \in \dsimplex}\probpml(\bp,\phi) \leq \expo{\bigO{k \log \frac{\Ns}{k}}} \cdot \cphi \cdot \max_{\bS \in \bZ}\bg(\bS)~.$$
Note that for any $\bS\in \bZ$, we also have $\bS \in \bZS$, where $\bqS$ is the discrete pseudo-distribution associated with respect to $\bS$ (See \Cref{defn:distS}). Therefore,
$$\expo{-O(k \log (\Ns+n))}\cdot	\cphi \cdot	\max_{\bS \in \bZ}\bg(\bS) \leq \expo{-O(k \log (\Ns+n))}\cdot	\cphi \cdot	\max_{\bq \in \dsimplex}\max_{\bS \in \za}\bg(\bS) \leq  \max_{\bq \in \dsimplex}\probpml(\bp,\phi)~.$$
For the last inequality in the above derivation we used \Cref{eq:caseone}. Now combining the previous two inequalities we conclude the proof.
\end{proof}
The previous theorem provides an upper bound for the probability of profile with respect to any discrete pseudo-distribution. However one issue with this upper bound is that it is not efficiently computable because the set $\bZ$ is not a convex set (because of the integrality constraints). We relax these integrality constraints and define the following new set.
\begin{defn}\label{def:bzfrac}
	Let $\bZfrac\subseteq \R_{\geq 0}^{\ell \times (k+1)}$ be the set of non-negative matrices, such that any $\bS \in \bZfrac$ satisfies,
	\begin{equation}
	\sum_{i\in[1,\ell]}\bS_{i,j}=\phi_{j} \text { for all }j \in [1,k] \text{ and } \sum_{i \in [1,k]} \ri \sum_{j \in [0,k]}\bS_{i,j}\leq 1~.
	\end{equation}
\end{defn}
\begin{lemma}\label{lem:dpmlapprox}
	Let $\bR$ be a probability discretization set. Given a profile $\phi$, the following holds,
	\begin{equation}
	\max_{\bq \in \dsimplex}\probpml(\bp,\phi) \leq \expo{\bigO{k \log \frac{\Ns}{k}}} \cdot \cphi \cdot \max_{\bS \in \bZfrac}\bg(\bS)
	\end{equation}
	\end{lemma}
\begin{proof}
By \Cref{pmlprob:approx} we already have,
$$\max_{\bq \in \dsimplex}\probpml(\bp,\phi) \leq \expo{\bigO{k \log \frac{\Ns}{k}}} \cdot \cphi \cdot \max_{\bS \in \bZ}\bg(\bS)~.$$
The lemma holds because $\bZ \subseteq \bZfrac$.
\end{proof}
Note in the above lemma, the upper bound only depends on the profile \footnote{$\cphi$ has no dependency on $\phi_{0}$.} and we removed all dependencies related to distributions (and also $\phi_{0}$). Next we show that this upper bound can be efficiently computed by using the result that function $\bg(\bS)$ is log concave in $\bS$.
\begin{lemma}[Lemma 4.16 in \cite{CSS19}]
	Function $\bg(\bS)$ is log concave in $\bS$.
	\end{lemma}
\begin{theorem}[Theorem 4.17 in \cite{CSS19}]\label{thm:convrt}
	Given a profile $\phi \in \Phi^{n}$, the optimization problem $\max_{\bS \in \bZfrac}\log \bg(\bS)$ can be solved in time $\otilde(k^2\ell)$. \footnote{Note here we hide the logarithmic dependence on $n$, the size of sample.}
	\end{theorem}
\subsection{Rounding Algorithm}\label{subsec:algmain}
In the previous section we provided an efficiently computable upper bound for the probability of profile $\phi$ with respect to any discrete pseudo-distribution $\bq \in \dsimplex$. This upper bound returns a solution $\bS \in \bZfrac$ and we need to round this solution to construct a discrete pseudo-distribution that approximates this upper bound. In this section we provide a rounding algorithm that takes as input $\bS \in \bZfrac$ and returns a solution $\bSext \in \bZext$, where $\bRext$ is an extended probability discretization set. Further using $\bSext\in \bZext$, we construct a discrete pseudo-distribution $\bqSext$ with respect to $\bRext$ such that $\probpml(\bqSext,\phi)$ approximates the upper bound and therefore is an approximate PML distribution. Our rounding algorithm is technical and we next provide a overview to better understand it.

\paragraph{Overview of the rounding algorithm:} The goal of the rounding algorithm is to take a fractional solution $\textbf{S}\defeq \mathrm{arg}\max_{\bS' \in \bZfrac}\log \bg(\bS')$ as input and round each row sum to an integral value while preserving the column sums and $\fnggg(\textbf{S})$ value. Our rounding algorithm proceeds in three steps:
\paragraph{Step 1:} Consider the fractional solution $\textbf{S} \in \R_{\geq 0}^{\ell \times (k+1)}$ and recall the rows are indexed by the elements of set $\bR$ (which represent probability values). We first round the rows corresponding to the higher probability values by simply taking the floor (rounding down to the nearest integer) of each entry. This procedure ensures the integrality of the row sums (corresponding to higher probability values) but violates the column sum constraints. To satisfy the column sum constraints and the distributional constraint (i.e. last condition in \Cref{eq:zrphi}) simultaneously, we create rows corresponding to new probability values using \Cref{alg:create}. However to ensure that all these new rows also have integral row sums, we modify the (old) rows corresponding to lower probability values accordingly. Let $\textbf{S}^{(1)}$ be the solution returned by the first step of the rounding algorithm. \Cref{alg:create} ensures that the $\fnggg(\textbf{S}^{(1)})$ value is not much smaller than $\fnggg(\textbf{S})$. In $\textbf{S}^{(1)}$, all the new rows and (old) rows corresponding to higher probability values have integral row sums and we round the remaining rows corresponding to smaller probability values next.
\paragraph{Step 2:} In this step, we round all the rows corresponding to the smaller probability values. For each of these rows, we scale all the entries in a particular row by the same factor to ensure that the row sum is rounded down to the nearest integer. Similar to the step 1, using \Cref{alg:create} we create rows corresponding to new probability values to maintain the column sum constraints and the distributional constraint; all these new rows further correspond to small probability values. Unlike in the previous step, the new rows created in step two may not have integral row sums but these rows have a nice diagonal structure. Let $\textbf{S}^{(2)}$ be this intermediate solution created in step 2. \Cref{alg:create} ensures that the $\fnggg(\textbf{S}^{(2)})$ value is not much smaller than $\fnggg(\textbf{S}^{(1)})$ (and hence $\fnggg(\textbf{S})$). Note all the row sums in $\textbf{S}^{(2)}$ are integral except the new rows created in step 2 that all have small probability values and have diagonal structure.
\paragraph{Step 3:} In this final step, using \Cref{alg:structurerounding} we round the new rows created in step 2. \Cref{alg:structurerounding} exploits the low probability and diagonal structure in these rows. The diagonal structure ensures that there is just one non-zero entry in any particular row and we modify the solution $\textbf{S}^{(2)}$ (from the previous step) as follows. We transfer the mass from a non-integral lower probability value row to an immediate higher probability value row until the (lower probability value) row sum is integral. This process might violate the distributional constraint  and we rescale the probability values accordingly to satisfy this constraint. Let $\bSext$ be the solution returned by step 3. We ensure that all column sums are preserved, all row sums are integral and the $\fnggg(\bSext)$ value is not much smaller than $\fnggg(\textbf{S}^{(2)})$ (and hence not much smaller than $\fnggg(\textbf{S})$).\\
\\
In the remainder of this section we state all three algorithms and the results corresponding to them. For continuity of reading, we defer the proofs of these results to \Cref{subsec:missing}. For convenience, we first state \Cref{alg:structurerounding} that rounds the rows corresponding to the low probability values in step 3 of our main rounding algorithm (\Cref{alg:rounding}). We follow up this algorithm with a lemma that summarizes the guarantees provided by it. Later we state \Cref{alg:create} that creates rows corresponding to new probability values to preserve the column sums and the distributional constraint. This algorithm is invoked as a subroutine in both step 1 and 2 of \Cref{alg:rounding}. Finally, we state our main rounding algorithm that consists of three different steps. We then state results analyzing each of these steps separately. The final result (\Cref{thm:stepthree}), is the main theorem of this subsection that summarizes the final guarantees promised by our rounding algorithm.
\renewcommand{\ztbtt}{[0,\btt]}
\begin{algorithm}[H]
	\caption{Structured Rounding Algorithm}
	\begin{algorithmic}[1]
		\Procedure{StructuredRounding}{$x,w,a$}
		\State {\bf Input}: $x \in (0,1)_{\R}^{\ztbtt}$, $w \in \R^{\ztbtt}$ and $a =\sum_{j \in \ztbtt}x_{j} \in \Z$.
		\State {\bf Output}: $z\in \R^{\ztbtt \times \ztbtt}$ and $s \in \R^{a}$.
		\State Initialize $z=\textbf{0}^{\ztbtt \times \ztbtt}$.
		\State For each $i\in [1,a]$, let $s_{i}$ denote the smallest index such that $\sum_{j \leq s_{i}}x_{j} > i-1$ and let $s_{a+1}=k$.
		\For{$i \in [1,a]$}
		\begin{equation}
		z_{s_{i},j}=
		\begin{cases}
		x_{j}& \text{ if } s_{i}<j<s_{i+1}~,\\
		\sum_{j'\leq s_{i}}x_{j'}-(i-1)& \text{ if } j=s_{i}~,\\
		1-\sum_{s_{i}\leq j' <s_{i+1}} z_{s_{i},j'}& \text{ if } j=s_{i+1}~.\\
		\end{cases}
		\end{equation}
		\EndFor
		\State {\bf Return }$z$ and $s$.
		\EndProcedure
	\end{algorithmic}
	\label{alg:structurerounding}
\end{algorithm}
The next lemma summarizes the quality of the solution produced by \Cref{alg:structurerounding}.
\begin{lemma}\label{lem:structured}
	Given a set of reals $x_{j} \in (0,1)$ for all $j\in \ztbtt$ such that $\sum_{j \in \ztbtt}x_{j} \in \Z$, weights $w_{j}$ for all $j\in \ztbtt$ and exponents $\mmjj \in \Z$ for all $j\in \ztbtt$ \footnote{Here $m_{0}$ need not be equal to zero.}. Using \Cref{alg:structurerounding}, we can efficiently compute a matrix $z \in [0,1 ]_{\R}^{\ztbtt \times \ztbtt}$ such that the following conditions hold,
	\begin{enumerate}
		\item $\sum_{j\in \ztbtt}z_{i,j} \in \{0,1\} \text{ for all } ~i\in \ztbtt$ and $\sum_{i \in \ztbtt}z_{i,j}=x_{j}~ \text{ for all } ~j\in \ztbtt$.
		\item $\sum_{i\in \ztbtt}\left(\sum_{j\in \ztbtt}z_{i,j}\right)w_{i} \leq \sum_{j\in \ztbtt}x_{j}w_{j} +  \max_{j\in \ztbtt}w_{j}$.
		\item $\prod_{j\in \ztbtt}w_{j}^{\mmjj x_{j}} \leq \prod_{i\in \ztbtt} \prod_{j\in \ztbtt}w_{i}^{\mmjj z_{i,j}}$.
	\end{enumerate}
\end{lemma}

\renewcommand{\bttpo}{(\btt+1)}
\renewcommand{\bSp}{\textbf{B}}
\renewcommand{\bS}{\textbf{C}}
\renewcommand{\boo}{t}

We next provide description of \Cref{alg:create}. The algorithm takes input $(\bSp,\bS,\bR,\phi)$ and creates a new probability discretization set $\bR'$ (lines 6-10). The solution $\bSp'$ outputted by the algorithm belongs to $\textbf{Z}^{\phi,frac}_{\bR'}$, has same column sums as $\bSp$ and the value $\bg(\bSp')$ is lower bounded by $\bg(\bSp)$. 
\begin{algorithm}[H]
	\caption{Create New Probability Values}
	\begin{algorithmic}[1]
		\Procedure{$\algcreate$}{$\bSp,\bS,\bR,\phi$}
		\State {\bf Input}: Probability discretization set $\bR$ ($|\bR|=\boo$), profile $\phi$ (let $k$ be the number of distinct frequencies) and $\bSp \in \bZfrac \subseteq \R^{[1,\boo]\times \ztbtt}$ and $\bS\in \R^{[1,\boo]\times \ztbtt}$ such that $\bS_{i,j} \leq \bSp_{i,j}$ for all $i\in [1,\boo]$ and $j\in \ztbtt$. Let $\pvec_{i}$ be the $i$'th element of $\bR$.
		\State {\bf Output}: Probability discretization set $\bR'$ and $\bSp' \in \R^{[1,\boo+\newk]\times \ztbtt}$.
		\State Initialize $\bSp'=\mzero^{[1,\boo+\newk]\times \ztbtt}$.
		\State $\bSp'_{ij}= \bS_{ij}  \text{ for all } i \in [1,\boo],j \in \ztbtt~.$ 
		\For{$j \in \ztbtt$}
		\State Create a new row with probability value $\pvec_{\boo+1+j}=\frac{\sum_{i \in \otboo}(\bSp_{ij}-\bS_{ij})\pvec_{i}}{\sum_{i \in \otboo}(\bSp_{ij}-\bS_{ij})}$.
		\State Assign $\bSp'_{\boo+1+j,j}=\sum_{i \in \otboo}(\bSp_{ij}-\bS_{ij})$.
		\EndFor
		\State Define $\bR' \defeq  \bR \cup \{\pvec_{\boo+1+j} \}_{j \in \ztbtt}$.
		\State {\bf Return:} $\bR'$ and $\bSp'$.
		\EndProcedure
	\end{algorithmic}
	\label{alg:create}
\end{algorithm}

The next lemma summarizes the quality of the solution produced by \Cref{alg:create}.
\begin{lemma}\label{lem:create}
	The solution $(\bR',\bSp')$ returned by \Cref{alg:create} satisfies the following conditions:
	\begin{enumerate}
		\item $\sum_{j\in \ztbtt}\bSp'_{i,j}=\sum_{j\in \ztbtt}\bS_{i,j}$ for all $i\in [1,\boo]$.
		\item For any $i \in [\boo+1,\boo+\newk]$ let $j \in \ztbtt$ be such that $i=\boo+1+j$ then $\bSp'_{\boo+1+j,j'}=0$ for all $j' \in \ztbtt$ and $j' \neq j$. (Diagonal Structure)
		\item For any $i\in [\boo+1,\boo+\newk]$ let $j \in \ztbtt$ be such that $i=\boo+1+j$, then $\sum_{j'\in \ztbtt}\bSp'_{i,j'}=\bSp'_{\boo+1+j,j}=\phi_{j}-\sum_{i' \in [1,\boo]}\bS_{i',j}$.
		\item $\bSp' \in \textbf{Z}^{\phi,frac}_{\bR'}$ and $\sum_{i \in [1,\boo+\newk]} \sum_{j \in \ztbtt} \bSp'_{i,j}=\sum_{i \in [1,\boo]} \sum_{j \in \ztbtt} \bSp_{i,j}$.
		\item Let $\alpha_{i}\defeq \sum_{j\in \ztbtt}\bSp_{i,j} - \sum_{j\in \ztbtt}\bS_{i,j}$ for all $i\in [1,\boo]$ and $\logparam \defeq \max(\sum_{i \in [1,\boo]}(\bSp \onevec)_{i} , \boo \times \btt)$, then 
		$\bg(\bSp') \geq \expo{-O\left(  \sum_{i\in [1,\boo]} \alpha_{i}\log \logparam \right)} \bg(\bSp)~.$
		\item For each $j \in \ztbtt$, the new row corresponds to the probability value,
		$\pvec_{\boo+1+j}=\frac{\sum_{i \in \otboo}(\bSp_{ij}-\bS_{ij})\pvec_{i}}{\sum_{i \in \otboo}(\bSp_{ij}-\bS_{ij})}$.
	\end{enumerate}
	\end{lemma}

\renewcommand{\bSp}{\textbf{S}}
\renewcommand{\bS}{\textbf{A}}
\renewcommand{\boo}{\ell}

In the remainder of this section, we state and analyze our rounding algorithm. Our algorithm works in three steps, and we show that all the solutions produced during the intermediate and final steps all have the desired approximation guarantee. We divide the analysis into three lemmas. Each of the lemmas \ref{lem:stepone}, \ref{lem:steptwo} and \ref{thm:stepthree} analyze the guarantees provided by the intermediate solutions $\bSone$, $\bStwo$ and final solution $\bSext$ respectively. 

\begin{algorithm}[H]
	\caption{Rounding Algorithm}
	\begin{algorithmic}[1]
		\Procedure{Rounding}{$\bSp$}
		\State {\bf Input}: Probability discretization set $\bR$, profile $\phi \in \Phi^{n}$ and $\bSp \in \bZfrac \subseteq \R^{[1,\ell]\times \ztbtt}$.
		\State {\bf Output}: Probability discretization set $\bRext$ and $\bSext$.
		\State {\bf Step 1}:
		\State Initialize $\ma=\mzero^{[1,\boo]\times \ztbtt}$. Let $\pvec_{i}$ be the $i$'th element of $\bR$.
		\State Define $\bH\defeq \{i\in [1,\ell]~|~\pvec_{i}>\tsqn \}$ and $\bL\defeq \{i\in [1,\ell]~|~\pvec_{i}\leq \tsqn \}$.
		\State $\ma_{ij}= \floor{\bSp_{ij}}  \text{ for all } i \in \bH,j \in \ztbtt~.$ 
		\State $\ma_{ij}= \bSp_{i,j}\frac{\floor{\sum_{i\in \bL}\bSp_{i,j}}}{\sum_{i\in \bL}\bSp_{i,j}}  \text{ for all } i \in \bL,j \in \ztbtt~.$
		\State $(\bSone,\bRone)=\algcreate(\bSp,\ma,\bR)$.
		\State {\bf Step 2}:
		\State Note $|\bRone|=\ell+\newk$ and $\bSone \subseteq \R^{[1,\ell+\newk]\times [0,k]}$. Let $\pvecone_{i}$ be the $i$'th element of $\bRone$. 
		\State Let $\bHone\defeq \{i\in [1,\ell+\newk]~|~\pvecone_{i}>\tsqn \}$ and $\bLone\defeq \{i\in [1,\ell+\newk]~|~\pvecone_{i}\leq \tsqn \}$.
		\State Define $\bAone=\mzero^{[1,\boo+\newk]\times \ztbtt}$.
		\State $\bAone_{ij}= \bSone_{ij}  \quad \text{ for all } i \in \bHone,j \in \ztbtt~.$
		\State $\bAone_{ij}= \bSone_{ij} \frac{\floor{(\bSone \onevec)_{i}}}{(\bSone \onevec)_{i}} \quad \text{ for all } i \in \bLone,j \in \ztbtt~.$
		\State $(\bStwo,\bRtwo)=\algcreate(\bSone,\bAone,\bRone)$.
		\State {\bf Step 3}:
		\State Note $|\bRtwo|=\ell+2\newk$ and $\bStwo \subseteq \R^{[1,\ell+2\newk]\times [0,k]}$. Let $\pvectwo_{i}$ be the $i$'th element of $\bRtwo$.
		\State Let $w,x \in \R^{\ztbtt}$, such that $w_{j}\defeq \pvectwo_{\boo+\newk+1+j}$ and $x_{j}\defeq \bStwo_{\boo+\newk+1+j}-\floor{\bStwo_{\boo+\newk+1+j}}$ for all $j \in \ztbtt$. Define $a\defeq \sum_{j\in \ztbtt}x_{j}$.
		\State Let $(z,s)\defeq \mathrm{StructuredRounding}(x,w,a)$.
		\State Initialize $\bSext = 0^{[1,\ell+2\newk]\times [0,k]}$.
		\State Assign $\bSext_{i,j}=\bStwo_{i,j}$ for all $i \in [1,\boo+\newk]$ and $j \in \ztbtt$.
		\State Assign $\bSext_{\boo+\newk+1+j,j'}=\floor{\bStwo_{\boo+\newk+1+j,j'}}+z_{j,j'}$ for all $j,j' \in \ztbtt$.
		\State Define $\bRext \defeq \{\frac{\pvectwo_{i}}{1+\tsqn}~|~\text{for all } i \in [1,\boo+2\newk]  \}$.
		\State \textbf{return} $\bRext$ and $\bSext$.
		\EndProcedure
	\end{algorithmic}
	\label{alg:rounding}
\end{algorithm}

The next lemma summarizes the quality of the intermediate solution $(\bSone,\bRone)$ produced by Step 1 of \Cref{alg:rounding}.
\begin{lemma}\label{lem:stepone}
	The solution $(\bSone,\bRone)$ returned by the step 1 of \Cref{alg:rounding} satisfies the following:
	\begin{enumerate}
		\item $(\bSone \onevec)_{i} \in \Z$ for all $i\in \bHone$.
		\item $\bSone \in \bZone$ and $\sum_{i \in [1,\boo+\newk]} \sum_{j \in \ztbtt} \bSone_{i,j}=\sum_{i \in [1,\boo]} \sum_{j \in \ztbtt} \bSp_{i,j}$.
		\item $\bg(\bSone) \geq\expo{-O\left(\left(\tsqni+\btt \right)\log \logparam \right)} \bg(\bSp)$, where $\logparam= \max (\sum_{i \in [1,\boo]}(\bSp \onevec)_{i}, \boo \times \btt )$.
	\end{enumerate}
	\end{lemma}

Using \Cref{lem:stepone} we now provide the guarantees for the solution $\bStwo$ returned by the step 2 of \Cref{alg:rounding}.
\begin{lemma}\label{lem:steptwo}
	The solution $(\bStwo,\bRtwo)$ returned by the step 2 of \Cref{alg:rounding} satisfies the following,
	\begin{enumerate}
		\item $(\bStwo \onevec)_{i} \in \Z$ for all $i \in [1,\boo+\newk]$.
		\item $\bStwo_{\boo+\newk+1+j,j'}=0$ for all $j,j' \in \ztbtt$ and $j\neq j'$ (Diagonal Structure).
		\item $\bStwo \in \bZtwo$ and $\sum_{i \in [1,\boo+2\newk]} \sum_{j \in \ztbtt} \bStwo_{i,j}=\sum_{i \in [1,\boo+\newk]} \sum_{j \in \ztbtt} \bSone_{i,j}$.
		\item $\sum_{i \in [\boo+\newk+1,\boo+2\newk]} (\bStwo \onevec)_{i} \in \Z$.
		\item For any $j \in \ztbtt$, $\pvectwo_{\boo+\newk+1+j} \leq \tsqn$.
		\item $\bg(\bStwo) \geq\expo{-O\left(\left(\tsqni+\boo+\btt \right)\log \logparam \right)} \bg(\bSp)$.
	\end{enumerate}
\end{lemma}

Using \Cref{lem:steptwo} we now provide the guarantees for the final solution $\bSext$ returned by \Cref{alg:rounding}.
\begin{thm}\label{thm:stepthree}
	The final solution returned $(\bSext,\bRext)$ by \Cref{alg:rounding} satisfies the following,
	\begin{enumerate}
		\item $\bSext \in \bZext$.
		\item $\bg(\bSext) \geq \expo{-O\left(\left(\tsqni+\boo+k+\tsqn n \right)\log \logparam \right)} \bg(\bSp)$.
	\end{enumerate}
	\end{thm}

	\subsection{Combining everything together}\label{subsec:combine}
	Here we combine the analysis from previous two sections to provide an efficient algorithm to compute an $\expo{\sqrt{n}\log n}$ approximate PML distribution. The main contribution of this section is to define a probability discretization set $\bR$ that guarantees existence of a discrete pseudo-distrbution $\bq$ with respect to $\bR$ that is also an $\expo{\sqrt{n}\log n}$ approximate PML pseudo-distribution. We further use this probability discretization set $\bR$ and combine it with results from the previous two sections to finally output an $\expo{\sqrt{n}\log n}$ approximate PML distribution. In this direction, we first provide definition of $\bR$ that has desired guarantees and such a set $\bR$ was already constructed in \cite{CSS19} and we formally state results from \cite{CSS19} that help us define such a set $\bR$. 
	
\begin{lemma}[Lemma 4.1 in \cite{CSS19}]\label{lemminmain}
	For any profile $\phi \in \Phi^{n}$, there exists a distribution $\bp'' \in \simplex$ such that $\bp''$ is an $\expo{-6}$-approximate PML distribution and $\min_{x \in \bX:\bp''_x \neq 0}\bp''_x \geq \frac{1}{2n^2}$.
\end{lemma} 
 The above lemma allows to define a region in which our approximate PML takes all its probability values and we use idea similar to \cite{CSS19} to define this region.

Let $\bR\defeq\{(1+\epso)^{1-i}\}_{i \in [\boo]}$ be a discretization of probability space, where $\boo=O(\frac{\log n}{\epso})$ is the smallest integer such that $\frac{1}{4n^2} \leq (1+\epso)^{1-\boo}\leq \frac{1}{2n^2}$ for some $\epso \in (0,1)$. Fix any arbitrary order for the elements of set $\bR$, we use $\ri$ to denote the $i$'th element of this set. 
We next state a result in \cite{CSS19} that captures the effect of this discretization.
\renewcommand{\bp}{\textbf{p}}
\begin{lemma}[Lemma 4.4 in \cite{CSS19}]\label{lem:probdisc}
	For any profile $\phi \in \Phi^{n}$ and distribution $\bp \in \simplex$, its discrete pseudo-distribution $\bq=\disc(\bp) \in \dsimplex$ satisfies:
	$$\probpml(\bp,\phi) \geq \probpml(\bq,\phi) \geq \expo{-\epso n}\probpml(\bp,\phi)~.$$
\end{lemma}

We are now ready to state our final algorithm. Following this algorithm, we prove that it returns an approximate PML distribution.
\begin{algorithm}[H]
	\caption{Algorithm for approximate PML}\label{euclid}
	\begin{algorithmic}[1]
		\Procedure{Approximate PML}{$\phi,\bR$}
		\State {\bf Input}: Profile $\phi \in \Phi^n$ and probability discretization set $\bR$.
		\State {\bf Output}: A distribution $\bpapprox$.
		\State Solve $\bSp=\argmax_{\bS\in \bZfrac} \log \bg(\bS)$. 
		\State Use \Cref{alg:rounding} to round the fractional solution $\bSp$ to integral solution $\bSext \in \bZext$.
		\State Construct discrete pseudo-distribution $\qext$ with respect to $\bSext$ (See \Cref{defn:distS}).
		\State \textbf{return} $\bpapprox\defeq \frac{\qext}{\|\qext\|_1}$.
		\EndProcedure
	\end{algorithmic}
\label{alg:final}
\end{algorithm}
\renewcommand{\pml}{\textbf{p}_{\mathrm{pml}}}
\secondthm*
\begin{proof}
	Choose $\epso=\frac{\log n}{\sqrt{n}}$ and let the probability discretization space $\bR\defeq\{(1+\frac{1}{\sqrt{n}})^{1-i}\}_{i \in [\boo]}$ and $\boo\defeq |\bR|$ be the smallest integer such that $\frac{1}{2n^2} \geq (1+\frac{1}{\sqrt{n}})^{1-\boo} \geq \frac{1}{4n^2}$ and therefore $\boo \in O(\sqrt{n})$. Let $\pvec_{i}$ be the $i$'th element of set $\bR$ and we have $\pvec_{i} \geq \frac{1}{4n^2}$.
	
	Given profile $\phi$, let $\pml$ be the PML distribution. Define $\dpml \defeq \floor{\pml}_{\bR}$ and by \Cref{lem:probdisc} (and choice of $\epso$) we have,
	\begin{equation}\label{eq:ddone}
	\probpml(\dpml,\phi)\geq \expo{- O(\sqrt{n} \log n)}\probpml(\pml,\phi)~.
	\end{equation}

	Let $\bSp \defeq \argmax_{\bS\in \bZfrac}\bg(\bS)$, then by \Cref{lem:dpmlapprox} we have,
	\begin{equation}\label{eq:ddtwo}
	\max_{\bq \in \dsimplex}\probpml(\bp,\phi) \leq \expo{\bigO{k \log \frac{\Ns}{k}}} \cdot \cphi \cdot \bg(\bSp)~.
	\end{equation}
	Note $\dpml \in \dsimplex$, therefore $\probpml(\dpml,\phi) \leq \max_{\bq \in \dsimplex}\probpml(\bq,\phi)$ and further combined with equations \ref{eq:ddone} and \ref{eq:ddtwo} we have,
	\begin{equation}\label{eq:ddthree}
	\probpml(\pml,\phi) \leq \expo{\bigO{k \log \frac{\Ns}{k} +\sqrt{n} \log n}} \cdot \cphi \cdot \bg(\bSp)~.
	\end{equation}
	Let $\bSext$ and $\bRext$ be the solution returned by \Cref{alg:rounding}, then by the second condition of \Cref{thm:stepthree} we have,
	\begin{equation}\label{eq:ddfour}
	\bg(\bSext) \geq \expo{-O\left(\left(\tsqni+\boo+k+\tsqn n \right)\log \logparam \right)} \bg(\bSp)
	\end{equation}
	Combining equations \ref{eq:ddthree} and \ref{eq:ddfour} we have,
	\begin{equation}\label{eq:z}
		\probpml(\pml,\phi) \leq \expo{\bigO{k \log \frac{\Ns}{k} +\sqrt{n} \log n + \left(\tsqni+\boo+k+\tsqn n \right)\log \logparam }} \cdot \cphi \cdot \bg(\bSext)~.
	\end{equation}
	We now simplify the above expression by providing the bounds and values for parameters $k, \ell,\gamma, \Ns$ and $\Delta$. We choose $\gamma=\frac{1}{\sqrt{n}}$ and recall $\boo \in O(\sqrt{n})$. Given $n$ samples, the number of distinct frequencies in upper bounded by $\sqrt{n}$ and therefore $k \leq \sqrt{n}$. By \Cref{lemminmain}, up to constant multiplicative loss we can assume that the minimum non-zero probability value of our approximate PML distribution is at least $\frac{1}{4n^2}$ and therefore the support $\Ns \leq 4n^2$. Recall by the third condition of \Cref{lem:stepone}, we have $\logparam= \max (\sum_{i \in [1,\boo]}(\bSp \onevec)_{i}, \boo \times \btt )$. The condition $\bSp \in \bZfrac$ implies $\sum_{i \in [1,\boo]}\pvec_{i}(\bSp \onevec)_{i} \leq 1$ and further using $\pvec_{i} \geq \frac{1}{4n^2}$ for all $i \in [1,\boo]$ we have $\sum_{i \in [1,\boo]}(\bSp \onevec)_{i} \leq 4n^2$. Therefore $\logparam \leq \max (4n^2,\ell \times k) \in O(n^2)$.
	
	Substituting these values in \Cref{eq:z} we get,
	\begin{equation}\label{eq:ddfive}
	\probpml(\pml,\phi) \leq \expo{\bigO{\sqrt{n} \log n }} \cdot \cphi \cdot \bg(\bSext)~.
	\end{equation}
	
	By the first condition of \Cref{thm:stepthree} we have $\bSext \in \bZext$. Let $\qext$ be the discrete pseudo-distribution with respect to $\bSext$, then the condition $\bSext \in \bZext$ further implies $\bSext \in \bZqext$ and combined with \Cref{pmlprob:approx} we have,
	\begin{equation}\label{eq:ddsix}
	\expo{-O(k \log (\Ns+n))}\cdot	\cphi \cdot	\bg(\bSext) \leq \probpml(\qext,\phi)
	\end{equation}
	Combining equations \ref{eq:ddfive}, \ref{eq:ddsix}, $k\leq \sqrt{n}$ and $\Ns \leq 4n^2$ we have,
	\begin{equation}\label{eq:ddseven}
	\probpml(\qext,\phi) \geq \expo{-\bigO{\sqrt{n} \log n }} \probpml(\pml,\phi)~.
	\end{equation}
	Define $\bpapprox\defeq \frac{\qext}{\|\qext\|_1}$, then $\bpapprox$ is a distribution, $\probpml(\bpapprox,\phi) \geq \probpml(\qext,\phi)$ (because $\qext$ is a pseudo-distribution and $\|\qext\|_1\leq 1$) and combined with \Cref{eq:ddseven} we get,
	\begin{equation}
	\probpml(\bpapprox,\phi) \geq \expo{-\bigO{\sqrt{n} \log n }} \probpml(\pml,\phi)~.
	\end{equation}
	Therefore $\bpapprox$ is an $\expo{-\bigO{\sqrt{n} \log n }}$-approximate PML distribution. 
	
	In the remainder of the proof we argue about the running time of our final algorithm for approximate PML. Step 4 of the algorithm, that is the convex program $\argmax_{\bS\in \bZfrac} \log \bg(\bS)$ can be solved in $\otilde(k^2 \ell)$ time (See \Cref{thm:convrt}). \Cref{alg:create} ($\algcreate$) and \Cref{alg:structurerounding} ($\mathrm{StructuredRounding}$) can be implemented in $\otilde(k \ell)$ and $\otilde(k^2)$ time respectively; therefore, the \Cref{alg:rounding} (Rounding algorithm) can be implemented in $\otilde(k \ell)$ time. Combining everything together our final algorithm (\Cref{alg:final}) can be implemented in $\otilde(k^2 \ell)$ time. Further using $k,\ell \in O(\sqrt{n})$, we conclude the proof.
\end{proof}

\subsection{Missing Proofs from \Cref{subsec:algmain}}\label{subsec:missing}
Here we provide the proofs for all the lemmas and theorems in \Cref{subsec:algmain}
\begin{proof}[Proof of \Cref{lem:structured}]
	Without loss of generality assume $w_{0} \geq w_{1} \geq w_{2} \dots \geq w_{k}$. Let $a\defeq \sum_{j\in \ztbtt}x_{j} $, we invoke \Cref{alg:structurerounding} with inputs $(x,w,a)$. Let $s\in \Z^{a}$ and $z\in \R^{\ztbtt \times \ztbtt}$ be the output of \Cref{alg:structurerounding}. We now provide the proof for the three conditions in the lemma.
	
	{\bf Condition 1}: By construction of \Cref{alg:structurerounding}, for any $s \in \{s_{i}\}_{i\in [1,a]}$ we have $\sum_{j\in \ztbtt}z_{s,j}=1$ (Line 6) and for any other $s \in \ztbtt \backslash \{s_{i}\}_{i\in [1,a]}$ we have $\sum_{j\in \ztbtt}z_{s,j}=0$. Therefore the first part of condition 1 holds.
	
	For any $j \in \ztbtt$, one of the following two cases holds,
	\begin{enumerate}
		\item If $j\in \{s_{i}\}_{i\in [1,a]}$ and in this case let $i \in [1,a]$ be such that $s_{i}=j$. By line 6 (third case) of the algorithm we have,
		\begin{equation}\label{eq:fone}
		z_{s_{i-1},j}=1-\left(\sum_{j'\leq s_{i-1}}x_{j'}-(i-2)+\sum_{s_{i-1}<j'<s_{i}}x_{j'} \right)=(i-1)-\sum_{j'<s_{i}}x_{j'}~.
		\end{equation}
		We now analyze the term $\sum_{i'\in \ztbtt}z_{i',j}$,
		$$\sum_{i'\in \ztbtt}z_{i',j}=z_{s_{i},j}+z_{s_{i-1},j}=\sum_{j'\leq s_{i}}x_{j'}-(i-1)+(i-1)-\sum_{j'<s_{i}}x_{j'}=x_{s_{i}}=x_{j}~.$$
		The first equality follows because for $i' \in \ztbtt \backslash \{s_{i},s_{i-1}\}$ we have $z_{i',j}=0$ and this follows by the second and third case in line 6 of the algorithm. In the second equality we substituted values for $z_{s_{i},s_{i}}$ and $z_{s_{i-1},s_{i}}$ using second case (Line 6) and \Cref{eq:fone} respectively.
		\item Else $j\in \ztbtt\backslash \{s_{i}\}_{i\in [1,a]}$, and in this case let $i \in [1,a]$ be such that $s_{i}<j<s_{i+1}$. Then by the first case in line 6 of the algorithm we have,
		$$\sum_{i'\in \ztbtt}z_{i',j}=z_{s_{i},j}=x_{j}~.$$
	\end{enumerate}
	{\bf Condition 2:} Consider $\sum_{i\in \ztbtt}\left(\sum_{j\in \ztbtt}z_{i,j}\right)w_{i}$,
	\begin{equation}\label{eq:gone}
	\begin{split}
	\sum_{i\in \ztbtt}\left(\sum_{j\in \ztbtt}z_{i,j}\right)w_{i}&=\sum_{i\in [1,a]}\left(\sum_{s_{i}\leq j  \leq s_{i+1}}z_{\si,j}\right)w_{\si} \leq \sum_{i\in [1,a]}\sum_{s_{i}\leq j  \leq s_{i+1}}z_{\si,j}(w_{j}+w_{s_{i}}-w_{s_{i+1}})\\
	&\leq \sum_{i\in [1,a]}\sum_{s_{i}\leq j  \leq s_{i+1}}z_{\si,j}w_{j}+\sum_{i\in [1,a]}\sum_{s_{i}\leq j  \leq s_{i+1}}z_{\si,j}(w_{s_{i}}-w_{s_{i+1}})\\
	&= \sum_{i\in [1,a]}\sum_{j  \in \ztbtt}z_{\si,j}w_{j}+\sum_{i\in [1,a]}\sum_{s_{i}\leq j  \leq s_{i+1}}z_{\si,j}(w_{s_{i}}-w_{s_{i+1}})~.\\
	\end{split}
	\end{equation}
	The first equality follows because rest of the other entries are zero. In the second inequality we used $j \leq s_{i+1}$ and therefore $w_{j} \geq w_{s_{i+1}}$ by our assumption at the beginning of the proof. In the remainder, we simplify both the terms. Consider the first term in the final expression above,
	\begin{equation}\label{eq:gtwo}
	\sum_{i\in [1,a]}\sum_{j  \in \ztbtt}z_{\si,j}w_{j}=\sum_{j\in \ztbtt}w_{j}\sum_{i  \in [1,a]}z_{\si,j}=\sum_{j\in \ztbtt}w_{j}x_{j}~.
	\end{equation}
	In the first equality we interchanged the summations. In the second equality we used $\sum_{i  \in [1,a]}z_{\si,j}=\sum_{i' \in \ztbtt}z_{i',j}$ and further invoked condition 1 of the lemma. Now consider the second term in the final expression of \Cref{eq:gone},
	\begin{equation}\label{eq:gthree}
	\begin{split}
	\sum_{i\in [1,a]}\sum_{s_{i}\leq j  \leq s_{i+1}}z_{\si,j}(w_{s_{i}}-w_{s_{i+1}})&= \sum_{i\in [1,a]}(w_{s_{i}}-w_{s_{i+1}})\sum_{s_{i}\leq j  \leq s_{i+1}}z_{\si,j}
	= \sum_{i\in [1,a]}(w_{s_{i}}-w_{s_{i+1}})\\
	&= (w_{s_{1}}-w_{s_{x+1}})
	\leq \max_{j\in \ztbtt}w_{j}~.
	\end{split}
	\end{equation}
	The second equality follows by line 6 of the algorithm.	Condition 2 follows by combining equations \ref{eq:gone}, \ref{eq:gtwo} and \ref{eq:gthree}.
	
	{\bf Condition 3:} First we show that $z_{i,j}>0$ implies $i\leq j$. Consider $j \in \ztbtt$,
	\begin{enumerate}
		\item If $j\in \{s_{i}\}_{i\in [1,a]}$, in this case let $i \in [1,a]$ be such that $s_{i}=j$. Then by the second and third case in line 6 of the algorithm we have,
		$z_{i',j}>0$ implies $i'\in \{s_{i},s_{i-1}\}$. Further, using $s_{i-1}<s_{i}$ and $s_{i}=j$ we have $i' \leq j$.
		\item Else $j\in \ztbtt\backslash \{s_{i}\}_{i\in [1,a]}$ and in this case let $i \in [1,a]$ be such that $s_{i}<j<s_{i+1}$. Then by the first case in line 6 of the algorithm we have, $z_{i',j}>0$ implies $i'=s_{i}$. Further, using $s_{i}<j$ we have $i'<j$.
	\end{enumerate}
	Using the above implication we have,
	\begin{equation}
	\begin{split}
	\prod_{j\in \ztbtt}w_{j}^{\mmjj x_{j}}&=\prod_{j\in \ztbtt}w_{j}^{\mmjj \sum_{i\in \ztbtt}z_{i,j}}=\prod_{i\in \ztbtt}\prod_{j\in \ztbtt}w_{j}^{\mmjj z_{i,j}}\leq \prod_{i\in \ztbtt}\prod_{j\in \ztbtt}w_{i}^{\mmjj z_{i,j}}
	\end{split}
	\end{equation}
	In the first equality we used $x_{j}=\sum_{i\in \ztbtt}z_{i,j}$ for all $j \in \ztbtt$ (Condition 1). In the final inequality, we used the result $z_{i,j}>0$ implies $i\leq j$ and further combined it with the assumption at the begining of the proof, that is, $w_{i}\geq w_{j}$ for all $i,j \in \ztbtt$ and $i \leq j$.
\end{proof}

\renewcommand{\bttpo}{(\btt+1)}
\renewcommand{\bSp}{\textbf{B}}
\renewcommand{\bS}{\textbf{C}}
\renewcommand{\boo}{t}

\begin{proof}[Proof of \Cref{lem:create}] 
	Define $\phi_{0}\defeq \sum_{i \in [1,\boo]} \bSp_{i,0}$. In the following we provide the proof for each case.
	
	{\bf Condition 1:} For each $i\in [1,\boo]$, $\bSp'_{i,j} =\bS_{i,j}$ for all $j\in \ztbtt$ and the first condition holds.
	
	{\bf Condition 2:} Note $\bSp'$ is initialized to a zero matrix (Line 4). Further for any $i\in [\boo+1,\boo+\newk]$ let $j \in \ztbtt$ be such that $i=\boo+1+j$, then the algorithm only updates the $\bSp'_{\boo+1+j,j}$'th entry in the $i$'th row and keeps rest of the entries unchanged. Therefore the second condition holds.
	
	{\bf Condition 3:} For each $i\in [\boo+1,\boo+\newk]$ let $j \in \ztbtt$ be such that $i=\boo+1+j$, then $\sum_{j'\in \ztbtt}\bSp'_{i,j'}=\bSp'_{\boo+1+j,j}=\sum_{i' \in \otboo}(\bSp_{i',j}-\bS_{i',j})=\phi_{j}-\sum_{i' \in [1,\boo]}\bS_{i',j}$. The first equality holds because of the Condition 2. The third equality follows from the Line 8 of the algorithm. The last equality holds because $\bSp \in \bZfrac$ and we have $\sum_{i \in [1,\ell]}\bSp_{i,j}=\phi_{j}$.
	
	{\bf Condition 4:} Here we provide the proof for $\bSp' \in \textbf{Z}^{\phi,frac}_{\bR'}$. 
	For any $j\in \ztbtt$, we first show that $\sum_{i\in [1,\boo+\newk]}\bSp'_{i,j}=\phi_{j}$.
	\begin{align*}
	\sum_{i\in [1,\boo+\newk]}\bSp'_{i,j}&=\sum_{i \in [1,\boo]}\bSp'_{i,j}+\sum_{i \in [\boo+1,\boo+\newk]}\bSp'_{i,j}
	=\sum_{i \in [1,\boo]}\bS_{i,j}+\bSp'_{\boo+1+j,j}\\
	&=\sum_{i \in [1,\boo]}\bS_{i,j}+\phi_{j}-\sum_{i \in [1,\boo]}\bS_{i,j}=\phi_{j}\\
	\end{align*}
	The second equality follows because $\bSp'_{i,j}=\bS_{i,j}$ for all $i\in [1,\boo]$ and $j\in \ztbtt$ (Line 6) and $\sum_{i \in [\boo+1,\boo+\newk]}\bSp'_{i,j}=\bSp'_{\boo+1+j,j}$ (Condition 2). The third equality follows from the Condition 3.
	
	We next show that $\sum_{i\in [1,\boo+\newk]}\pvec_{i} \left( \sum_{j\in \ztbtt}\bSp'_{i,j} \right) \leq 1$.
	\begin{equation}
	\begin{split}
	\sum_{i\in [1,\boo+\newk]}\pvec_{i} \left( \sum_{j\in \ztbtt}\bSp'_{i,j} \right)&=\sum_{i\in [1,\boo]}\pvec_{i} \left( \sum_{j\in \ztbtt}\bSp'_{i,j} \right)+ \sum_{j\in \ztbtt}\pvec_{\boo+1+j} \bSp'_{\boo+1+j,j}\\
	&=\sum_{i\in [1,\boo]}\pvec_{i} \left( \sum_{j\in \ztbtt}\bS_{i,j} \right)+ \sum_{j\in \ztbtt}\frac{\sum_{i \in \otboo}(\bSp_{ij}-\bS_{ij})\pvec_{i}}{\sum_{i \in \otboo}(\bSp_{ij}-\bS_{ij})} \left(\sum_{i \in \otboo}(\bSp_{i,j}-\bS_{i,j}) \right)\\
	&=\sum_{i\in [1,\boo]}\pvec_{i} \left( \sum_{j\in \ztbtt}\bS_{i,j} \right)+ \sum_{j\in \ztbtt}\sum_{i \in \otboo}(\bSp_{ij}-\bS_{ij})\pvec_{i}\\
	&=\sum_{i\in [1,\boo]}\pvec_{i} \left( \sum_{j\in \ztbtt}\bSp_{i,j} \right) \leq 1\\
	\end{split}
	\end{equation}
	In the first equality, we divided the summation into two parts and for the second part we used Condition 3. In the second equality we used Line 7 and 8 of the algorithm. In the third and fourth equality we simplified the expression. In the final inequality we used $\bSp \in \bZfrac$.
	
	Combining all the conditions together we have $\bSp' \in \textbf{Z}^{\phi,frac}_{\bR'}$. In the remainder we show that $\sum_{i \in [1,\boo+\newk]} \sum_{j \in \ztbtt} \bSp'_{i,j}=\sum_{i \in [1,\boo]} \sum_{j \in \ztbtt} \bSp_{i,j}$.
	
	Recall we already showed that $\sum_{i\in [1,\boo+\newk]}\bSp'_{i,j}=\phi_{j}$ for all $j \in \ztbtt$. Recall $\phi_{0}= \sum_{i \in [1,\boo]} \bSp_{i,0}$ and $\bSp \in \bZfrac$ implies $\phi_{j}=\sum_{i\in [1,\boo]}\bSp_{i,j}$ for all $j \in [1,\btt]$. Therefore we have, $$\sum_{i \in [1,\boo+\newk]} \sum_{j \in \ztbtt} \bSp'_{i,j}=\sum_{i \in [1,\boo]} \sum_{j \in \ztbtt} \bSp_{i,j}$$
	
	{\bf Condition 5:} We first provide the explicit expressions for $\bg(\bSp')$ and $\bg(\bSp)$ below: 
	$$ \bg(\bSp')=\left(\prod_{i \in \otboo}\pp_i^{(\bSp' \mvec)_i}\frac{\expo{(\bSp' \onevec)_i \log (\bSp' \onevec)_i}}{\prod_{j \in \ztbtt}\expo{\bSp'_{ij}\log \bSp'_{ij}}}\right)\left(\prod_{j \in \ztbtt} \pp_{\boo+1+j}^{\nn_j\bSp'_{\boo+1+j,j} }\cdot 1\right)$$
	$$ \bg(\bSp)=\prod_{i \in \otboo}\left(\pp_i^{(\bSp \mvec)_i}\frac{\expo{(\bSp \onevec)_i \log (\bSp \onevec)_i}}{\prod_{j \in \ztbtt}\expo{\bSp_{ij}\log \bSp_{ij}}}\right)$$
	Note in the expression for $\bg(\bSp')$ we used Condition 2. In the above two definitions for $\bg(\bSp')$ and $\bg(\bSp)$, we refer to the expression involving $\pvec_{i}$'s as the probability term and the rest as the counting term. We start the analysis of Condition 5 by first bounding the probability term:
	\begin{equation}\label{eq:probterm}
	\begin{split}
	\prod_{i \in \otboo}&\pp_i^{(\bSp \mvec)_{i}}=\left(\prod_{i \in \otboo}\pp_i^{(\bSp' \mvec)_i}\right)\left(\prod_{i \in \otboo}\pp_i^{\sum_{j \in \ztbtt}\nn_j(\bSp_{ij}-\bSp'_{ij})}\right)
	=\left(\prod_{i \in \otboo}\pp_i^{(\bSp' \mvec)_i}\right)\left(\prod_{j \in \ztbtt} \prod_{i \in \otboo}\pp_i^{\nn_j(\bSp_{ij}-\bSp'_{ij})}\right)\\
	&=\left(\prod_{i \in \otboo}\pp_i^{(\bSp' \mvec)_i}\right)\left(\prod_{j\in \ztbtt} \left(\prod_{i \in \otboo}\pp_i^{(\bSp_{ij}-\bSp'_{ij})}\right)^{\nn_j}\right)=\left(\prod_{i \in \otboo}\pp_i^{(\bSp' \mvec)_i}\right)\left(\prod_{j \in \ztbtt} \left(\prod_{i \in \otboo}\pp_i^{(\bSp_{ij}-\bS_{ij})}\right)^{\nn_j}\right)\\
	&\leq \left(\prod_{i \in \otboo}\pp_i^{(\bSp' \mvec)_i}\right)\left(\prod_{j \in \ztbtt} \left(\frac{\sum_{i \in \otboo}\pp_i(\bSp_{ij}-\bS_{ij})}{\sum_{i \in \otboo}(\bSp_{ij}-\bS_{ij})}\right)^{\nn_j\sum_{i \in \otboo}(\bSp_{ij}-\bS_{ij})}\right)\\
	&\leq \left(\prod_{i \in \otboo}\pp_i^{(\bSp' \mvec)_i}\right)\left(\prod_{j \in \ztbtt} \pp_{\boo+1+j}^{\nn_j\bSp'_{\boo+1+j,j}}\right)\\
	\end{split}
	\end{equation}
	The first three inequalities simplify the expression. The fourth equality follows because $\bSp'_{i,j} =\bS_{i,j}$ for all $i\in [1,\boo]$ and $j\in \ztbtt$. The fifth inequality follows from AM-GM inequality. The final expression above is the probability term associated with $\bSp'$ and the equation above shows that our rounding procedure only increases the probability term and it remains to bound the counting term.
	\begin{equation}\label{eq:countterm}
	\begin{split}
	\frac{\bg(\bSp')}{\bg(\bSp)} & \geq \prod_{i \in \otboo}\frac{\expo{(\bSp' \onevec)_i \log (\bSp' \onevec)_i-(\bSp \onevec)_{i} \log (\bSp \onevec)_{i}}}{\prod_{j \in \ztbtt}\expo{\bSp'_{ij}\log \bSp'_{ij}-\bSp_{ij}\log \bSp_{ij}}}\\
	&=\prod_{i \in \otboo}\frac{\expo{(\bS \onevec)_i \log (\bS \onevec)_i-(\bSp \onevec)_{i} \log (\bSp \onevec)_{i}}}{\prod_{j \in \ztbtt}\expo{\bS_{ij}\log \bS_{ij}-\bSp_{ij}\log \bSp_{ij}}}~.
	\end{split}
	\end{equation}
	Consider the numerator in the above expression, for each $i\in \otboo$ let $s_{i}\defeq(\bS \onevec)_i$, then
	\begin{equation}
	\begin{split}
	\prod_{i \in \otboo}\expo{(\bS \onevec)_i \log (\bS \onevec)_i-(\bSp \onevec)_{i} \log (\bSp \onevec)_{i}}&=\prod_{i \in \otboo}\expo{s_i \log s_i-(s_{i}+\alpha_{i}) \log (s_{i}+\alpha_{i})}\\
	&=\prod_{i \in \otboo}\expo{s_{i}\log \frac{s_{i}}{s_{i}+\alpha_{i}} -\alpha_{i}\log (s_{i}+\alpha_{i})}\\
	&\geq \prod_{i \in \otboo}\expo{s_{i}\frac{-\alpha_{i}}{s_{i}} -\alpha_{i}\log (s_{i}+\alpha_{i})}\\
	&\geq \expo{-O\left(\log (\sum_{i\in [1,\boo]}s_{i})\sum_{i\in \otboo}\alpha_{i} \right)}\\
	&\geq \expo{-O\left(\sum_{i\in \otboo}\alpha_{i} \log \logparam \right)}~.
	\end{split}
	\end{equation}
	In the third inequality we used $\log (1+x) \geq \frac{x}{1+x}$ for all $x \geq -1$. The final inequality follows because $\sum_{i\in [1,\boo]}s_{i} \leq\sum_{i\in [1,\boo]}(\bSp \onevec)_{i} \leq  \logparam$.
	Now consider the denominator in the above expression, let $\alpha_{i,j}=\bSp_{i,j}-\bS_{i,j}$ for all $i\in \otboo$ and $j\in \ztbtt$, then 
	\begin{equation}\label{eq:derive}
	\begin{split}
	\prod_{i \in \otboo}\prod_{j \in \ztbtt}\expo{\bS_{ij}\log \bS_{ij}-\bSp_{ij}\log \bSp_{ij}}&=\prod_{i \in \otboo}\prod_{j \in \ztbtt}\expo{\bS_{ij}\log \bS_{ij}-(\bS_{ij}+\alpha_{i,j})\log (\bS_{ij}+\alpha_{i,j})}\\
	&=\prod_{i \in \otboo}\prod_{j \in \ztbtt}\expo{\bS_{ij}\log \frac{\bS_{ij}}{\bS_{ij}+\alpha_{i,j}}-\alpha_{i,j}\log (\bS_{ij}+\alpha_{i,j})}\\
	&\leq \prod_{i \in \otboo}\prod_{j \in \ztbtt}\expo{-\alpha_{i,j}\log (\bS_{ij}+\alpha_{i,j})}\\ 
	&\leq \prod_{i \in \otboo}\prod_{j \in \ztbtt}\expo{-\alpha_{i,j}\log \alpha_{i,j}}\leq \expo{O\big(\log(\boo \times \btt)\sum_{i\in \otboo} \alpha_{i} \big)}\\
	& \leq \expo{O\left(\sum_{i\in \otboo} \alpha_{i} \log \logparam\right)}~.
	\end{split}
	\end{equation}
	In the third inequality we used $\alpha_{i,j} \geq 0$ and therefore $\bS_{ij}\log \frac{\bS_{ij}}{\bS_{ij}+\alpha_{i,j}} \leq 0$. In the fourth inequality we used $\log (\bS_{ij}+\alpha_{i,j}) \geq \log \alpha_{i,j}$. In the fifth inequality we used $\sum_{j \in \ztbtt} \alpha_{i,j}=\alpha_{i}$ for all $i\in \otboo$ and further $\sum_{i\in \otboo}\sum_{j \in \ztbtt}- \alpha_{i,j} \log  \alpha_{i,j}=\sum_{i\in \otboo}\alpha_{i} \left(\sum_{j \in \ztbtt} -\frac{\alpha_{i,j}}{\alpha_{i}} \log  \frac{\alpha_{i,j}}{\alpha_{i}}-\log \alpha_{i} \right) \leq \log (k+1)\sum_{i\in \otboo} \alpha_{i}  -\sum_{i\in \otboo}\alpha_{i} \log \alpha_{i}$. Now consider the term $-\sum_{i\in \otboo}\alpha_{i} \log \alpha_{i}$ and note that $-\sum_{i\in \otboo}\alpha_{i} \log \alpha_{i}=(\sum_{i\in \otboo}\alpha_{i})\left(-\sum_{i\in \otboo}\frac{\alpha_{i}}{\sum_{i\in \otboo}\alpha_{i}} \log \frac{\alpha_{i}}{\sum_{i\in \otboo}\alpha_{i}}-\log \sum_{i\in \otboo}\alpha_{i} \right)\leq (1+\log \boo)\sum_{i\in \otboo}\alpha_{i}$. The fifth inequality in \Cref{eq:derive} follows by combining the previous two derivations together. The final inequality follows because $\boo \times \btt \leq \logparam$.
	
	{\bf Condition 6:} This condition follows immediately from Line 7 of the algorithm.
\end{proof}

\renewcommand{\bSp}{\textbf{S}}
\renewcommand{\bS}{\textbf{A}}
\renewcommand{\boo}{\ell}

\begin{proof}[Proof of \Cref{lem:stepone}]
	In the following we provide the proof for the claims in the lemma.
	
	{\bf Condition 1:} Note $\bHone \subseteq \bH \cup [\ell+1,\ell+\newk]$, where $[\ell+1,\ell+\newk]$ are the indices corresponding to the new rows created by the procedure $\algcreate$ (\Cref{alg:create}). Consider any $i \in \bHone$, then one the following two cases hold,
	\begin{enumerate}
		\item If $i \in \bH$, then by the first condition of \Cref{lem:create} we have $(\bSone \onevec)_{i}=(\ma \onevec)_{i}=\sum_{j\in \ztbtt} \ma_{i,j}=\sum_{j\in \ztbtt} \floor{\bSp_{i,j}} \in \Z$. 
		\item Else $i \in [\ell+1,\ell+\newk]$ and in this case we have $\sum_{i\in [1,\boo]}\ma_{i,j}=\sum_{i\in \bH}\ma_{i,j}+\sum_{i\in \bL}\ma_{i,j}=\sum_{i\in \bH}\floor{\bSp_{i,j}}+\floor{\sum_{i\in \bL}\bSp_{i,j}}\in \Z$. The second equality in the previous derivation follows from Line 7 and 8 of the algorithm. The previous derivation combined with third condition of \Cref{lem:create} we get,  $(\bSone \onevec)_{i}=\phi_{j}-\sum_{i\in [1,\boo]}\ma_{i,j}\in \Z$.
	\end{enumerate}
	$(\bSone \onevec)_{i} \in \Z$ in both the cases and the condition 1 follows.
	
	{\bf Condition 2:} This condition follows immediately from the fourth condition of \Cref{lem:create}.
	
	{\bf Condition 3:} Let $\alpha_{i}=\sum_{j\in \ztbtt}\bSp_{i,j} - \sum_{j\in \ztbtt}\ma_{i,j}$ for all $i\in [1,\boo]$. First we upper bound the term $\sum_{i\in \bH}\alpha_{i}$. Consider $\sum_{i\in \bH}\alpha_{i} \leq \sum_{i\in \bH} \sum_{j\in \ztbtt}\bSp_{i,j} \leq \tsqni$. The last inequality follows because of the constraint $\sum_{i\in [1,\boo]}\pvec_{i} \sum_{j\in \ztbtt}\bSp_{i,j} \leq 1$ ($\bSp \in \bZfrac$) and $\pvec_{i} >\tsqn$ for all $i \in \bH$. 
	
	We now upper bound the term $\sum_{i\in \bL}\alpha_{i}$. Consider $\sum_{i\in \bL}\alpha_{i}=\sum_{i\in \bL}\left(\sum_{j\in \ztbtt}\bSp_{i,j} - \sum_{j\in \ztbtt}\ma_{i,j}\right)=\sum_{j\in \ztbtt}\left(\sum_{i\in \bL}\bSp_{i,j} - \sum_{i\in \bL}\ma_{i,j}\right)$. Further $\sum_{i\in \bL}\ma_{i,j}=\floor{\sum_{i\in \bL}\bSp_{i,j}}$ for all $j\in \ztbtt$ (Line 8 of the algorithm) and we get $\sum_{i\in \bL}\alpha_{i} \leq k+1$. 
	
	Therefore $\sum_{i\in [\boo]}\alpha_{i}=\sum_{i\in \bH}\alpha_{i}+\sum_{i\in \bL}\alpha_{i} \leq \tsqni+k+1$ and combined with fifth condition \Cref{lem:create} we have, 
	$$\bg(\bSone) \geq \expo{-O\left(\left(\tsqni+k\right)\log \logparam \right)}\bg(\bSp) ~.$$
\end{proof}

\begin{proof}[Proof of \Cref{lem:steptwo}] 
	In the following we provide proof for all the conditions in the lemma.
	
	{\bf Condition 1:} For all $i \in [1,\boo+\newk]$, one of the following two conditions hold,
	\begin{enumerate}
		\item If $i \in \bHone$, then by the first condition of \Cref{lem:create} we have $(\bStwo \onevec)_{i}=(\bAone \onevec)_{i}=(\bSone \onevec)_{i} \in \Z$. The last expression follows from first condition of \Cref{lem:stepone}.
		\item Else $i \in \bLone$, then again by the first condition of \Cref{lem:create} we have $(\bStwo \onevec)_{i}=(\bAone \onevec)_{i}=\floor{(\bSone \onevec)_{i}} \in \Z$. The last equality follows from Line 15 of the algorithm.
	\end{enumerate}
	For all $i \in [1,\boo+\newk]$, we have $(\bStwo \onevec)_{i} \in \Z$ and therefore condition 1 holds.
	
	{\bf Condition 2:} This condition follows immediately from the second condition of \Cref{lem:create}.
	
	{\bf Condition 3:} This condition follows immediately from the fourth condition of \Cref{lem:create}.
	
	{\bf Condition 4:} Consider the term $\sum_{i \in [\boo+\newk+1,\boo+2\newk]} (\bStwo \onevec)_{i}$,
	\begin{equation}
	\begin{split}
	\sum_{i \in [\boo+\newk+1,\boo+2\newk]} (\bStwo \onevec)_{i}&=\sum_{i \in [1,\boo+2\newk]} (\bStwo \onevec)_{i} - \sum_{i \in [1,\boo+\newk]}(\bStwo \onevec)_{i}\\
	&=\sum_{j \in \ztbtt} \phi_{j} - \sum_{i \in [1,\boo+\newk]}(\bAone \onevec)_{i}\\
	&=\sum_{j \in \ztbtt} \phi_{j} - \left(\sum_{i \in \bHone}(\bAone \onevec)_{i}+\sum_{i \in \bLone}(\bAone \onevec)_{i} \right)\\
	&=\sum_{j \in \ztbtt} \phi_{j} - \left(\sum_{i \in \bHone}(\bSone \onevec)_{i}+\sum_{i \in \bLone} \floor{(\bSone \onevec)_{i}} \right) \in \Z
	\end{split}
	\end{equation}
	In the first equality we add and substract $\sum_{i \in [1,\boo+\newk]}(\bStwo \onevec)_{i}$ term. The first term in the second equality follows because $\sum_{i \in [1,\boo+2\newk]} (\bStwo \onevec)_{i}=\sum_{j \in \ztbtt} \sum_{i \in [1,\boo+2\newk]} \bStwo_{i,j}=\sum_{j \in \ztbtt} \phi_{j}$ and the last equality follows because $\bStwo \in \bZtwo$ (Condition 3). The second term in the second equality follows by the first condition of \Cref{lem:create}. In the third equality we divided the summation terms over $\bHone$ and $\bLone$. In the fourth equality we used Line 14 of the algorithm and further for any $i \in \bLone$ Line 15 implies $(\bAone \onevec)_{i}=\sum_{j\in \ztbtt}\bSone_{ij} \frac{\floor{(\bSone \onevec)_{i}}}{(\bSone \onevec)_{i}}=\floor{(\bSone \onevec)_{i}}$. Finally by first condition of \Cref{lem:stepone} we have $(\bSone \onevec)_{i} \in \Z$ for all $i \in \bHone$ and $\phi_{j} \in \Z$ for all $j \in \ztbtt$. Therefore, $\sum_{i \in [\boo+\newk+1,\boo+2\newk]} (\bStwo \onevec)_{i} \in \Z$ and the condition 4 holds.
	
	{\bf Condition 5:} For any $j \in \ztbtt$ we have,
	\begin{equation}
	\begin{split}
	\pvectwo_{\boo+\newk+1+j}&=\frac{\sum_{i \in [1,\boo+\newk]}(\bSone_{ij}-\bAone_{ij})\pvecone_{i}}{\sum_{i \in [1,\boo+\newk]}(\bSone_{ij}-\bAone_{ij})}
	=\frac{\sum_{i \in \bLone}(\bSone_{ij}-\bAone_{ij})\pvecone_{i}}{\sum_{i \in \bLone}(\bSone_{ij}-\bAone_{ij})}\\
	& \leq \tsqn \frac{\sum_{i \in \bLone}(\bSone_{ij}-\bAone_{ij})}{\sum_{i \in \bLone}(\bSone_{ij}-\bAone_{ij})} \leq \tsqn.
	\end{split}
	\end{equation}
	The first equality follows from the sixth condition of \Cref{lem:create}. The second equality follows because $\bSone_{i,j}=\bAone_{i,j}$ for all $i \in \bHone$ and $j \in \ztbtt$ (Line 14). The third inequality follows because $\bSone_{i,j} \geq \bAone_{i,j}$ for all $i \in \bLone$ and $j \in \ztbtt$ (Line 15) and further $\pvecone_{i} \leq \tsqn$ for all $i \in \bLone$ (Line 12).
	
	{\bf Condition 6:} For any $i \in [1,\boo+\newk]$, let $\alpha_{i}=\sum_{j \in \ztbtt}\bSone_{i,j} -\sum_{j \in \ztbtt}\bAone$. Note $\alpha_{i}=0$ for all $i\in \bHone$ (Line 14) and $\alpha_{i}= (\bSone \onevec)_{i}-\floor{(\bSone \onevec)_{i}} \leq 1$ for all $i \in \bLone$ (Line 15). Therefore $\sum_{i \in [1,\boo+\newk]} \alpha_{i}\leq |\bLone| \leq \boo+\newk$ and further combined with the fifth condition of \Cref{lem:create} we have $\bg(\bStwo) \geq \expo{-O\left((\boo+\btt)\log \logparam \right)} \bg(\bSone)$. Note by the third condition of \Cref{lem:stepone} we have $\bg(\bSone) \geq\expo{-O\left(\left(\tsqni+\btt \right)\log \logparam \right)} \bg(\bSp)$. Combining the previous two inequalities we get $\bg(\bStwo) \geq \expo{-O\left((\boo+\btt+\tsqni)\log \logparam \right)} \bg(\bSp)$ and condition 6 holds.
\end{proof}

\begin{proof}[Proof of \Cref{thm:stepthree}] In the following we provide proof for the two conditions of the theorem.
	
	{\bf Condition 1:} Here we provide the proof for the condition $\bSext \in \bZext$.
	\begin{enumerate}
		\item For all $i \in [1,\boo+2\newk]$, consider $(\bSext \onevec)_{i}$. If $i \in [1,\boo+\newk]$, then $(\bSext \onevec)_{i}=(\bStwo \onevec)_{i} \in \Z$. The first equality follows by line 22 of the algorithm and the last expression follows by first condition of \Cref{lem:steptwo}. Else $i \in [\boo+\newk+1,\boo+2\newk]$, let $j$ be such that $i=\boo+\newk+1+j$, then $(\bSext \onevec)_{i}= \sum_{j' \in \ztbtt}\bSext_{\boo+\newk+1+j,j'}=\sum_{j' \in \ztbtt}\left(\floor{\bStwo_{\boo+\newk+1+j,j'}}+z_{j,j'}\right)=\floor{\bStwo_{\boo+\newk+1+j,j}}+\sum_{j'\in \ztbtt}z_{j,j'} \in \Z$. The second equality follows by line 23 of the algorithm. The third equality follows from the second condition of \Cref{lem:steptwo}. Finally by the first condition of \Cref{lem:structured} we have $\sum_{j'\in \ztbtt}z_{j,j'} \in \{0,1\}$ for all $j \in \ztbtt$ and therefore $(\bSext \onevec)_{i} \in \Z$ for any $i \in [\boo+\newk+1,\boo+2\newk]$. 
		
		Combining the analysis of cases $i \in [1,\boo+\newk]$ and $i \in [\boo+\newk+1,\boo+2\newk]$ the condition 1 holds.
		
		\item For all $j \in \ztbtt$,
		\begin{equation}\label{eq:dddone}
		\begin{split}
		\sum_{i\in [1,\boo+2\newk]} \bSext_{i,j}&=\sum_{i\in [1,\boo+\newk]} \bSext_{i,j}+\sum_{i\in [\boo+\newk+1,\boo+2\newk]} \bSext_{i,j}\\
		&=\sum_{i\in [1,\boo+\newk]} \bStwo_{i,j}+\sum_{j' \in \ztbtt} \left(\floor{\bStwo_{\boo+\newk+1+j',j}}+z_{j',j}\right)~.
		\end{split}
		\end{equation}
		The second equality follows because $\bSext_{i,j}=\bStwo_{i,j}$ for all $i \in [1,\boo+\newk]$ (Line 22) and $\bSext_{i,j}=\floor{\bStwo_{\boo+\newk+1+j',j}}+z_{j',j}$ for all $i \in [\boo+\newk+1,\boo+2\newk]$ (Line 23).
		We next simplify the second term in the above expression.
		\begin{equation}\label{eq:dddtwo}
		\begin{split}
		\sum_{j' \in \ztbtt} \left(\floor{\bStwo_{\boo+\newk+1+j',j}}+z_{j',j}\right)&=\floor{\bStwo_{\boo+\newk+1+j,j}}+\sum_{j' \in \ztbtt}z_{j',j}
		=\floor{\bStwo_{\boo+\newk+1+j,j}}+x_{j}\\
		&=\bStwo_{\boo+\newk+1+j,j}=\sum_{i\in [\boo+\newk+1,\boo+2\newk]} \bStwo_{i,j}~.
		\end{split}
		\end{equation}
		In the first and final equality we used the second condition of \Cref{lem:steptwo} (Diagonal Structure). In the second equality we used the first condition of \Cref{lem:structured}. In the third equality we used the definition of $x_{j}$ (Line 19).
		Combining equations \ref{eq:dddone} and \ref{eq:dddtwo} we get,
		$$\sum_{i\in [1,\boo+2\newk]} \bSext_{i,j}=\sum_{i\in [1,\boo+2\newk]} \bStwo_{i,j} =\phi_{j}$$
		In the last inequality we used $\bStwo \in \bZtwo$.
		
		\item Let $\pvecext_{i}$ for all $i \in [1,\boo+2\newk]$ be the $i$'th element of $\bRext$. Consider $\sum_{i \in [1,\boo+2\newk]}\pvecext_{i}(\bSext \onevec)_{i}$, we have, 
		\begin{equation}\label{eq:derione}
		\begin{split}
		\sum_{i \in [1,\boo+2\newk]}&\pvecext_{i}(\bSext \onevec)_{i}=\sum_{i \in [1,\boo+2\newk]}\frac{\pvectwo_{i}}{1+\tsqn}(\bSext \onevec)_{i}\\
		&=\frac{1}{1+\tsqn} \sum_{i \in [1,\boo+\newk+1]}\pvectwo_{i}(\bStwo \onevec)_{i}+\frac{1}{1+\tsqn}\sum_{i \in [\boo+\newk+1,\boo+2\newk]}\pvectwo_{i}(\bSext \onevec)_{i}.
		\end{split}
		\end{equation}
		The first equality follows from Line 24 of the algorithm. In the second equality we divided the summation into two parts and used $\bSext_{i,j}=\bStwo_{i,j}$ for all $i \in [1,\boo+\newk+1]$ and $j \in \ztbtt$ (Line 22) for the first part. We now simplify the second part of the above expression.
		
		\begin{equation}\label{eq:deritwo}
		\begin{split}
		\sum_{i \in [\boo+\newk+1,\boo+2\newk]}\pvectwo_{i}(\bSext \onevec)_{i}
		&=\sum_{j\in \ztbtt}\pvectwo_{\boo+\newk+1+j}\sum_{j'\in \ztbtt}\left(\floor{\bStwo_{\boo+\newk+1+j,j'}}+z_{j,j'}\right)\\
		&=\sum_{j\in \ztbtt}w_{j}\left(\bStwo_{\boo+\newk+1+j,j}-x_{j}\right) +\sum_{j\in \ztbtt}w_{j}\sum_{j'\in \ztbtt}z_{j,j'}\\
		&\leq \sum_{j\in \ztbtt}w_{j}\left(\bStwo_{\boo+\newk+1+j,j}-x_{j}\right)+ \sum_{j\in \ztbtt} w_{j}x_{j}+\max_{j \in \ztbtt}w_{j}\\
		&=\sum_{i \in [\boo+\newk+1,\boo+2\newk]}\pvectwo_{i}(\bStwo \onevec)_{i}+\tsqn~.
		\end{split}
		\end{equation}
		In the first equality we expanded the $(\bSext \onevec)_{i}$ term. Further we used $\bSext_{\boo+\newk+1+j,j'}=\floor{\bStwo_{\boo+\newk+1+j,j'}}+z_{j,j'}$ for all $j,j' \in \ztbtt$ (Line 23). In the second equality we used the second condition of \Cref{lem:steptwo} (Diagonal Structure) and further combined it with definitions of $w_{j}$ and $x_{j}$ from Line 19 of the algorithm. The third inequality follows from second condition of \Cref{lem:structured}. In the final inequality we used $\max_{j \in \ztbtt}w_{j} \leq \tsqn$ that follows from the definition of $w_{j}$ and fifth condition of \Cref{lem:steptwo}. Further we combined it with $\bStwo_{\boo+\newk+1+j,j}=(\bStwo \onevec)_{i}$ that follows from the second condition of \Cref{lem:steptwo}.

		Combining equations \ref{eq:derione} and \ref{eq:deritwo} we have,
		$$\sum_{i \in [1,\boo+2\newk]}\pvecext_{i}(\bSext \onevec)_{i}\leq \frac{1}{1+\tsqn} \left(\sum_{i \in [1,\boo+2\newk]}\pvectwo_{i}(\bStwo \onevec)_{i}+\tsqn \right)\leq 1~.$$
		In the final inequality we used $\bStwo \in \bZtwo$ and therefore $\sum_{i \in [1,\boo+2\newk]}\pvectwo_{i}(\bStwo \onevec)_{i} \leq 1$.
	\end{enumerate}
	The condition 1 holds by combining the analysis of all the above three cases.
	
	{\bf Condition 2:} Recall the definition of $\bg(\bSext)$,
	\begin{align*}
	\bg(\bSext)=\prod_{i \in [1,\boo+2\newk]} \left({\pvecext_i}^{(\bSext \mvec)_i} \frac{\expo{(\bSext \onevec)_i \log (\bSext \onevec)_i}}{\prod_{j \in \ztbtt}\expo{\bSext_{ij}\log \bSext_{ij}}}\right)
	\end{align*}
	In the above expression consider the probability term,
	\begin{equation}\label{eq:done}
	\begin{split}
	\prod_{i \in [1,\boo+2\newk]} & {\pvecext_i}^{(\bSext \mvec)_i}=\prod_{i \in [1,\boo+2\newk]} \left(\frac{\pvectwo_i}{1+\tsqn} \right)^{(\bSext \mvec)_i}\\
	&\geq  \expo{-O(\tsqn n)} \left(\prod_{i \in [1,\boo+\newk]} {\pvectwo_i}^{(\bSext \mvec)_i} \right)\left(\prod_{i \in [\boo+\newk+1,\boo+2\newk]} {\pvectwo_i}^{(\bSext \mvec)_i} \right)\\
	&=\expo{-O(\tsqn n)}\left(\prod_{i \in [1,\boo+\newk]} {\pvectwo_i}^{(\bStwo \mvec)_i} \right)\left(\prod_{i \in [\boo+\newk+1,\boo+2\newk]} {\pvectwo_i}^{(\bSext \mvec)_i} \right)~.
	\end{split}
	\end{equation}
	In the first equality we used line 24 of the algorithm. In the second inequality we used $\sum_{i \in [1,\boo+2\newk]}(\bSext \mvec)_i=n$ that further implies $\left(1+\tsqn\right)^{-\sum_{i \in [1,\boo+2\newk]}(\bSext \mvec)_i} \geq \expo{-O(\tsqn n)}$. In the third equality we used $\bSext_{i,j}=\bStwo_{i,j}$ for all $i \in [1,\boo+\newk]$ and $j \in \ztbtt$ (Line 22).
	We now analyze the second product term in the final expression above,
	\begin{equation}\label{eq:dtwo}
	\begin{split}
	\prod_{i \in [\boo+\newk+1,\boo+2\newk]} & {\pvectwo_i}^{(\bSext \mvec)_i}=\prod_{j\in \ztbtt} {\pvectwo_{\boo+\newk+1+j}}^{\sum_{j'\in \ztbtt}\bSext_{\boo+\newk+1+j,j'} \mvec_{j'}}\\
	&=\prod_{j\in \ztbtt} {\pvectwo_{\boo+\newk+1+j}}^{\sum_{j'\in \ztbtt}\left(\floor{\bStwo_{\boo+\newk+1+j,j'}}+z_{j,j'} \right) \mvec_{j'}}\\
	&=\left( \prod_{j\in \ztbtt}{\pvectwo_{\boo+\newk+1+j}}^{\floor{\bStwo_{\boo+\newk+1+j,j}}} \right) \left(\prod_{j\in \ztbtt} {\pvectwo_{\boo+\newk+1+j}}^{\sum_{j'\in \ztbtt}z_{j,j'} \mvec_{j'}} \right).
	\end{split}
	\end{equation}
	The second equality follows from line 23 of the algorithm. The third equality follows from the second condition of \Cref{lem:steptwo} (Diagonal Structure).
	
	Now consider the second product term in the above expression.
	\begin{equation}\label{eq:dthree}
	\begin{split}
	\prod_{j\in \ztbtt} {\pvectwo_{\boo+\newk+1+j}}^{\sum_{j'\in \ztbtt}z_{j,j'} \mvec_{j'}}&=\prod_{j\in \ztbtt} w_{j}^{\sum_{j'\in \ztbtt}z_{j,j'} \mvec_{j'}} 
	\geq \prod_{j\in \ztbtt} w_{j}^{x_{j} \mvec_{j}}~.
	\end{split}
	\end{equation}
	In the first equality we used the definition of $w_{j}$ (Line 19). The second inequality follows from the third condition of \Cref{lem:structured}.
	
	Combining equations \ref{eq:dtwo}, \ref{eq:dthree} and further using $x_{j}=\bStwo_{\boo+\newk+1+j,j}-\floor{\bStwo_{\boo+\newk+1+j,j}}$ for all $j \in \ztbtt$ (Line 19) we have,
	\begin{equation}\label{eq:dfour}
	\prod_{i \in [\boo+\newk+1,\boo+2\newk]} {\pvectwo_i}^{(\bSext \mvec)_i}\geq \prod_{j\in \ztbtt} {\pvectwo_{\boo+\newk+1+j}}^{\bStwo_{\boo+\newk+1+j,j} \mvec_{j}}=\prod_{i \in [\boo+\newk+1,\boo+2\newk]} {\pvectwo_i}^{(\bStwo \mvec)_i}~.
	\end{equation}
	In the final inequality we used the second condition of \Cref{lem:steptwo} (Diagonal Structure).
	
	Combining equations \ref{eq:done} and \ref{eq:dfour} we have,
	$$\prod_{i \in [1,\boo+2\newk]} {\pvecext_i}^{(\bSext \mvec)_i} \geq \expo{-O(\tsqn n)} \prod_{i \in [1,\boo+2\newk]} {\pvectwo_i}^{(\bStwo \mvec)_i}$$
	
	Using the above expression we have,
	\begin{equation}\label{eq:eone}
	\begin{split}
	\frac{ \bg(\bSext)}{\bg(\bStwo)} &\geq \expo{-O(\tsqn n)} \prod_{i \in [1,\boo+2\newk]} \left( \frac{\expo{(\bSext \onevec)_i \log (\bSext \onevec)_i-(\bStwo \onevec)_i \log (\bStwo \onevec)_i}}{\prod_{j' \in \ztbtt}\expo{\bSext_{i,j'}\log \bSext_{i,j'}-\bStwo_{i,j'}\log \bStwo_{i,j'}}}\right)\\
	&=\expo{-O(\tsqn n)} \prod_{i \in [\boo+\newk+1,\boo+2\newk]} \left( \frac{\expo{(\bSext \onevec)_i \log (\bSext \onevec)_i-(\bStwo \onevec)_i \log (\bStwo \onevec)_i}}{\prod_{j' \in \ztbtt}\expo{\bSext_{i,j'}\log \bSext_{i,j'}-\bStwo_{i,j'}\log \bStwo_{i,j'}}}\right)\\
	&=\expo{-O(\tsqn n)} \prod_{i \in [\boo+\newk+1,\boo+2\newk]} \expo{(\bSext \onevec)_i \log (\bSext \onevec)_i-\sum_{j' \in \ztbtt}\bSext_{i,j'}\log \bSext_{i,j'}}~.
	\end{split}
	\end{equation}
	In the second equality we used $\bSext_{i,j}=\bStwo_{i,j}$ for all $i \in [1,\boo+\newk]$ and $j \in \ztbtt$ (Line 22). The third inequality follows by the second condition of \Cref{lem:steptwo} (Diagonal Structure). In the remainder of the proof we lower bound the term in the final expression.
	
	For each $i \in [\boo+\newk+1,\boo+2\newk]$ let $j \in \ztbtt$ be such that $i=\boo+\newk+1+j$, then $(\bSext \onevec)_i =\sum_{j' \in \ztbtt} (\floor{\bStwo_{\boo+\newk+1+j,j'}}+z_{j,j'}) = \floor{\bStwo_{\boo+\newk+1+j,j}} + \sum_{j' \in \ztbtt}z_{j,j'}$. The first equality follows from line 23 of the algorithm. The second equality follows by the second condition of \Cref{lem:steptwo} (Diagonal Structure). Using first condition of \Cref{lem:structured}, one of the following two cases hold,
	\begin{enumerate}
		\item If $\sum_{j' \in \ztbtt}z_{j,j'}=0$, then $z_{j,j'}=0$ for all $j' \in \ztbtt$. Using second condition of \Cref{lem:steptwo} (Diagonal Structure), we have $\bSext_{\boo+\newk+1+j,j'}=\floor{\bStwo_{\boo+\newk+1+j,j'}} + z_{j,j'}=0$ for all $j' \in \ztbtt$ and $j' \neq j$. Further note, $(\bSext \onevec)_i=\floor{\bStwo_{\boo+\newk+1+j,j}} + \sum_{j' \in \ztbtt}z_{j,j'}=\bSext_{\boo+\newk+1+j,j}$. Combining previous two equalities we have,
		$(\bSext \onevec)_i \log (\bSext \onevec)_i-\sum_{j' \in \ztbtt}\bSext_{i,j'}\log \bSext_{i,j'}=0$. Therefore,
		\begin{equation}\label{eq:etwo}
		\expo{(\bSext \onevec)_i \log (\bSext \onevec)_i-\sum_{j' \in \ztbtt}\bSext_{i,j'}\log \bSext_{i,j'}} \geq 1~.
		\end{equation}
		
		\item If $\sum_{j' \in \ztbtt}z_{j,j'}=1$, then $z_{j,j'} \in [0,1]_{\R}$ for all $j' \in \ztbtt$. Using second condition of \Cref{lem:steptwo} (Diagonal Structure), we have $\bSext_{i,j'}=\bSext_{\boo+\newk+1+j,j'}=\floor{\bStwo_{\boo+\newk+1+j,j'}} + z_{j,j'}=z_{j,j'}$ for all $j' \in \ztbtt$ and $j' \neq j$. Therefore, $\sum_{j' \in \ztbtt}\bSext_{i,j'}\log \bSext_{i,j'}=(\floor{\bStwo_{\boo+\newk+1+j,j}}+z_{j,j})\log (\floor{\bStwo_{\boo+\newk+1+j,j}}+z_{j,j})+\sum_{j' \neq j}z_{j,j'} \log z_{j,j'} \leq (\floor{\bStwo_{\boo+\newk+1+j,j}}+z_{j,j})\log (\floor{\bStwo_{\boo+\newk+1+j,j}}+z_{j,j})$. The final inequality follows because $z_{j,j'} \in [0,1]_{\R}$ and $z_{j,j'} \log z_{j,j'} \leq 0$ for all $j' \in \ztbtt$.
		
		Further note, $(\bSext \onevec)_i=\floor{\bStwo_{\boo+\newk+1+j,j}} + \sum_{j' \in \ztbtt}z_{j,j'}=\floor{\bStwo_{\boo+\newk+1+j,j}}+1$. Combining previous two inequalities we have,
		$(\bSext \onevec)_i \log (\bSext \onevec)_i-\sum_{j' \in \ztbtt}\bSext_{i,j'}\log \bSext_{i,j'} \geq (\floor{\bStwo_{\boo+\newk+1+j,j}}+1) \log (\floor{\bStwo_{\boo+\newk+1+j,j}}+1)-(\floor{\bStwo_{\boo+\newk+1+j,j}}+z_{j,j})\log(\floor{\bStwo_{\boo+\newk+1+j,j}}+z_{j,j}) \geq 0$. The last inequality follows because of the following: If $\floor{\bStwo_{\boo+\newk+1+j,j}}=0$, then the inequality follows because $z_{j,j} \in [0,1]_{\R}$ and $z_{j,j}\log z_{j,j} \leq 0$. Else $\floor{\bStwo_{\boo+\newk+1+j,j}} \geq 1$, in this case we use the fact that $x \log x$ is a monotonically increasing for $x \geq 1$. 
		
		Therefore 
		\begin{equation}\label{eq:ethree}
		\expo{(\bSext \onevec)_i \log (\bSext \onevec)_i-\sum_{j' \in \ztbtt}\bSext_{i,j'}\log \bSext_{i,j'}} \geq 1~.
		\end{equation}
	\end{enumerate}
	
	Combining equations \ref{eq:etwo} and \ref{eq:ethree}, for all $i \in [\boo+\newk+1,\boo+2\newk]$ we have,
	$$\expo{(\bSext \onevec)_i \log (\bSext \onevec)_i-\sum_{j' \in \ztbtt}\bSext_{i,j'}\log \bSext_{i,j'}} \geq 1~.$$
	Substituting previous inequality in \Cref{eq:eone} we get,
	$$\frac{ \bg(\bSext)}{\bg(\bStwo)} \geq \expo{-O(\tsqn n)}~.$$
	Further the condition 2 of the theorem follows by combining the above inequality with the sixth condition of \Cref{lem:steptwo}.
\end{proof}

\section{Acknowledgments}

We thank Jayadev Acharya and Yanjun Han for helpful conversations. We thank the anonymous reviewer for pointing out the alternative proof of the quality of scaled Sinkhorn and Bethe approximations on approximating the permanent of matrices with a bounded number of distinct columns (see \Cref{sec:related} and \Cref{sec:altproof}).

\bibliographystyle{alpha}
\bibliography{PML}

\newcommand{\etalchar}[1]{$^{#1}$}
\begin{thebibliography}{RWdRvSB99}

\bibitem[AA11]{AA11}
Scott Aaronson and Alex Arkhipov.
\newblock The computational complexity of linear optics.
\newblock In {\em Proceedings of the Forty-third Annual ACM Symposium on Theory
  of Computing}, STOC '11, pages 333--342, New York, NY, USA, 2011. ACM.

\bibitem[ADM{\etalchar{+}}10]{ADMOP10}
J.~Acharya, H.~Das, H.~Mohimani, A.~Orlitsky, and S.~Pan.
\newblock Exact calculation of pattern probabilities.
\newblock In {\em 2010 IEEE International Symposium on Information Theory},
  pages 1498--1502, June 2010.

\bibitem[ADOS16]{ADOS16}
Jayadev Acharya, Hirakendu Das, Alon Orlitsky, and Ananda~Theertha Suresh.
\newblock A unified maximum likelihood approach for optimal distribution
  property estimation.
\newblock {\em CoRR}, abs/1611.02960, 2016.

\bibitem[AOST14]{AOST14}
Jayadev Acharya, Alon Orlitsky, Ananda~Theertha Suresh, and Himanshu Tyagi.
\newblock The complexity of estimating rényi entropy.
\newblock In {\em Proceedings of the Twenty-Sixth Annual ACM-SIAM Symposium on
  Discrete Algorithms}, 2014.

\bibitem[AOST17]{AOST17}
Jayadev Acharya, Alon Orlitsky, Ananda~Theertha Suresh, and Himanshu Tyagi.
\newblock Estimating renyi entropy of discrete distributions.
\newblock {\em IEEE Trans. Inf. Theor.}, 63(1):38--56, January 2017.

\bibitem[AR18]{AR19}
Nima Anari and Alireza Rezaei.
\newblock A tight analysis of bethe approximation for permanent.
\newblock {\em CoRR}, abs/1811.02933, 2018.

\bibitem[AS04]{AS04}
Noga Alon and Joel~H Spencer.
\newblock {\em The probabilistic method}.
\newblock John Wiley \& Sons, 2004.

\bibitem[Bar96]{Bar96}
Alexander~I Barvinok.
\newblock Two algorithmic results for the traveling salesman problem.
\newblock {\em Mathematics of Operations Research}, 21(1):65--84, 1996.

\bibitem[Bar17]{Bar17}
Alexander Barvinok.
\newblock {\em Combinatorics and Complexity of Partition Functions}.
\newblock Springer Publishing Company, Incorporated, 1st edition, 2017.

\bibitem[BF93]{BF93}
John Bunge and Michael Fitzpatrick.
\newblock Estimating the number of species: a review.
\newblock {\em Journal of the American Statistical Association},
  88(421):364--373, 1993.

\bibitem[Bre73]{Bre73}
L.~M. Bregman.
\newblock Certain properties of nonnegative matrices and their permanents.
\newblock 1973.

\bibitem[BZLV16]{BZLV16}
Y.~{Bu}, S.~{Zou}, Y.~{Liang}, and V.~V. {Veeravalli}.
\newblock Estimation of kl divergence between large-alphabet distributions.
\newblock In {\em 2016 IEEE International Symposium on Information Theory
  (ISIT)}, pages 1118--1122, July 2016.

\bibitem[CCG{\etalchar{+}}12]{CCGLMCL12}
Robert~K Colwell, Anne Chao, Nicholas~J Gotelli, Shang-Yi Lin, Chang~Xuan Mao,
  Robin~L Chazdon, and John~T Longino.
\newblock Models and estimators linking individual-based and sample-based
  rarefaction, extrapolation and comparison of assemblages.
\newblock {\em Journal of plant ecology}, 5(1):3--21, 2012.

\bibitem[Cha84]{Chao84}
A~Chao.
\newblock Nonparametric estimation of the number of classes in a population.
  scandinavianjournal of statistics11, 265-270.
\newblock {\em Chao26511Scandinavian Journal of Statistics1984}, 1984.

\bibitem[CL92]{Chao92}
Anne Chao and Shen-Ming Lee.
\newblock Estimating the number of classes via sample coverage.
\newblock {\em Journal of the American statistical Association},
  87(417):210--217, 1992.

\bibitem[CSS19]{CSS19}
Moses Charikar, Kirankumar Shiragur, and Aaron Sidford.
\newblock Efficient profile maximum likelihood for universal symmetric property
  estimation.
\newblock In {\em Proceedings of the 51st Annual ACM SIGACT Symposium on Theory
  of Computing}, STOC 2019, pages 780--791, New York, NY, USA, 2019. ACM.

\bibitem[DS13]{DS13}
Timothy Daley and Andrew~D Smith.
\newblock Predicting the molecular complexity of sequencing libraries.
\newblock {\em Nature methods}, 10(4):325, 2013.

\bibitem[ET76]{ET76}
Bradley Efron and Ronald Thisted.
\newblock Estimating the number of unseen species: How many words did
  shakespeare know?
\newblock {\em Biometrika}, 63(3):435--447, 1976.

\bibitem[F{\"u}r05]{Fur05}
Johannes F{\"u}rnkranz.
\newblock Web mining.
\newblock In {\em Data mining and knowledge discovery handbook}, pages
  899--920. Springer, 2005.

\bibitem[GS14]{GS14}
Leonid {Gurvits} and Alex {Samorodnitsky}.
\newblock {Bounds on the permanent and some applications}.
\newblock {\em arXiv e-prints}, page arXiv:1408.0976, Aug 2014.

\bibitem[GS18]{GS16}
Daniel Grier and Luke Schaeffer.
\newblock New hardness results for the permanent using linear optics.
\newblock In {\em Proceedings of the 33rd Computational Complexity Conference},
  CCC '18, pages 19:1--19:29, Germany, 2018. Schloss Dagstuhl--Leibniz-Zentrum
  fuer Informatik.

\bibitem[GTPB07]{GTPB07}
Zhan Gao, Chi-hong Tseng, Zhiheng Pei, and Martin~J Blaser.
\newblock Molecular analysis of human forearm superficial skin bacterial biota.
\newblock {\em Proceedings of the National Academy of Sciences},
  104(8):2927--2932, 2007.

\bibitem[Gur05]{Gur05}
Leonid Gurvits.
\newblock On the complexity of mixed discriminants and related problems.
\newblock In {\em Proceedings of the 30th International Conference on
  Mathematical Foundations of Computer Science}, MFCS'05, pages 447--458,
  Berlin, Heidelberg, 2005. Springer-Verlag.

\bibitem[{Gur}11]{Gur11}
Leonid {Gurvits}.
\newblock {Unleashing the power of Schrijver's permanental inequality with the
  help of the Bethe Approximation}.
\newblock {\em arXiv e-prints}, page arXiv:1106.2844, Jun 2011.

\bibitem[HHRB01]{HHRB01}
Jennifer~B Hughes, Jessica~J Hellmann, Taylor~H Ricketts, and Brendan~JM
  Bohannan.
\newblock Counting the uncountable: statistical approaches to estimating
  microbial diversity.
\newblock {\em Appl. Environ. Microbiol.}, 67(10):4399--4406, 2001.

\bibitem[HJM17]{HJM17}
Yanjun {Han}, Jiantao {Jiao}, and Rajarshi {Mukherjee}.
\newblock {On Estimation of \$L\_\{r\}\$-Norms in Gaussian White Noise Models}.
\newblock {\em arXiv e-prints}, page arXiv:1710.03863, Oct 2017.

\bibitem[HJW16]{HJW16}
Yanjun Han, Jiantao Jiao, and Tsachy Weissman.
\newblock Minimax estimation of {KL} divergence between discrete distributions.
\newblock {\em CoRR}, abs/1605.09124, 2016.

\bibitem[HJW18]{HJW18}
Yanjun Han, Jiantao Jiao, and Tsachy Weissman.
\newblock Local moment matching: A unified methodology for symmetric functional
  estimation and distribution estimation under wasserstein distance.
\newblock {\em arXiv preprint arXiv:1802.08405}, 2018.

\bibitem[HJWW17]{HJW17}
Yanjun {Han}, Jiantao {Jiao}, Tsachy {Weissman}, and Yihong {Wu}.
\newblock {Optimal rates of entropy estimation over Lipschitz balls}.
\newblock {\em arXiv e-prints}, page arXiv:1711.02141, Nov 2017.

\bibitem[HO19]{HO19}
Yi~{Hao} and Alon {Orlitsky}.
\newblock {The Broad Optimality of Profile Maximum Likelihood}.
\newblock {\em arXiv e-prints}, page arXiv:1906.03794, Jun 2019.

\bibitem[JHW16]{JHW16}
J.~Jiao, Y.~Han, and T.~Weissman.
\newblock Minimax estimation of the l1 distance.
\newblock In {\em 2016 IEEE International Symposium on Information Theory
  (ISIT)}, pages 750--754, July 2016.

\bibitem[JSV04]{JSV04}
Mark Jerrum, Alistair Sinclair, and Eric Vigoda.
\newblock A polynomial-time approximation algorithm for the permanent of a
  matrix with nonnegative entries.
\newblock {\em J. ACM}, 51(4):671--697, July 2004.

\bibitem[JVHW15]{JVHW15}
J.~Jiao, K.~Venkat, Y.~Han, and T.~Weissman.
\newblock Minimax estimation of functionals of discrete distributions.
\newblock {\em IEEE Transactions on Information Theory}, 61(5):2835--2885, May
  2015.

\bibitem[KLR99]{KLR99}
Ian Kroes, Paul~W Lepp, and David~A Relman.
\newblock Bacterial diversity within the human subgingival crevice.
\newblock {\em Proceedings of the National Academy of Sciences},
  96(25):14547--14552, 1999.

\bibitem[LSW98]{LSW98}
Nathan Linial, Alex Samorodnitsky, and Avi Wigderson.
\newblock A deterministic strongly polynomial algorithm for matrix scaling and
  approximate permanents.
\newblock In {\em Proceedings of the Thirtieth Annual ACM Symposium on Theory
  of Computing}, STOC '98, pages 644--652, New York, NY, USA, 1998. ACM.

\bibitem[OSS{\etalchar{+}}04]{OSSVZ04}
A.~Orlitsky, S.~Sajama, N.~P. Santhanam, K.~Viswanathan, and Junan Zhang.
\newblock Algorithms for modeling distributions over large alphabets.
\newblock In {\em International Symposium on Information Theory, 2004. ISIT
  2004. Proceedings.}, pages 304--304, 2004.

\bibitem[OSW16]{OSW16}
Alon Orlitsky, Ananda~Theertha Suresh, and Yihong Wu.
\newblock Optimal prediction of the number of unseen species.
\newblock {\em Proceedings of the National Academy of Sciences},
  113(47):13283--13288, 2016.

\bibitem[OSZ03]{OSZ03}
A.~Orlitsky, N.~P. Santhanam, and J.~Zhang.
\newblock Always good turing: asymptotically optimal probability estimation.
\newblock In {\em 44th Annual IEEE Symposium on Foundations of Computer
  Science, 2003. Proceedings.}, pages 179--188, Oct 2003.

\bibitem[PBG{\etalchar{+}}01]{PBGELLSD01}
Bruce~J Paster, Susan~K Boches, Jamie~L Galvin, Rebecca~E Ericson, Carol~N Lau,
  Valerie~A Levanos, Ashish Sahasrabudhe, and Floyd~E Dewhirst.
\newblock Bacterial diversity in human subgingival plaque.
\newblock {\em Journal of bacteriology}, 183(12):3770--3783, 2001.

\bibitem[PGM{\etalchar{+}}01]{ASNRMAS01}
A.~{Porta}, S.~{Guzzetti}, N.~{Montano}, R.~{Furlan}, M.~{Pagani},
  A.~{Malliani}, and S.~{Cerutti}.
\newblock Entropy, entropy rate, and pattern classification as tools to typify
  complexity in short heart period variability series.
\newblock {\em IEEE Transactions on Biomedical Engineering}, 48(11):1282--1291,
  Nov 2001.

\bibitem[PJW17]{PJW17}
D.~S. {Pavlichin}, J.~{Jiao}, and T.~{Weissman}.
\newblock {Approximate Profile Maximum Likelihood}.
\newblock {\em ArXiv e-prints}, December 2017.

\bibitem[PW96]{PW96}
Nina~T. Plotkin and Abraham~J. Wyner.
\newblock {\em An Entropy Estimator Algorithm and Telecommunications
  Applications}, pages 351--363.
\newblock Springer Netherlands, Dordrecht, 1996.

\bibitem[Rad97]{Rad97}
Jaikumar Radhakrishnan.
\newblock An entropy proof of bregman's theorem.
\newblock {\em Journal of Combinatorial Theory, Series A}, 77(1):161 -- 164,
  1997.

\bibitem[RCS{\etalchar{+}}09]{RCSWTKRWC09}
Harlan~S Robins, Paulo~V Campregher, Santosh~K Srivastava, Abigail Wacher,
  Cameron~J Turtle, Orsalem Kahsai, Stanley~R Riddell, Edus~H Warren, and
  Christopher~S Carlson.
\newblock Comprehensive assessment of t-cell receptor $\beta$-chain diversity
  in $\alpha$$\beta$ t cells.
\newblock {\em Blood}, 114(19):4099--4107, 2009.

\bibitem[RRSS07]{RRSS07}
S.~Raskhodnikova, D.~Ron, A.~Shpilka, and A.~Smith.
\newblock Strong lower bounds for approximating distribution support size and
  the distinct elements problem.
\newblock In {\em 48th Annual IEEE Symposium on Foundations of Computer Science
  (FOCS'07)}, pages 559--569, Oct 2007.

\bibitem[RVZ17]{RVZ17}
Aditi Raghunathan, Gregory Valiant, and James Zou.
\newblock Estimating the unseen from multiple populations.
\newblock {\em CoRR}, abs/1707.03854, 2017.

\bibitem[RWdRvSB99]{RWDB99}
Fred Rieke, Davd Warland, Rob de~Ruyter~van Steveninck, and William Bialek.
\newblock {\em Spikes: Exploring the Neural Code}.
\newblock MIT Press, Cambridge, MA, USA, 1999.

\bibitem[Sch78]{Sch78}
A~Schrijver.
\newblock A short proof of minc's conjecture.
\newblock {\em Journal of Combinatorial Theory, Series A}, 25(1):80 -- 83,
  1978.

\bibitem[Sch98]{Sch98}
Alexander Schrijver.
\newblock Counting 1-factors in regular bipartite graphs.
\newblock {\em Journal of Combinatorial Theory, Series B}, 72(1):122 -- 135,
  1998.

\bibitem[Spe82]{Minc78}
E.~Spence.
\newblock H. minc, permanents (encyclopedia of mathematics and its
  applications, vol. 6, addison-wesley advanced book programme, 1978), xviii
  205 pp., 21.50.
\newblock {\em Proceedings of the Edinburgh Mathematical Society},
  25(1):110–110, 1982.

\bibitem[TE87]{TE87}
Ronald Thisted and Bradley Efron.
\newblock Did shakespeare write a newly-discovered poem?
\newblock {\em Biometrika}, 74(3):445--455, 1987.

\bibitem[Val79]{Val79}
L.G. Valiant.
\newblock The complexity of computing the permanent.
\newblock {\em Theoretical Computer Science}, 8(2):189 -- 201, 1979.

\bibitem[VBB{\etalchar{+}}12]{VBBVP12}
Martin Vinck, Francesco~P. Battaglia, Vladimir~B. Balakirsky, A.~J.~Han Vinck,
  and Cyriel M.~A. Pennartz.
\newblock Estimation of the entropy based on its polynomial representation.
\newblock {\em Phys. Rev. E}, 85:051139, May 2012.

\bibitem[Von11]{Von11}
Pascal~O. Vontobel.
\newblock The bethe permanent of a non-negative matrix.
\newblock {\em CoRR}, abs/1107.4196, 2011.

\bibitem[Von12]{Von12}
Pascal~O. Vontobel.
\newblock The bethe approximation of the pattern maximum likelihood
  distribution.
\newblock pages 2012--2016, 07 2012.

\bibitem[{Von}13]{Von13}
P.~O. {Vontobel}.
\newblock The bethe permanent of a nonnegative matrix.
\newblock {\em IEEE Transactions on Information Theory}, 59(3):1866--1901,
  March 2013.

\bibitem[Von14]{Von14}
P.~O. Vontobel.
\newblock The bethe and sinkhorn approximations of the pattern maximum
  likelihood estimate and their connections to the valiant-valiant estimate.
\newblock In {\em 2014 Information Theory and Applications Workshop (ITA)},
  pages 1--10, Feb 2014.

\bibitem[VV11a]{VV11b}
G.~Valiant and P.~Valiant.
\newblock The power of linear estimators.
\newblock In {\em 2011 IEEE 52nd Annual Symposium on Foundations of Computer
  Science}, pages 403--412, Oct 2011.

\bibitem[VV11b]{VV11a}
Gregory Valiant and Paul Valiant.
\newblock Estimating the unseen: An n/log(n)-sample estimator for entropy and
  support size, shown optimal via new clts.
\newblock In {\em Proceedings of the Forty-third Annual ACM Symposium on Theory
  of Computing}, STOC '11, pages 685--694, New York, NY, USA, 2011. ACM.

\bibitem[WY15]{WY15}
Y.~{Wu} and P.~{Yang}.
\newblock {Chebyshev polynomials, moment matching, and optimal estimation of
  the unseen}.
\newblock {\em ArXiv e-prints}, April 2015.

\bibitem[WY16a]{WY16}
Y.~Wu and P.~Yang.
\newblock Minimax rates of entropy estimation on large alphabets via best
  polynomial approximation.
\newblock {\em IEEE Transactions on Information Theory}, 62(6):3702--3720, June
  2016.

\bibitem[WY16b]{WY16a}
Yihong {Wu} and Pengkun {Yang}.
\newblock {Sample complexity of the distinct elements problem}.
\newblock {\em arXiv e-prints}, page arXiv:1612.03375, Dec 2016.

\bibitem[YFW05]{YFW05}
J.~S. {Yedidia}, W.~T. {Freeman}, and Y.~{Weiss}.
\newblock Constructing free-energy approximations and generalized belief
  propagation algorithms.
\newblock {\em IEEE Transactions on Information Theory}, 51(7):2282--2312, July
  2005.

\bibitem[ZVV{\etalchar{+}}16]{ZVVKCSLSDM16}
James Zou, Gregory Valiant, Paul Valiant, Konrad Karczewski, Siu~On Chan,
  Kaitlin Samocha, Monkol Lek, Shamil Sunyaev, Mark Daly, and Daniel~G.
  MacArthur.
\newblock Quantifying unobserved protein-coding variants in human populations
  provides a roadmap for large-scale sequencing projects.
\newblock {\em Nature Communications}, 7:13293 EP --, 10 2016.

\end{thebibliography}
\appendix
\newcommand{\mB}{\textbf{B}}
\newcommand{\ml}{\textbf{L}}
\newcommand{\mr}{\textbf{R}}
\section{Alternative proof for the distinct column case.}\label{sec:altproof}
Here we provide an alternative and simpler proof for \Cref{thm:main} which was pointed to us by an anonymous reviewer. This alternative proof is derived using Corollary 3.4.5 in Barvinok's book~\cite{Bar17} (which is further derived using the Bregman-Minc inequality) and we formally state it below.
\begin{lemma}[Corollary 3.4.5 from \cite{Bar17}]\label{lem:barvinok}
	Suppose that $\bQ$ is a $N \times N$ doubly stochastic matrix that satisfies,
	$$\bQ_{i,j} \leq \frac{1}{b_i} \text{ for all }i \in [N],j \in [N]$$
	for some positive integers $b_1,\dots b_N$. Then,
	$$\perm(\bQ) \leq \prod_{i \in [N]} \frac{(b_i!)^{1/b_i}}{b_i}~.$$
\end{lemma}
Using the above result, we now prove \Cref{thm:main} and we restate it for convenience. 
\mainthm*
\begin{proof}[Alternative proof for \Cref{thm:main}]
The lower bound follows from \Cref{cor:ssinklb} and in the remainder we prove the upper bound. Let $\bQ$ be the maximizer of the scaled Sinkhorn objective, then it is a well know fact that $\bQ$ satisfies,
$$\bQ=\ml \ma \mr~,$$ 
where matrices $\ml$ and $\mr$ are the left and right non-negative diagonal matrices. Further by the symmetry of the objective, there exists an optimum solution $\bQ$ that has at most $k$ distinct columns and we work with such an optimum solution. As $\ml$ and $\mr$ are diagonal matrices, the following two inequalities are trivial,
\begin{equation}\label{eq:ssinkone}
\perm(\bQ)=\perm(\ml) \perm(\ma) \perm(\mr)~,
\end{equation}
\begin{equation}\label{eq:ssinktwo}
\ssink(\bQ)=\perm(\ml)~ \ssink(\ma) ~\perm(\mr),
\end{equation}
Further note that for all doubly stochastic matrices $\bQ$ we always have,
\begin{equation}\label{eq:ssinkthree}
\exp(-N) \leq \ssink(\bQ)~.
\end{equation}
Therefore combining \Cref{eq:ssinkone,eq:ssinktwo,eq:ssinkthree}, to prove the upper bound it is enough to show that,
$$\perm(\bQ) \leq \expo{O\left(k \log \frac{N}{k}\right)} \cdot \exp(-N)~.$$	
As matrix $\bQ$ has at most $k$ distinct columns, let the multiplicities of these distinct columns be $\phi_1,\ldots,\phi_k$. Note that if a column has multiplicity $\phi_i$, the maximal element in this column is at most $1/\phi_i$. Now by \Cref{lem:barvinok} (Corollary 3.4.5. in \cite{Bar17}), we have 
$$\perm(\bQ) \leq \prod_{i=1}^k \frac{\phi_i !}{\phi_i^{\phi_i}} \leq \expo{O\left(k \log\frac{N}{k} \right)} \cdot \expo{-N}~,$$
where the last inequality follows because the term $\prod_{i=1}^k \frac{\phi_i !}{\phi_i^{\phi_i}}$ is maximized when all $\phi_i$'s are equal and take value $N/k$. Therefore we conclude the proof.
\end{proof}

\end{document}